\documentclass[10pt]{article}

\usepackage{amsthm}
\usepackage{graphicx} 
\usepackage{array} 

\usepackage{amsmath, amssymb, amsfonts, verbatim}
\usepackage{hyphenat,epsfig,subcaption,multirow}
\usepackage{nicefrac}
\usepackage{paralist}
\usepackage{thm-restate}
\usepackage{mleftright}

\usepackage[font=small, labelfont=bf]{caption}

\usepackage[usenames,dvipsnames]{xcolor}
\usepackage[ruled]{algorithm2e}

\DeclareFontFamily{U}{mathx}{\hyphenchar\font45}
\DeclareFontShape{U}{mathx}{m}{n}{
      <5> <6> <7> <8> <9> <10>
      <10.95> <12> <14.4> <17.28> <20.74> <24.88>
      mathx10
      }{}
\DeclareSymbolFont{mathx}{U}{mathx}{m}{n}
\DeclareMathSymbol{\bigtimes}{1}{mathx}{"91}

\usepackage{tcolorbox}
\tcbuselibrary{skins,breakable}
\tcbset{enhanced jigsaw}

\usepackage[normalem]{ulem}
\usepackage[compact]{titlesec}

\definecolor{DarkRed}{rgb}{0.5,0.1,0.1}
\definecolor{DarkBlue}{rgb}{0.1,0.1,0.5}

\usepackage{nameref}
\definecolor{ForestGreen}{rgb}{0.1333,0.5451,0.1333}
\definecolor{Red}{rgb}{0.9,0,0}
\usepackage[linktocpage=true,
	pagebackref=true,colorlinks,
		urlcolor=black,
	linkcolor=DarkRed,citecolor=ForestGreen,
	bookmarks,bookmarksopen,bookmarksnumbered]
	{hyperref}
\usepackage[noabbrev,nameinlink]{cleveref}
\crefname{property}{property}{Property}
\creflabelformat{property}{(#1)#2#3}
\crefname{equation}{eq}{Eq}
\creflabelformat{equation}{(#1)#2#3}

\usepackage{bm}
\usepackage{url}
\usepackage{xspace}
\usepackage[mathscr]{euscript}

\usepackage{tikz}
\usetikzlibrary{arrows}
\usetikzlibrary{arrows.meta}
\usetikzlibrary{shapes}
\usetikzlibrary{backgrounds}
\usetikzlibrary{positioning}
\usetikzlibrary{decorations.markings}
\usetikzlibrary{patterns}
\usetikzlibrary{calc}
\usetikzlibrary{fit}

\usepackage{mdframed}

\usepackage[noend]{algpseudocode}
\makeatletter
\def\BState{\State\hskip-\ALG@thistlm}
\makeatother

\usepackage{cite}
\usepackage{enumitem}

\usepackage[margin=1in]{geometry}

\newtheorem{theorem}{Theorem}
\newtheorem{lemma}{Lemma}[section]
\newtheorem{invariant}{Invariant}

\newtheorem{claim}[lemma]{Claim}

\newtheorem*{claim*}{Claim}
\newtheorem*{proposition*}{Proposition}
\newtheorem*{lemma*}{Lemma}
\newtheorem*{problem*}{Problem}

\crefname{lemma}{Lemma}{Lemmas}
\crefname{claim}{Claim}{Claims}

\newtheorem{mdresult}{Result}

\newtheorem*{mdresult*}{Main Result}

\newtheorem{definition}[lemma]{Definition}
\theoremstyle{definition}
\newtheorem{remark}[lemma]{Remark}
\newtheorem{observation}[lemma]{Observation}

\newtheorem{mdinvariant}[lemma]{Lemma}

\theoremstyle{definition}
\newtheorem{mdalg}{Algorithm}
\newenvironment{Algorithm}{\begin{tbox}\begin{mdalg}}{\end{mdalg}\end{tbox}}

\newtheorem{mdimp}{Implementation}
\newenvironment{Implementation}{\begin{tbox}\begin{mdimp}}{\end{mdimp}\end{tbox}}

\allowdisplaybreaks

\renewcommand{\qed}{\nobreak \ifvmode \relax \else
      \ifdim\lastskip<1.5em \hskip-\lastskip
      \hskip1.5em plus0em minus0.5em \fi \nobreak
      \vrule height0.75em width0.5em depth0.25em\fi}

\newcommand{\Qed}[1]{\ensuremath{\qed_{\textnormal{~#1}}}}

\renewcommand{\leq}{\leqslant}
\renewcommand{\geq}{\geqslant}

\newcommand{\eps}{\ensuremath{\varepsilon}}

\newcommand{\card}[1]{\left\vert{#1}\right\vert}

\newcommand{\ceil}[1]{{\left\lceil{#1}\right\rceil}}
\newcommand{\floor}[1]{{\left\lfloor{#1}\right\rfloor}}

\newcommand{\set}[1]{\ensuremath{\left\{ #1 \right\}}}
\newcommand{\poly}{\mbox{\rm poly}}

\newcommand{\etal}{{\it et al.\,}}

\newenvironment{tbox}{\begin{tcolorbox}[
		enlarge top by=5pt,
		enlarge bottom by=5pt,
		boxsep=0pt,
		left=4pt,
		right=4pt,
		top=10pt,
		arc=0pt,
		boxrule=1pt,toprule=1pt,
		colback=white
		]
	}
	{\end{tcolorbox}}

\newcommand{\II}{\ensuremath{\mathbb{I}}}

\newcommand{\mireal}[1][]{
	\ifx\relax#1\relax%
	\II(\mione \,; \mitwo)%
	\else%
	\II(\mione \,; \mitwo\mid #1)%
	\fi
}

\newcommand{\bv}{\mathbf{v}}

\newcommand{\ph}{\phi}
\newcommand{\rank}{\textnormal{{rank}}}
\newcommand{\rmin}{\textnormal{{r-min}}}
\newcommand{\rmax}{\textnormal{{r-max}}}
\newcommand{\g}{g}
\newcommand{\G}{G}
\newcommand{\gs}{g^*}
\newcommand{\Gs}{G^*}
\newcommand{\del}{\Delta}

\newcommand{\St}{S}

\newcommand{\QS}{\ensuremath{\textnormal{\texttt{QS}}}\xspace}
\newcommand{\ei}[1]{\ensuremath{e_{#1}}}

\renewcommand{\bv}{\ensuremath{\textnormal{{b-value}}}}

\newcommand{\Band}{\ensuremath{\textnormal{Band}}}

\newcommand{\To}{type-1\xspace}
\newcommand{\Tt}{type-2\xspace}

\newcommand{\Bt}[1]{\ensuremath{B^{(#1)}}}

\newcommand{\Seg}{\ensuremath{\textnormal{seg}}}

\newcommand{\wt}{w}
\newcommand{\wtsum}[1]{W_{#1}}
\newcommand{\swt}{S_w}
\newcommand{\suf}{\mathsf{Unfold}(S_w)}
\newcommand{\WQS}{\ensuremath{\textnormal{\texttt{WQS}}}\xspace}

\newcommand{\cpy}[2]{#1^{(#2)}}
\newcommand{\tim}[1]{t_{#1}}

\newcommand{\mer}{\text{Merge}\xspace}

\newcommand{\DeleteTime}{\text{DeleteTime}}

\newcommand{\RNum}[1]{\uppercase\expandafter{\romannumeral #1\relax}}

\newcommand{\emm}{k}

\title{Generalizing Greenwald-Khanna Streaming Quantile Summaries for Weighted Inputs}

\author{
Sepehr Assadi\footnote{(\href{mailto:sepehr.assadi@rutgers.edu}{sepehr.assadi@rutgers.edu})
Department of Computer Science, Rutgers University. Research supported in part by the NSF 
CAREER Grant CCF-2047061, and gift from Google Research.} 
\and 
Nirmit 
Joshi\footnote{(\href{mailto:nirmit@u.northwestern.edu}{nirmit@u.northwestern.edu})
	Department of Computer Science, Northwestern University.} 
\and 
Milind Prabhu\footnote{(\href{mailto:milindpr@umich.edu}{milindpr@umich.edu})
	Department of Computer Science and Engineering,\! University of Michigan,\! Ann Arbor.} 
\and
Vihan Shah\footnote{(\href{mailto:vihan.shah98@rutgers.edu}{\text{vihan.shah98@rutgers.edu}}) 
Department of Computer Science, Rutgers University.  Research supported in part by the NSF 
CAREER Grant CCF-2047061. Part of this work was done when the author was an undergraduate 
student at Rutgers University-Camden and was supported in part by the NSF grant CCF-1910565.} 
}

\date{}

\begin{document}
\maketitle

\thispagestyle{empty}
\pagenumbering{roman}

\begin{abstract}
	Estimating quantiles, like the median or percentiles, is a fundamental task in data mining and data science. 
	A (streaming) quantile summary is a data structure that can process a set $S$ of $n$ elements in a streaming fashion and at the end, for any $\phi \in (0,1]$, return 
	a $\phi$-quantile of $S$ up to an $\eps$ error, i.e., return a $\phi'$-quantile  with $\phi' = \phi \pm \eps$. We are particularly interested in comparison-based 
	summaries that only compare elements of the universe under a total ordering and are otherwise completely oblivious of the universe. The best known deterministic quantile summary is the 20-year old Greenwald-Khanna (GK) summary that uses $O((1/\eps) \log{(\eps  n)})$ space [SIGMOD'01]. This bound was recently proved to be optimal
	for all deterministic comparison-based summaries by Cormode and Vesle\'y [PODS'20].
	
	\medskip
	
	 In this paper, we study weighted quantiles, a generalization of the quantiles problem, where each element arrives with a positive integer weight which denotes the number of copies of that element being inserted. The only known method of handling weighted inputs via GK summaries is the naive approach of breaking each weighted element into multiple unweighted items, and feeding them one by one to the summary, which results in a prohibitively large update time (proportional to the maximum weight of input elements).
  
	\medskip
	
	We give the first non-trivial extension of GK summaries for weighted inputs and show that it takes $O((1/\eps) \log{(\eps  n)})$ space and $O(\log(1/\eps)+\log\log (\eps n) )$ update time per element to process a stream of length $n$ (under some quite mild assumptions on the range of weights and $\eps$).
	En route to this, we also simplify the original GK summaries for unweighted quantiles. 

\end{abstract}

\clearpage
%
%
\setcounter{tocdepth}{3}
\tableofcontents

\clearpage

\pagenumbering{arabic}
\setcounter{page}{1}


\section{Introduction}
Given a set $S$ of elements $x_1,\ldots,x_n$ from a totally ordered universe, the rank of an element $x$ in this universe, denoted by $\rank(x)$, is the number of elements $x_j$ in $S$ with $x_j \leq x$. Similarly, 
the $\phi$-quantile of $S$, for any $\phi \in (0,1]$, is the element $x_i \in S$ with $\rank(x_i) = \ceil{\phi \cdot n}$. Computing quantiles is a fundamental problem with a wide range of applications considering they provide a concise representation of the distribution of 
the input elements. 
Throughout this paper, we solely focus on comparison-based algorithms for this problem that can only compare two elements of the universe according to their ordering and are otherwise
completely oblivious to the universe.

We are interested in the quantile estimation problem in the \emph{streaming} model, introduced in the seminal work of Alon, Matias, and Szegedy~\cite{AlonMS96}. In this model, the elements of $S$ arrive one by one in an arbitrary order, and the streaming algorithm 
can make just one pass over this data and use a limited memory and thus cannot simply store $S$ entirely. Already more than four decades ago, Munro and Paterson proved that one cannot solve this problem \emph{exactly} in the streaming model~\cite{MunroP78} and thus the focus has been on finding \emph{approximation} algorithms: Given $\eps > 0$, the algorithm
is allowed to return an $\eps$-approximate $\phi$-quantile, i.e., a $(\phi \pm \eps)$-quantile. 
More formally, we are interested in the following data structure: 

\begin{definition}[\textbf{Quantile Summary}]\label{def:summary}
	An \emph{$\eps$-approximate quantile summary} processes any set of elements in a streaming fashion and at the end finds an $\eps$-approximate $\phi$-quantile for any given quantile $\phi \in (0,1]$, defined
	as any $\phi'$-quantile for $\phi' \in [\phi - \eps, \phi + \eps]$.
\end{definition}

In the absence of the streaming aspect of the problem, one can always compute an $\eps$-approximate quantile summary in $O(1/\eps)$ space; simply store the $\eps$-quantile, $3\eps$-quantile, $5\eps$-quantile and so on from $S$. 
It is easy to see that given any $\phi$, returning the closest stored quantile results in an $\eps$-approximate $\phi$-quantile. It is also easy to see that this space is information-theoretically optimal for the problem. However, this approach cannot be directly 
implemented in the streaming model as a-priori it is not clear how to compute the needed quantiles of $S$ in the first place. 

The first (streaming) $\eps$-approximate quantile summary was proposed by Manku, Rajagopalan, and Lindsay~\cite{RajagopalanML98}. The MRL summary uses $O((1/\eps)\log^2{\!(\eps n)})$ space and requires prior knowledge of the length of the stream. This summary was soon after improved by Greenwald and Khanna~\cite{GreenwaldK01} who proposed the GK summary that uses $O((1/\eps)\log{\!(\eps n)})$ space and no longer
requires knowing the length of the stream. This is the state-of-the-art for deterministic comparison-based summaries. By allowing for randomization, one can further improve upon the space requirement of these algorithms and achieve bounds with no dependence on the length of the stream. The state-of-the-art result for randomized summaries is an algorithm due to Karnin, Lang and Liberty \cite{KarninLL16} which uses $O((1/\eps)\log\log{(1/\eps \delta)})$ space to construct an $\eps$-approximate quantile summary with probability at least $(1-\delta)$. We provide a more detailed discussion of the literature on randomized summaries and non-comparison based summaries in \Cref{sec:related}.

 While using randomization gives streaming algorithms which are more space-efficient, a major 
 drawback of most of these algorithms is that their analysis crucially depends on the assumption that 
 the input stream is independent of the randomness used by the algorithm. This assumption is 
 unrealistic in several settings; for instance, when the future input to the algorithm depends on its 
 previous outputs. Recently, this has invoked an interest in \emph{adversarially robust} algorithms 
 that work even when an adversary is allowed to choose the stream adaptively 
 \cite{mironov2011sketching, gilbert2012recovering, hardt2013robust, alon2021adversarial, 
 gilbert2012reusable, naor2015bloom}. Deterministic algorithms are inherently adversarially robust 
 and therefore understanding them is an interesting goal in itself.
 
 In this paper, we focus on deterministic summaries; specifically on furthering our understanding of GK summaries. Over the years, two important questions have been raised about them: Is it possible to \emph{improve} the space of GK summaries, perhaps even all the way down to the information-theoretic optimal bound of $O(1/\eps)$? And, is it possible to \emph{simplify} GK summaries and their intricate analysis in a way that allows for generalizations of these summaries to be 
more easily proposed and studied? (see Problem~2 of ``List of Open Problems in Sublinear Algorithms''~\cite{sub2} posed by Cormode or~\cite{KarninLL16,CormodeV20,LuoWYC16,AgarwalCHPWY12} for similar variations of this question, for example, when the input items are weighted). 

The first question  was addressed initially by Hung and Ting~\cite{HungT10} who proved an $\Omega((1/\eps)\log{\!(1/\eps)})$ space  lower bound for $\eps$-approximate quantile summaries, improving over the information-theoretic bound. Very recently, this question was fully settled by Cormode and Vesle\'y~\cite{CormodeV20} who proved that in fact GK summaries are asymptotically optimal: $\Omega((1/\eps)\log{\!(\eps n)})$ space is needed for any deterministic (comparison-based) summary. The second question above however is still left without a satisfying resolution. In this paper, we make progress toward answering this question by showing that the GK summary can be generalized to handle weighted inputs. Formally, we present algorithms to construct the following data-structure.

\begin{definition}[\textbf{Weighted Quantile Summary}]\label{def:weighted-summary}
Consider a weighted stream $\swt$ of $n$ updates $(x_i, \wt(x_i))$ for $1 \leq i \leq n$. The $i$-th update denotes the insertion of $\wt(x_i)$ copies of the element $x_i$ (the weight $\wt(x_i)$ is guaranteed to be a positive integer). We define $W_k = \sum_{i = 1}^ {k} \wt(x_i)$ to be the sum of the weights of the first $k$ elements of $\swt$. An \emph{$\eps$-approximate weighted quantile summary} is a data-structure that makes a single pass over $\swt$ and at the end, for any $\phi \in [0,1)$, finds an $x_j$ such that,
\begin{align}\label{eq: weighted-goal}
 \left(\sum\limits_{x_i < x_j} \wt(x_i) ,\;  \wt(x_j) + \sum\limits_{x_i < x_j} \wt(x_i)\right]   \cap \big[(\phi - \eps) \wtsum{n} ,(\phi + \eps) \wtsum{n}\big] \neq \emptyset.
\end{align}
\end{definition}

\subsection{Our Contributions.} 
One approach to construct a weighted quantile summary is to break each weighted item into multiple unweighted items and feed them to an unweighted summary such as the GK summary. However, such algorithms are slow since the time required to process an element is proportional to its weight. As such, it has been asked if faster algorithms exist. We answer this in the affirmative by proposing a fast algorithm for this problem in \Cref{sec: weighted-stream-gk}. In particular, this algorithm uses $O((1/\eps) \log{(\eps  n)})$ space and $O(\log(1/\eps)+\log\log (\eps n) )$ update time per element to process a stream of length $n$, when the weights are $\poly(n)$ and $\eps \geq \frac{1}{n^{1-\delta}}$ for any $\delta \in (0,1)$ (\Cref{thm:weighted-GK}). This matches the space and time complexity of the GK summary when used to summarize a stream of $n$ unweighted items \cite{GreenwaldK01, LuoWYC16}. To our knowledge, this constitutes the first (non-trivial) extension of the GK algorithm for weighted streams.

En route to this, we also present a new description of the GK summaries by simplifying or entirely bypassing several of their more intricate components in~\cite{GreenwaldK01} such as their so-called ``tree representation'' and their complex ``compress'' operations in \Cref{sec:final}. 
As a warm-up to this, we also present a simple and greedy algorithm for unweighted quantiles in \Cref{sec:basic} and show that it requires $O((1/\eps) \log^2 (\eps n))$ space (we also extend this algorithm to the weighted setting in \Cref{sec:wt-greedy}). 
This algorithm, although has a suboptimal space bound, will be useful in motivating and providing intuition for GK summaries. Interestingly, this summary is quite similar (albeit \emph{not} identical) to the so-called GKAdaptive summary~\cite{LuoWYC16} that was already proposed by~\cite{GreenwaldK01} as a more practical variant of their GK summaries (Luo~\etal~\cite{LuoWYC16} further confirmed this by showing that GKAdaptive outperforms GK summaries experimentally). While no theoretical guarantees are known for GKAdaptive, we prove that this slight modification of this algorithm submits to a simple analysis of an $O((1/\eps)\log^2{\!(\eps n)})$ space upper bound (\Cref{thm:greedy}).

We also emphasize that, similar to the original GK summaries, our weighted extension does not need the knowledge of the stream length. This guarantee implies that we can track the quantiles throughout the stream, with error proportional to the current weight of the stream, and not only at the end.

\subsection{Further Related Work}\label{sec:related}

Quantile estimation in the streaming model has been extensively studied from numerous angles; we 
refer the interested reader to~\cite{MunroP78,GuhaM09} for results on multi-pass 
algorithms,~\cite{GuhaM09,ChakrabartiJP08} for random-arrival 
streams,~\cite{CormodeKMS06,GuptaZ03} for biased quantiles,
and~\cite{AgarwalCHPWY12} for mergeable summaries, as some representative examples. In the 
following, we limit our discussion to the most basic variant of this problem: obtaining an 
$\eps$-approximate summary in a single pass over an
adversarially ordered stream. 

We already discussed the deterministic summaries of~\cite{GreenwaldK01,RajagopalanML98} and 
the lower bounds in~\cite{HungT10,CormodeV20} (see also~\cite{LuoWYC16} for an experimental 
study of these results). 
We now discuss the literature on randomized summaries.
A simple application of Chernoff bound shows that sampling $O((1/\eps^2)\log{(1/\eps)})$ elements 
from the stream
results in an $\eps$-approximate quantile summary with high constant probability. Assuming that we 
know the length of the stream, we can sample these elements beforehand and feed them to the GK summary and obtain a randomized algorithm with space $O((1/\eps)\log{(1/\eps)})$ 
(see~\cite{MankuRL99}). 
Felber and Ostrovsky showed in~\cite{FelberO15} how to lift the assumption on the knowledge of 
the length of the stream in this approach. Finally, Karnin, Lang, and Liberty
designed a randomized summary of size $O((1/\eps)\log\log{(1/\eps )})$ for this problem which 
constitutes the state-of-the-art (the algorithm of~\cite{KarninLL16} answers a \emph{single} quantile 
query with a constant probability of success in $O(1/\eps)$ space 
which is optimal; see~\cite{KarninLL16,CormodeV20} for more on lower bounds for randomized 
summaries). An extension of this algorithm for the weighted quantiles problem was proposed in  
\cite{ivkin2019streaming}. This extension uses space $O( (1/\eps) \cdot \sqrt{\log 
\left(1/\eps\right)})$ and  $O(\log \left( 1/\eps \right))$ update time per element to output \textit{all} 
$\eps$-approximate quantiles with high probability.

In parallel to this line of work, researchers have also considered non comparison-based summaries 
over restricted universes (see, e.g.~\cite{CormodeM04,ShrivastavaBAS04,2013munro}). For 
instance, 
Shrivastava~\etal~\cite{ShrivastavaBAS04} designed a simple deterministic $\eps$-approximate 
quantile in $O(\frac{1}{\eps} \log{\card{U}})$ space for a known universe $U$ of elements. 
The space bound of this algorithm is incomparable to that of GK summaries (and our results) 
and the lower bounds of~\cite{HungT10,CormodeV20} no longer hold for these summaries. We note 
that using these non comparison-based summaries requires 
prior knowledge of the universe $U$, making them impractical for tracking quantiles over streams of 
floating point values or strings.


\section{Basic Setup (Unweighted Setting)}\label{sec:basics}

We first present the basic setup of our unweighted quantile summaries and preliminary definitions. Then, in the next two sections, we show how these summaries can be maintained in the streaming model using limited space and fast update time. 
We note that our summary closely mimics that of~\cite{GreenwaldK01}, and the main differences are in the algorithms that maintain the summary during the stream and in the analysis. We may repeat some things in the main paper for the sake of completeness.

\subsection{The Quantile Summary} 
Our summary, denoted by $\QS$, is naturally a collection of elements seen in the stream along with some metadata stored for each one.  We use $s$ to denote the 
number of elements stored in $\QS$ and $\ei{i}$, for any $i \in [s]$, to refer to the $i$-th smallest element in $\QS$ (we use $e$ to refer to an arbitrary element in $\QS$ when its rank is not relevant). 
The metadata stored per element will also fit in $O(1)$ words so that the total size of the summary is $O(s)$ words.
In particular, the main information for each element $e$ in $\QS$ is the following (the extra stored information will be defined in~\Cref{sec:time}): 
\begin{itemize}
	\item $\rmin(e)$ and $\rmax(e)$: lower and upper bounds maintained by $\QS$ on $\rank(e)$ among the elements seen so far in the stream (since we are not storing all elements, we cannot determine the exact rank of a stored element and thus focus on maintaining proper lower and upper bounds). Here, and throughout the paper, we let $\rmin(\ei{0}) = 0$ for notational convenience. 

\end{itemize}

For simplicity of exposition and to remove the corner cases, we insert a $+\infty$ element at the start of the stream, which is considered larger than any other element, and store it in $\QS$ as $\ei{s}$.
The $\rmin$ and $\rmax$ of this element is also always equal to the number of visited elements (including itself). Since $+\infty$ is the largest element, inserting it in $\St$ does not affect the rank of any other element.

During the stream, we need to be able to insert an element into or delete one from the summary, which is done as follows (the main task of our streaming algorithms is to decide which elements to insert and which ones to remove during the stream); see~\Cref{fig:insert-delete} for an illustration. 

\begin{tbox}
	\textbf{Insert($x$).} Inserts a given element $x$ into $\QS$. 
	
	\begin{enumerate}[label=$(\roman*)$]
	\addtolength{\itemindent}{3mm}
		\item Find the smallest element $\ei{i}$ in $\QS$ such that $\ei{i} > x$;
		\item Set $\rmin(x) = \rmin(\ei{i-1})+1$ and $\rmax(x) = \rmax(\ei{i})$; moreover, increase $\rmin(\ei{j})$ and $\rmax(\ei{j})$ by one for all $j \geq i$. Store element $x$ in $\QS$. 
	\end{enumerate}
	
	\textbf{Delete($\ei{i}$).} Deletes the element $\ei{i}$ from $\QS$. 
	
	\begin{enumerate}[label=$(\roman*)$]
	\addtolength{\itemindent}{3mm}
		\item Remove element $\ei{i}$ from $\QS$; keep all remaining $\rmin$, $\rmax$ values unchanged. 
		 
	\end{enumerate}
\end{tbox} 

It is immediate to verify that both these operations maintain $\rmin,\rmax$ values correctly. The following claim now shows that
as long as we can ensure ``proper'' bounds on $\rmin,\rmax$ values, we will be able to answer quantile queries correctly. 

\begin{figure}[t]		
\centering	
 \hspace{1.5cm}\resizebox{0.9\textwidth}{!}{\hspace*{-5em}\begin {tikzpicture}[-latex, auto ,node distance =7mm and 13mm ,on grid ,
semithick ,
state/.style ={ rectangle ,top color =white , bottom color = white ,
draw,black , text=black , minimum width =11 mm},
labelstate/.style ={ rectangle ,top color =white , bottom color = white ,
draw,white , text=black , minimum width =11 mm}]

\node[state] (A){\footnotesize$10$}; 
\node[state] (B) [right =of A] {\footnotesize$21$}; 
\node[state] (C) [right =of B] {\footnotesize$30$}; 
\node[labelstate] (L) [left =20mm of A] {\footnotesize \QS:};

\node[labelstate] (A1) [below =of A] {\footnotesize$(3,5)$}; 
\node[labelstate] (B1) [right =of A1] {\footnotesize$(6,9)$}; 
\node[labelstate] (C1) [right =of B1] {\footnotesize$(13,13)$}; 
\node[labelstate] (L1) [left =20mm of A1] {\footnotesize$( \textnormal{rmin,rmax})$:};

\node[labelstate] (A2) [below =of A1] {\footnotesize $(2,2)$}; 
\node[labelstate] (B2) [right =of A2] {\footnotesize $(3,3)$}; 
\node[labelstate] (C2) [right =of B2] {\footnotesize$(7,0)$}; 
\node[labelstate] (L2) [left =20mm of A2] {\footnotesize $( g,\Delta)$:};

\node[labelstate] (Ins) [above right =of C] {\footnotesize $\textbf{Insert}(25)$}; 

\node[labelstate] (D) [below right =of C] {$\Longrightarrow$}; 
\node[state] (E) [above right =of D] {\footnotesize\textcolor{black}{$10$}}; 
\node[state] (F) [right =of E] {\footnotesize$21$}; 
\node[state] (G) [right =of F] {\footnotesize\textcolor{black}{25}}; 
\node[state] (H) [right =of G] {\footnotesize$30$}; 

\node[labelstate] (E1) [below =of E] {\footnotesize\textcolor{black}{$(3,5)$}}; 
\node[labelstate] (F1) [right =of E1] {\footnotesize$(6,9)$}; 
\node[labelstate] (G1) [right =of F1] {\footnotesize\textcolor{black}{$(7,13)$}}; 
\node[labelstate] (H1) [right =of G1] {\footnotesize$(14,14)$}; 

\node[labelstate] (E2) [below =of E1] {\footnotesize\textcolor{black}{$(2,2)$}}; 
\node[labelstate] (F2) [right =of E2] {\footnotesize$(3,3)$}; 
\node[labelstate] (G2) [right =of F2] {\footnotesize\textcolor{black}{$(1,6)$}}; 
\node[labelstate] (H2) [right =of G2] {\footnotesize$(7,0)$}; 

\node[labelstate] (Del) [above right =of H] {\footnotesize\footnotesize $\textbf{Delete}(10)$}; 
\node[labelstate] (I) [below right =of H] {$\Longrightarrow$}; 
\node[state] (J) [above right =of I] {\footnotesize$21$}; 
\node[state] (K) [right =of J] {\footnotesize$25$}; 
\node[state] (L) [right =of K] {\footnotesize$30$}; 

\node[labelstate] (J1) [below =of J] {\footnotesize$(6,9)$}; 
\node[labelstate] (K1) [right =of J1] {\footnotesize$(7,13)$}; 
\node[labelstate] (L1) [right =of K1] {\footnotesize$(14,14)$}; 

\node[labelstate] (J2) [below =of J1] {\footnotesize$(5,3)$}; 
\node[labelstate] (K2) [right =of J2] {\footnotesize$(1,6)$}; 
\node[labelstate] (L2) [right =of K2] {\footnotesize$(7,0)$}; 

\end{tikzpicture}}	
\caption{An illustration of the update operations in the summary starting from some arbitrary state (the parameters $(g,\Delta)$ in this figure are defined in~\Cref{sec:inv}).}	
\label{fig:insert-delete}	
\end{figure}

\begin{claim}\label{lem:search}
	Suppose $\QS$ is a summary of an unweighted stream of length $n$ and satisfies $\rmax(\ei{i}) - 
	\rmin(\ei{i-1}) \leq \eps n$ for all $i \in [s]$; then $\QS$ is an $\eps$-approximate quantile 
	summary.  
\end{claim}
\begin{proof}
	
	Given any $\ph$, let $r= \ceil{\phi \cdot n}$. We will show that $\QS$ always contains an element 
	$e$ with rank in $[r-\eps n, r+\eps n]$, namely, $\rmin(e)\geq r-\eps n$ and $\rmax(e)\leq r+\eps 
	n$. Returning this element
	as the answer would thus give us an $\eps$-approximate $\phi$-quantile.

	For any $r\leq \eps n$, we have $\rmin(\ei{1}) \geq 0 \geq r-\eps n$ and $\rmax(\ei{1}) \leq 
	\rmin(\ei{0}) + \eps n  \leq r+ \eps n$ by the lemma's condition and thus
	we can return $\ei{1}$. Otherwise, let $\ei{j}$ be the smallest element with $\rmin(\ei{j}) \geq  
	r-\eps n$, thus $\rmin(\ei{j-1}) < r-\eps n$; again, by the lemma's condition, we have
	\[
	\rmax(\ei{j}) \leq \rmin(\ei{j-1})+ \eps n < r-\eps n +\eps n\leq r+\eps n,
	\] 
	and thus we can return $\ei{j}$ here. 
\end{proof}

The main goal of our algorithms in subsequent sections is to update $\QS$ in a way that the condition of~\Cref{lem:search} is satisfied without having to store too many elements.

\subsection{Time steps and Bands}\label{sec:time}

Another important notion is that of \emph{time steps} and \emph{bands}. 
In this section, we define this for the unweighted setting, and then we build upon it to define them 
for the weighted setting in \Cref{sec: weighted-prelims}.
We measure the time as the number of elements 
appeared in the stream so far in multiples of $\Theta(1/\eps)$. Formally, 
\begin{definition}[\textbf{Time Steps}]\label{def:time}
	Let $\ell := \frac{1}{\eps}$, which we assume is an integer. We partition the stream into consecutive 
	\textbf{chunks} of size $\ell$; the \textbf{time step} $t$ then refers to the $t$-th chunk of 
	elements denoted by $(x^{(t)}_1,\ldots,x^{(t)}_\ell)$ 
	(we assume that the length of the stream is a multiple of $\ell$)\footnote{Both assumptions in this 
	definition are without loss of generality: we can change the value of $\eps$ by an $O(1)$ factor to 
	guarantee the first one and add $O(1/\eps)$ dummy
	elements at the end to guarantee the second one.}. We define $t_0(x)$ as the time step in which 
	$x$ appears in the stream.
\end{definition}
The next important definitions are \emph{band-values} and \emph{bands} borrowed 
from~\cite{GreenwaldK01}. Roughly speaking, we would like to be able to partition 
elements of the stream into a ``small'' number of groups (bands) so that elements within a group 
have ``almost the same'' time of insertion (as a proxy on how accurate our estimate of their 
$\rmin,\rmax$ is) For more details, see \Cref{sec:time}. Formally, 
\begin{definition}[\textbf{Band Values and Bands}]
	\label{def:band}
	For any element $x$ of the unweighted stream $S$, we assign an integer called a 
	\textbf{band-value}, denoted by $\bv(x)$, as follows: 
	\begin{enumerate}[label=$(\roman*)$]
		\item At the time step $t = t_0(x)$, we set $\bv(x) = 0$;
		\item At any time step $t > t_0(x)$, if $t$ is a multiple of $2^{\bv(x)}$, then we increase 
		$\bv(x)$ by one. 
	\end{enumerate}
\end{definition}

\Cref{def:band} is basically a geometric grouping of elements based on the value $t-t_0(x)$ at every 
time step $t$ (see~\Cref{fig:band}). 
The only difference of this grouping from typical geometric grouping ideas
is that it remains ``stable'' over time, in that elements that enter the same band at some point, 
continue to maintain the same band-value from thereon (although possibly with different band-values over time). This is formalized in the following two observations.

\begin{observation}\label{obs:band-length}
	For an element $x$ and at any $t \geq  t_0(x)$, $2^{\bv(x)-1}-2 \leq t-t_0(x) \leq 2^{\bv(x)+1}$. 
\end{observation}

This is simply because whenever $\bv(x)$ becomes some $\alpha$, it takes $x$ at least $2^{\alpha-1}$  and at most $2^{\alpha}$ time steps to increase its band-value again by one (by~\Cref{def:band}). Thus, 
\begin{align*}
	t-t_0(x) + 1 &= \sum_{\alpha=0}^{\bv(x)} \text{\# of time steps with band-value $\alpha$}  \geq \sum_{\alpha=1}^{\bv(x)-1} 2^{\alpha-1} = 2^{\bv(x)-1}-1; \\
	t-t_0(x) + 1&= \sum_{\alpha=0}^{\bv(x)} \text{\# of time steps with band-value $\alpha$} \leq \sum_{\alpha=0}^{\bv(x)} 2^{\alpha} = 2^{\bv(x)+1}. 
\end{align*}

A corollary of~\Cref{obs:band-length} is that the \emph{number of band-values} at a time $t$ is $\Bt{t} = O(\log t)$ and that at any point of time, the total number of elements belonging to bands $0$ to $\alpha$ is at most $\ell \cdot 2^{\alpha+1}$ 
 (as for any element $x$ with $\bv(x) \leq \alpha$, $t-t_0(x) \leq 2^{\alpha+1}$ and thus 
only elements arriving in the most recent $2^{\alpha+1}$ steps may belong to bands from $0$ to $\alpha$  and each time step includes $\ell$ elements). We record this in the following equation for future reference: 
\begin{align}
	\text{\# of band values } \Bt{t}=O(\log t) \quad \textnormal{and} \quad \text{$\card{\Band_{\leq \alpha}} \leq \ell \cdot 2^{\alpha+1}$ for all $\alpha \geq 0$}. \label{eq:band-alpha}
\end{align}

\begin{figure}[t]		
\centering	
\subcaptionbox{Progression of the band-values of elements inserted at time step $2$. \label{fig:band-progress}}%
  [0.76\linewidth]{ \resizebox{.75\textwidth}{!}{

\begin{tikzpicture}

\node[ minimum width=12pt, minimum height=20pt, inner sep=0pt] (10) {\scriptsize 3}; 
\draw ($(10)+(0pt,5pt)$) to ($(10) + (0pt,15pt)$); 
\foreach \x in {1,...,15}
{
	\pgfmathtruncatemacro{\prev}{\x - 1}
	\pgfmathtruncatemacro{\name}{1\prev}
	\pgfmathtruncatemacro{\xx}{\x+3}
	\node[minimum width=12pt, minimum height=20pt, inner sep=0pt] (1\x) [right=5pt of \name]{\large \xx}; 
	\draw ($(1\x)+(0pt,5pt)$) to ($(1\x) + (0pt,15pt)$); 
}

\draw[line width=1pt, ->] ($(10) + (-10pt,10pt)$) to ($(115) + (+10pt,10pt)$);

\node at ($(115) + (+30pt,0pt)$) {\normalsize Time};

\node[ minimum width=12pt, minimum height=15pt, inner sep=0pt] (b11) at ($(10) + (-20pt,+20pt)$) {\scriptsize 1}; 
\draw ($(b11)+(5pt,0pt)$) to ($(b11) + (15pt,0pt)$); 
\foreach \x in {2,...,4}
{
	\pgfmathtruncatemacro{\prev}{\x - 1}
	\pgfmathtruncatemacro{\name}{1\prev}
	\node[minimum width=12pt, minimum height=15pt, inner sep=0pt] (b1\x) [above=0pt of b\name]{\large \x}; 
	\draw ($(b1\x)+(5pt,0pt)$) to ($(b1\x) + (15pt,0pt)$); 
	
}

\draw[line width=1pt, ->] ($(10) + (-10pt,10pt)$) to ($(10) + (-10pt,80pt)$); 

\node at ($(10) + (-40pt,80pt)$) {\normalsize Band-value};


\node[circle, draw, fill=black, inner sep=1pt] (c1) [above=7pt of 10]{};

\node[circle,  fill=white, inner sep=1pt] (c) [left=3pt of 10]{};
\node[circle, draw, fill=black, inner sep=1pt] (c0) [above=5.5pt of c]{};

\node[circle, draw, fill=black, inner sep=1pt] (c2) [above=22pt of 11]{};
\node[circle, draw, fill=black, inner sep=1pt] (c3) [above=22pt of 12]{};
\node[circle, draw, fill=black, inner sep=1pt] (c4) [above=22pt of 13]{};
\node[circle, draw, fill=black, inner sep=1pt] (c5) [above=22pt of 14]{};

\node[circle, draw, fill=black, inner sep=1pt] (c6) [above=37pt of 15]{};
\node[circle, draw, fill=black, inner sep=1pt] (c7) [above=37pt of 16]{};
\node[circle, draw, fill=black, inner sep=1pt] (c8) [above=37pt of 17]{};
\node[circle, draw, fill=black, inner sep=1pt] (c9) [above=37pt of 18]{};
\node[circle, draw, fill=black, inner sep=1pt] (c10) [above=37pt of 19]{};
\node[circle, draw, fill=black, inner sep=1pt] (c11) [right=15pt of c10]{};
\node[circle, draw, fill=black, inner sep=1pt] (c12) [right=15pt of c11]{};
\node[circle, draw, fill=black, inner sep=1pt] (c13) [right=15pt of c12]{};

\node[circle, , fill=white, inner sep=1pt] (c14) [right=15pt of c13]{};

\node[circle, draw, fill=black, inner sep=1pt] (c15) [above=15pt of c14]{};
\node[circle, draw, fill=black, inner sep=1pt] (c16) [right=15pt of c15]{};
\node[circle, draw, fill=black, inner sep=1pt] (c17) [right=15pt of c16]{};

\node[circle, , fill=white, inner sep=1pt] (c18) [right=15pt of c17]{};

\draw (c0) to (c1);
\draw (c1) to (c2);
\draw (c2) to (c3);
\draw (c3) to (c4);
\draw (c4) to (c5);
\draw (c5) to (c6);
\draw (c6) to (c7);
\draw (c7) to (c8);
\draw (c8) to (c9);
\draw (c9) to (c10);
\draw (c10) to (c11);
\draw (c11) to (c12);
\draw (c12) to (c13);
\draw (c13) to (c15);
\draw (c15) to (c16);
\draw (c16) to (c17);
\draw (c17) to (c18);

\end{tikzpicture}}}	
\vspace{0.5cm}	

\subcaptionbox{Distribution of band values at $t=15$. \label{fig:dist-bands}}%
 [1\linewidth]{ \resizebox{0.75\textwidth}{!}{

\hspace*{-5em}
\begin {tikzpicture}[-latex, auto ,node distance =8mm and 8mm ,on grid ,
semithick ,
state1/.style ={ rectangle ,top color =white , bottom color = white ,
draw,black , text=black , minimum width =7 mm},
labelstate/.style ={ rectangle ,top color =white , bottom color = white ,
draw,white , text=black , minimum width =7 mm}]

\node[state1] (A1){$15$}; 
\node[state1] (A2) [right =of A1] {$14$}; 
\node[labelstate] (A3) [right =of A2] {}; 
\node[state1] (A4) [right =of A3] {$13$}; 
\node[state1] (A5) [right =of A4] {$12$}; 
\node[state1] (A6) [right =of A5] {$11$}; 
\node[state1] (A7) [right =of A6] {$10$}; 
\node[labelstate] (A8) [right =of A7] {}; 
\node[state1] (A9) [right =of A8] {$9$}; 
\node[state1] (A10) [right =of A9] {$8$}; 
\node[state1] (A11) [right =of A10] {$7$}; 
\node[state1] (A12) [right =of A11] {$6$}; 
\node[state1] (A13) [right =of A12] {$5$}; 
\node[state1] (A14) [right =of A13] {$4$}; 
\node[state1] (A15) [right =of A14] {$3$}; 
\node[state1] (A16) [right =of A15] {$2$}; 
\node[labelstate] (A17) [right =of A16] {}; 
\node[state1] (A18) [right =of A17] {$1$}; 

\node[labelstate] (Blank) [left =of A1] {}; 
\node[labelstate] (Blankx) [below =of A1] {}; 
\node[labelstate] (Blank1) [left =of Blank] {}; 

\node[labelstate] (L) [left =of Blank1] {Chunk number:}; 

\node[labelstate] (L2) [above =of L] {Band number:}; 

\node[labelstate] (B1) [above =of A1] {1}; 
\node[labelstate] (B3) [above =of A4] {2}; 
\node[labelstate] (B6) [above =of A9] {3}; 
\node[labelstate] (B10) [above =of A18] {4};

\end{tikzpicture}}}	
\caption{An illustration of band-values and bands.}	
\label{fig:band}	
\end{figure}

\begin{observation}\label{obs:stable}
	Suppose $\bv(x) \leq \bv(y)$ for elements $x,y$ at time $t$; then, during any time step $t' \geq t$, $\bv(x) \leq \bv(y)$. 
\end{observation}

This is simply because band-values of elements are updated \emph{simultaneously} based on the value of the current time step (and not a relative number). 

We will now prove a lemma that allows us to quickly compute $\bv$ of elements.
\begin{lemma}
\label{eq: closed-form}
	For an element $e$ and an integer $\alpha\geq 1$, $\bv(e) = \alpha$ after time step $t$ if and only if it satisfies:
	
	\begin{equation}
	2^{\alpha-1}  + (t \bmod 2^{\alpha-1}) \leq t-t_0(e)  < 2^{\alpha} + (t \bmod 2^{\alpha}).
	\end{equation}
	\end{lemma}

	\begin{proof}
	    Define $I_{\alpha,t} = [2^{\alpha-1}  + (t \bmod 2^{\alpha-1}), \, 2^{\alpha} + (t \bmod 2^{\alpha}))$.
		We first show by induction on $t$ that for an element $e$ with $\bv(e) = \alpha$, $t- t_0(e) \in I_{\alpha,t}$. When $e$ is first inserted into the summary, it has $\bv$ $0$ and in the very next time step it is promoted to $\Band_{1}$. If we consider time $t$, which is one time step after the time of insertion of $e$, it is easy to verify that $(t - t_0(e)) = 1$ lies in the interval $I_{1,t_{0}(e)}$. Suppose that the claim is true for $e$ up to time $t$. Let $\bv$ of $e$ at time $t$ be $\alpha$. 
		
		If $(t+1)$ is a multiple of $2^{\alpha}$, $\bv(e)$ increases to $\alpha + 1$. We wish to show that $(t+1) - t_0(e)$ is in the interval $I_{\alpha +1, t + 1}$. We observe that
		\begin{align*}
		(t+1) - t_0(e) &\in [2^{\alpha-1} +1 + (t \bmod 2^{\alpha-1}), \, 2^{\alpha} + 1 + (t \bmod 2^{\alpha})) \tag{Induction hypothesis}\\
		&= [2^\alpha  , 2^{\alpha+1} ) \tag{ $(t+1)$ is a multiple of $2^\alpha$} \\
		&= I_{\alpha +1, t+1}.
		\end{align*}

		If $(t+1)$ is not a multiple of $2^{\alpha}$, $\bv(e)$ does not change. We will show that $(t+1) - t_0(e)$ belongs to $I_{\alpha,t+1}$ in this case. By the induction hypothesis,
		\begin{align*}
	    (t+1) - t_0(e) & \leq  2^{\alpha} + 1 + (t \bmod 2^{\alpha})\\
	     & = 	2^{\alpha} + ((t+1) \bmod 2^{\alpha})\tag{$(t+1)$ is not a multiple of $2^\alpha$.}.
		\end{align*}
		
		To show the lower bound, we note that
		\begin{align*}
		(t+1) - t_0(e) & \geq 2^{\alpha-1} + 1 + (t \bmod 2^{\alpha-1})\\
		& \geq 2^{\alpha-1} + ((t+1) \bmod 2^{\alpha-1}). \tag{$(t \bmod 2^{\alpha-1}) +1 \geq (t+1) \bmod 2^{\alpha-1})$}
		\end{align*}
		
		This shows that $(t+1) - t_0(e)$ belongs to $I_{\alpha,t+1}$.
		
		The other direction of the lemma follows from the disjointness of the intervals $I_{\alpha,t}$ for a fixed $t$ and the fact that $(t-t_0(e))$ always belongs to $I_{\bv(e),t}$.
			\end{proof}
    
    It is easy to see that, using~\Cref{eq: closed-form}, we can efficiently compute the band of an element in $O(1)$ time. We note this observation in the following.
    \begin{observation}\label{obs: bvalue-compute}
There is an $O(1)$ time algorithm to compute $\bv(e)$ at time $t$ given $t_0(e)$.
\end{observation}

\subsection{Indirect Handling of (\texorpdfstring{$\rmin,\rmax$}{rmin,rmax}) and the Main Invariant}\label{sec:inv}

 It turns out that, as was observed in~\cite{GreenwaldK01}, instead of working with $\rmin,\rmax$ values directly, it would be much easier to work with the following two parameters: 
\begin{align}
	\g_i := \rmin(\ei{i}) - \rmin(\ei{i-1}) \qquad \del_i := \rmax(\ei{i}) - \rmin(\ei{i});
	\label{eq:g-delta}
\end{align}

Clearly, $\rmin(\ei{i})=\sum_{j=1}^{i} \g_j$ and $\rmax(\ei{i})=\del_i+ \sum_{j=1}^{i} \g_j$. Thus storing $(\g, \del)$ values for each
element is equivalent to storing their $\rmin$ and $\rmax$ values for our purpose. The main invariant we require from our algorithms is then the following. 

\begin{invariant}\label{inv:g-delta}
	At any time $t$ and for any element $\ei{i} \in \QS$, we have $g_i + \Delta_i \leq t$. 
\end{invariant}

By~\Cref{eq:g-delta}, $g_i + \Delta_i = \rmax(\ei{i}) - \rmin(\ei{i-1})$. As such, since $t = n/\ell = \eps n$ on a length $n$ stream, any algorithm that maintains~\Cref{inv:g-delta} is an $\eps$-approximate quantile summary 
by~\Cref{lem:search}. 

Let us now briefly point out the benefit of working with $(g,\Delta)$ values. Firstly, the following observation states how these values change by  the  operations of the summary (see~\Cref{fig:insert-delete}).

\begin{observation}\label{obs:g-del}
	In the summary $\QS$: 
	\begin{itemize}
		\item \textbf{Insert($x$):} Sets $\g(x) = 1$ and $\del(x) = \g_i + \del_i - 1$ and keeps the remaining $(\g,\Delta)$ values unchanged (here, ($\g_i,\del_i)$ refer to the corresponding values of $\ei{i}$ defined in the procedure); 
		\item \textbf{Delete($\ei{i}$):} Sets $\g_{i+1} = g_{i+1}+g_i$ and keeps the remaining $(\g,\Delta)$ values unchanged.
	\end{itemize}
\end{observation}
The proof of this observation is a direct corollary of~\Cref{eq:g-delta} and how these operations change $\rmin,\rmax$ values. We can now interpret each of $g,\Delta$ values as follows: 

\noindent
\textbf{$\bm{g}$-value.} We set the $g$-value of a newly inserted element to $1$. After that, $g$-value of an element $\ei{i}$ in $\QS$ can only change when $\ei{i-1}$ is deleted, which results in $g_{i} = g_{i}+g_{i-1}$. 
We say that $\ei{i}$ \emph{covers} $\ei{i-1}$ whenever $\ei{i-1}$ is deleted from
	the summary, in which case $\ei{i}$ also covers all elements that $\ei{i-1}$ was covering so far (every element only covers itself upon insertion). We define 
	\begin{itemize}
		\item $C(\ei{i})$: the set of elements covered by $\ei{i}$. By definition, at any point in time, 
		\begin{align} \label{obs: disjoint coverage}
		g_i = \card{C(\ei{i})} 
		\text{ and } C(\ei{i}) \cap C(\ei{j}) = \emptyset
		\end{align}
		for any $\ei{i},\ei{j}$ currently stored in $\QS$. 
	\end{itemize}

	Notice that for any element $\ei{i}$, $g_i$ represents the number of elements that are deleted from the stream and are now being covered by $\ei{i}$. Thus, a ``large'' value of $g_i$ for an element stored in $\QS$, while being a challenge 
	in maintaining~\Cref{inv:g-delta}, necessarily means that many elements of the stream are already deleted in $\QS$. This is a key property that helps us in bounding the size of $\QS$. 

\noindent
\textbf{$\bm{\Delta}$-value.} We set the $\Delta$-value of a newly inserted element to $g_{i}+\Delta_i - 1$ (where $\ei{i}$ is the smallest element of $\QS$ larger than the inserted element). After that, the $\Delta$-value of an element 
can no longer change (unlike $g$-values). Maintaining~\Cref{inv:g-delta} (for the element $\ei{i}$) then allows us to state: 
\begin{align}
	\Delta(x) \leq t_0(x). \label{eq:del-t0} 
\end{align}
As such, the earlier an element is inserted into the stream, the smaller its worst-case $\Delta$-value is. This, combined with the geometric grouping nature of bands we discussed earlier, 
implies that at any time $t$, the number of elements with ``relatively small'' $\Delta$-value (compared to $t$, namely, the current target of~\Cref{inv:g-delta}) is almost exponentially more than the elements with ``larger'' $\Delta$-value. 
This is the second main property that is going to help us in maintaining~\Cref{inv:g-delta} in a small space. 



\section{A Greedy \texorpdfstring{$O\!\left(\frac{1}{\eps} \cdot \log^2{\!(\eps  n)}\right)$ Size}{} Summary}\label{sec:basic}
As a warm-up to our main algorithm, we first present a very simple and greedy way of updating the quantile summary $\QS$ to maintain~\Cref{inv:g-delta} in $O(\frac{1}{\eps} \cdot \log^2{\!(\eps n)})$ space. 
At any step, the algorithm first inserts all arriving elements into $\QS$. Then, in the \emph{deletion step}, while there is an element whose deletion does not violate ~\Cref{inv:g-delta} (and another simple condition on \bv s), the algorithm deletes it from \QS. 


\begin{Algorithm}\label{alg:simple}
	A greedy algorithm for updating the quantile summary.
	
	\medskip
	
	\medskip
	
	For each time step $t$ with arriving items $(x^{(t)}_1,\ldots,x^{(t)}_\ell)$: 
	\begin{enumerate}[label=$(\roman*)$]
	\addtolength{\itemindent}{3mm}
		\item 
		Run $\textbf{Insert}(x^{(t)}_j)$ for each element of the chunk. 
		\item  Repeatedly run $\textbf{Delete}(\ei{i})$ for any (arbitrarily chosen) element $\ei{i}$ in \QS satisfying:
		\[
		(1)~\bv(\ei{i}) \leq \bv(\ei{i+1})\qquad \text{\underline{and}} \qquad (2)~\g_i+\g_{i+1}+\Delta_{i+1} \leq t 
		\]
	\end{enumerate}
\end{Algorithm}

\begin{theorem}\label{thm:greedy}
For any $\eps > 0$ and a stream of length $n$,~\Cref{alg:simple} maintains an $\eps$-approximate quantile summary in $O(\frac{1}{\eps} \cdot \log^2{\!(\eps n)})$ space. Also, there is an implementation of \Cref{alg:simple} that takes $O\big(\log(1/\eps)+ \log \log (\eps n)\big)$ worst-case processing time per element.
Finally, quantile queries can be answered in $O\big(\log(1/\eps)+ \log \log (\eps n)\big)$ worst-case time per query.
\end{theorem}

\Cref{alg:simple} maintains~\Cref{inv:g-delta} since it may only delete an element $\ei{i}$ if $g_i + g_{i+1} + \Delta_{i+1} \leq t$, which by~\Cref{obs:g-del} implies that $g_{i+1} + \Delta_{i+1} \leq t$ \emph{after} the deletion. The other $(g,\Delta)$-values remain
unchanged. As argued in~\Cref{sec:inv}, maintaining~\Cref{inv:g-delta} directly implies that $\QS$ is an $\eps$-approximate quantile summary throughout the stream. We now bound the size of $\QS$ under
this algorithm, taking a step towards proving~\Cref{thm:greedy}. Then, in ~\Cref{subsec: fast-imp-simple}, we show how this algorithm can be implemented quickly.

\begin{remark}
	The difference between the GKAdaptive algorithm mentioned in \cite{LuoWYC16} and \Cref{alg:simple} is that GKAdaptive checks only condition $(2)$ (see \Cref{alg:simple}) when it tries to delete elements.
\end{remark}

\subsection{The Space Analysis} \label{sec:greedy-space}

In this subsection, we prove a bound on the space used by \Cref{alg:simple}. Formally, we have the following:
\begin{lemma}\label{lem:greedy-space}
For any $\eps > 0$ and a stream of length $n$,~\Cref{alg:simple} maintains an $\eps$-approximate quantile summary in $O(\frac{1}{\eps} \cdot \log^2{\!(\eps n)})$ space. 
\end{lemma}

After performing the deletion step at time $t$, any element $\ei{i}$ present in $\QS$  either satisfies $\bv(\ei{i}) > \bv(\ei{i+1})$ or  $\g_i+\g_{i+1}+\Delta_{i+1} > t$; otherwise~\Cref{alg:simple} would have deleted this element. 
We refer to the elements in $\QS$ satisfying the former condition as \emph{\To} elements and the ones satisfying {only} the latter condition as \emph{\Tt} elements. Thus, each element is exactly one of the two types (except only $\ei{s}=+\infty$ which we can ignore). 

We first prove that the  number of \To elements cannot be much larger than the \Tt ones. This is simply because the band-values of consecutive \To elements strictly decreases from one element to the next and 
thus we cannot have many \To elements next to each other. 

\begin{lemma}\label{lem:type-1}
	After the deletion step at time step $t$, the number of \To elements stored in $\QS$ is $O(\log{t})$ times larger than the \Tt elements.
\end{lemma}

\begin{proof}
	Let $\ei{i},\ei{i+1},\ldots,\ei{j}$ be a sequence of \emph{consecutive} elements in $\QS$ which are all \To. By definition of \To elements, $\bv(\ei{i}) > \bv(\ei{i+1}) > \ldots > \bv(\ei{j})$. As the number of band values at time $t$ is $\Bt{t} = O(\log{t})$ by~\Cref{eq:band-alpha}, we have that length of this sequence can only be $O(\log{t})$. Thus, for every $O(\log{t})$ \To elements, there is at least one \Tt element which immediately implies the lemma. 
\end{proof}

\Cref{lem:type-1} allows us to focus solely on bounding the number of \Tt elements in the following, which is the main part of the proof in this section. 

\begin{lemma}\label{lem:type-2}
	After the deletion step at time step $t$, the number of \Tt elements in $\QS$ is $O(\ell \cdot \log{t})$. 
\end{lemma}


\begin{proof}
The first step of the proof is to simply consider any \Tt element $\ei{i}$ and lower bound $g_{i} + g_{i+1}$; intuitively, this number should be ``large'' considering
that the only reason we did not delete $\ei{i}$ was the condition (2) of the deletion step in~\Cref{alg:simple}. 
\begin{claim}\label{clm:lower-g}
After the deletion step at time step $t$, for any \Tt element $\ei{i}$,
		$ g_i + g_{i+1} \geq 2^{\bv(\ei{i+1})-1}-~2$.
\end{claim}
\begin{proof}
	As $\ei{i}$ is \Tt, the only reason for $\ei{i}$  to not be deleted from $\QS$ is $g_i + g_{i+1} + \Delta_{i+1} > t$. Moreover, recall that $\Delta_{i+1} \leq t_0(\ei{i+1})$ by~\Cref{eq:del-t0}, and hence, 
	\begin{align*}
		g_i + g_{i+1} > t - t_0(\ei{i+1}) \geq 2^{\bv(\ei{i+1})-1}-2,
	\end{align*}
	where the second inequality is by~\Cref{obs:band-length} which states that it takes exponential time in $\alpha$ for an element to reach $\Band_{\alpha}$. \Qed{\Cref{clm:lower-g}}
	
\end{proof}

\Cref{clm:lower-g} bounds the $g$-value of each \Tt element $\ei{i}$ as a function of the $g$-value and band-value of the \emph{next} element $\ei{i+1}$. Based on this, let us partition the \Tt elements 
	into $\Bt{t}$ (number of possible bands) sets $X_0,\ldots,X_{\Bt{t}}$ where for any band-value $\alpha$:
	\[
		X_{\alpha} := \set{\ei{i} \in \QS \mid \text{$\ei{i}$ is \Tt and $\bv(\ei{i+1}) = \alpha$}}.
	\] 
	Moreover, for any $\ei{i} \in X_{\alpha}$, since $\ei{i}$ is a \Tt element, 
	$\bv(\ei{i}) \leq \bv(\ei{i+1}) = \alpha$. We can use this and~\Cref{clm:lower-g} to bound the size of each $X_{\alpha}$ as follows:
	\begin{align}
		\card{X_{\alpha}} \cdot (2^{\alpha-1}-2) &\leq \sum_{\ei{i} \in X_{\alpha}} g_i + g_{i+1} \leq 2 \hspace{-15pt}\sum_{\ei{j} \in \QS \cap \Band_{\leq \alpha}} \hspace{-15pt} g_j, \label{eq:lem-type-2}
	\end{align}
	 where the second inequality is because both $\ei{i}$ and $\ei{i+1}$ belong to $\Band_{\leq \alpha}$ as argued above. The final step of the proof is then 
	to bound the sum of all $g$-values for elements in $\QS \cap \Band_{\leq \alpha}$. This can be done easily because $g$-values track  the number of covered elements and elements 
	of $\Band_{\leq \alpha}$ can  only cover elements in $\Band_{\leq \alpha}$ by the construction of the algorithm (rule $(1)$ in the deletion step) and \Cref{obs:stable}; the upper bound on the size of $\Band_{\leq \alpha}$ in~\Cref{eq:band-alpha}, then allows us to prove the following claim.

\begin{claim}\label{clm:cover}
	After the deletion step at time $t$ and for any $\alpha$, 
	$
	\sum_{\ei{i} \in \QS \cap \Band_{\leq \alpha}} g_i \leq \ell \cdot 2^{\alpha+1}.
	$ 
\end{claim}
\begin{proof}
	Fix an element $\ei{i}$ in $\QS \cap \Band_{\leq \alpha}$ and recall the definition of covered elements $C(\ei{i})$ in~\Cref{sec:inv}. In~\Cref{alg:simple}, we may only delete an element $\ei{i-1}$ 
	if $\bv(\ei{i-1}) \leq \bv(\ei{i})$ at the time of deletion. By~\Cref{obs:stable}, this continues to be the case throughout, implying that all of $C(\ei{i})$ belong to $\Band_{\leq \alpha}$ (where $\alpha$ is the current band-value of $\ei{i}$). We now have, 
	\begin{align*}
		\sum_{\ei{i} \in \QS \cap \Band_{\leq \alpha}} g_i &= \sum_{\ei{i} \in \QS \cap \Band_{\leq \alpha}} \card{C(\ei{i})} \tag{as $g_i = \card{C(\ei{i})}$ by~\Cref{obs: disjoint coverage}} \\
		&= |\bigcup_{\ei{i} \in \QS \cap \Band_{\leq \alpha}} C(\ei{i})| \tag{as $C(\ei{i})$'s for $\ei{i}$ in $\QS$ are disjoint by~\Cref{obs: disjoint coverage}} \\
		&\leq \card{\Band_{\leq \alpha}} \tag{as elements of $C(\ei{i})$'s belong to $\Band_{\leq \alpha}$} \\
		&\leq \ell \cdot 2^{\alpha+1}, \tag{by the bound in~\Cref{eq:band-alpha}} 
	\end{align*}
	finalizing the proof. \Qed{\Cref{clm:cover}} 
	
\end{proof}
By plugging the bounds in~\Cref{clm:cover} into~\Cref{eq:lem-type-2}, we have that, 
\[
	\card{X_{\alpha}} \cdot (2^{\alpha-1}-2) \leq \ell \cdot 2^{\alpha+2},
\]
which in turn implies $\card{X_{\alpha}} = O(\ell)$ when $\alpha>2$ (and for $\alpha \leq 2$ the total number of elements are at most $8 \cdot \ell = O(\ell)$ by \Cref{eq:band-alpha} anyway). Considering there are only $O(\log{t})$ choices of sets $X_{\alpha}$ (since the number of different bands is $O(\log{t})$ by~\Cref{eq:band-alpha}), 
we obtain the final result. 
\end{proof}

~\Cref{lem:type-1,lem:type-2} imply that after the deletion step at time $t$, the size of \QS is $O(\ell \cdot \log^2 t)$. 
This implies that the total size of the summary would be $O(\frac{1}{\eps} \cdot \log^2{\!(\eps n)})$ and concludes the proof of \Cref{lem:greedy-space} since $t = O(\eps n)$ and $\ell = O(1/\eps)$ (\Cref{def:time}). 
An important remark is in order.
\begin{remark}
\label{rem:greedy}
\emph{
In the space analysis, we bounded the number of \Tt elements in the summary after the deletion step by $O(\ell \cdot \log{t}) = O((1/\eps)\cdot\log{\!(\eps n)})$, which is quite efficient on is own. However, in the worst case, there can be  $O(\log t)$ \To elements for every \Tt element as shown in \Cref{fig:tightness}. Thus, \Cref{alg:simple} may end up storing as many as $O(\ell \cdot \log^2t) = O((1/\eps)\cdot\log^2{\!(\eps n)})$ \To elements in the summary, leading to its sub-optimal space requirement.   
}
\end{remark}

\begin{figure}[!ht]
\centering
\hspace{-1cm}\resizebox{1\textwidth}{!}{\tikzset{every picture/.style={line width=0.75pt}} 

\begin{tikzpicture}[x=0.75pt,y=0.75pt,yscale=-1,xscale=1]

\draw [line width=1.5]    (138.21,348.43) -- (653.11,348.43) (165.21,347.43) -- (165.21,349.43)(192.21,347.43) -- (192.21,349.43)(219.21,347.43) -- (219.21,349.43)(246.21,347.43) -- (246.21,349.43)(273.21,347.43) -- (273.21,349.43)(300.21,347.43) -- (300.21,349.43)(327.21,347.43) -- (327.21,349.43)(354.21,347.43) -- (354.21,349.43)(381.21,347.43) -- (381.21,349.43)(408.21,347.43) -- (408.21,349.43)(435.21,347.43) -- (435.21,349.43)(462.21,347.43) -- (462.21,349.43)(489.21,347.43) -- (489.21,349.43)(516.21,347.43) -- (516.21,349.43)(543.21,347.43) -- (543.21,349.43)(570.21,347.43) -- (570.21,349.43)(597.21,347.43) -- (597.21,349.43)(624.21,347.43) -- (624.21,349.43)(651.21,347.43) -- (651.21,349.43) ;
\draw [shift={(657.11,348.43)}, rotate = 180] [fill={rgb, 255:red, 0; green, 0; blue, 0 }  ][line width=0.08]  [draw opacity=0] (11.61,-5.58) -- (0,0) -- (11.61,5.58) -- cycle    ;
\draw [line width=1.5]    (138.21,199.24) -- (138.21,348.43) (139.21,223.24) -- (137.21,223.24)(139.21,247.24) -- (137.21,247.24)(139.21,271.24) -- (137.21,271.24)(139.21,295.24) -- (137.21,295.24)(139.21,319.24) -- (137.21,319.24)(139.21,343.24) -- (137.21,343.24) ;
\draw [shift={(138.21,195.24)}, rotate = 90] [fill={rgb, 255:red, 0; green, 0; blue, 0 }  ][line width=0.08]  [draw opacity=0] (11.61,-5.58) -- (0,0) -- (11.61,5.58) -- cycle    ;
\draw  [fill={rgb, 255:red, 232; green, 232; blue, 232 }  ,fill opacity=1 ] (146.52,217.68) -- (162.2,217.68) -- (162.2,226.64) -- (146.52,226.64) -- cycle ;
\draw  [fill={rgb, 255:red, 232; green, 232; blue, 232 }  ,fill opacity=1 ] (169.92,239.29) -- (185.61,239.29) -- (185.61,248.25) -- (169.92,248.25) -- cycle ;
\draw  [fill={rgb, 255:red, 232; green, 232; blue, 232 }  ,fill opacity=1 ] (219.43,281.6) -- (235.12,281.6) -- (235.12,290.56) -- (219.43,290.56) -- cycle ;
\draw  [fill={rgb, 255:red, 232; green, 232; blue, 232 }  ,fill opacity=1 ] (242.84,302.3) -- (258.52,302.3) -- (258.52,311.26) -- (242.84,311.26) -- cycle ;
\draw  [fill={rgb, 255:red, 232; green, 232; blue, 232 }  ,fill opacity=1 ] (292.35,216.78) -- (308.03,216.78) -- (308.03,225.74) -- (292.35,225.74) -- cycle ;
\draw  [fill={rgb, 255:red, 232; green, 232; blue, 232 }  ,fill opacity=1 ] (315.76,238.39) -- (331.44,238.39) -- (331.44,247.35) -- (315.76,247.35) -- cycle ;
\draw  [fill={rgb, 255:red, 232; green, 232; blue, 232 }  ,fill opacity=1 ] (365.27,280.7) -- (380.95,280.7) -- (380.95,289.66) -- (365.27,289.66) -- cycle ;
\draw  [fill={rgb, 255:red, 232; green, 232; blue, 232 }  ,fill opacity=1 ] (388.67,301.4) -- (404.36,301.4) -- (404.36,310.36) -- (388.67,310.36) -- cycle ;
\draw  [fill={rgb, 255:red, 232; green, 232; blue, 232 }  ,fill opacity=1 ] (476.89,216.78) -- (492.58,216.78) -- (492.58,225.74) -- (476.89,225.74) -- cycle ;
\draw  [fill={rgb, 255:red, 232; green, 232; blue, 232 }  ,fill opacity=1 ] (500.3,238.39) -- (515.98,238.39) -- (515.98,247.35) -- (500.3,247.35) -- cycle ;
\draw  [fill={rgb, 255:red, 232; green, 232; blue, 232 }  ,fill opacity=1 ] (549.81,280.7) -- (565.49,280.7) -- (565.49,289.66) -- (549.81,289.66) -- cycle ;
\draw  [fill={rgb, 255:red, 232; green, 232; blue, 232 }  ,fill opacity=1 ] (573.21,301.4) -- (588.9,301.4) -- (588.9,310.36) -- (573.21,310.36) -- cycle ;
\draw  [fill={rgb, 255:red, 232; green, 232; blue, 232 }  ,fill opacity=1 ] (145,376.12) -- (162.02,376.12) -- (162.02,386.03) -- (145,386.03) -- cycle ;
\draw  [color={rgb, 255:red, 92; green, 92; blue, 92 }  ,draw opacity=1 ][fill={rgb, 255:red, 130; green, 130; blue, 130 }  ,fill opacity=1 ] (268.94,324.81) -- (284.63,324.81) -- (284.63,333.77) -- (268.94,333.77) -- cycle ;
\draw  [color={rgb, 255:red, 92; green, 92; blue, 92 }  ,draw opacity=1 ][fill={rgb, 255:red, 130; green, 130; blue, 130 }  ,fill opacity=1 ] (414.78,323.91) -- (430.46,323.91) -- (430.46,332.87) -- (414.78,332.87) -- cycle ;
\draw  [color={rgb, 255:red, 92; green, 92; blue, 92 }  ,draw opacity=1 ][fill={rgb, 255:red, 130; green, 130; blue, 130 }  ,fill opacity=1 ] (599.32,323.91) -- (615,323.91) -- (615,332.87) -- (599.32,332.87) -- cycle ;
\draw  [color={rgb, 255:red, 92; green, 92; blue, 92 }  ,draw opacity=1 ][fill={rgb, 255:red, 130; green, 130; blue, 130 }  ,fill opacity=1 ] (226.41,377.11) -- (243.42,377.11) -- (243.42,387.02) -- (226.41,387.02) -- cycle ;
\draw   (100.91,219.61) .. controls (96.24,219.61) and (93.91,221.94) .. (93.91,226.61) -- (93.91,268.45) .. controls (93.91,275.12) and (91.58,278.45) .. (86.91,278.45) .. controls (91.58,278.45) and (93.91,281.78) .. (93.91,288.45)(93.91,285.45) -- (93.91,330.3) .. controls (93.91,334.97) and (96.24,337.3) .. (100.91,337.3) ;
\draw    (284.53,341.3) .. controls (313.51,358.24) and (336.51,378.24) .. (402.51,384.24) ;
\draw    (504.51,391.24) .. controls (555.51,386.24) and (584.51,358.24) .. (603.51,336.24) ;
\draw    (419.51,336.24) .. controls (422.51,367.24) and (432.51,375.24) .. (443.51,382.24) ;
\draw  [dash pattern={on 0.84pt off 2.51pt}]  (483.53,387.9) .. controls (492.07,389.91) and (547.51,368.24) .. (560.51,334.24) ;
\draw  [dash pattern={on 0.84pt off 2.51pt}]  (468.51,382.24) .. controls (490.51,380.24) and (502.51,362.24) .. (509.51,334.24) ;

\draw (422.34,255.06) node [anchor=north west][inner sep=0.75pt]  [font=\huge,xscale=1.3,yscale=1.3] [align=left] {$\displaystyle \dotsc $};
\draw (192.87,247.93) node [anchor=north west][inner sep=0.75pt]  [font=\large,rotate=-45,xscale=1.3,yscale=1.3] [align=left] {$\displaystyle \dotsc $};
\draw (341.7,246.03) node [anchor=north west][inner sep=0.75pt]  [font=\large,rotate=-45,xscale=1.3,yscale=1.3] [align=left] {$\displaystyle \dotsc $};
\draw (525.24,245.03) node [anchor=north west][inner sep=0.75pt]  [font=\large,rotate=-45,xscale=1.3,yscale=1.3] [align=left] {$\displaystyle \dotsc $};
\draw (614.83,364.8) node [anchor=north west][inner sep=0.75pt]  [font=\footnotesize,xscale=1.3,yscale=1.3] [align=left] {{\footnotesize Position}};
\draw (29.8,267.07) node [anchor=north west][inner sep=0.75pt]  [font=\footnotesize,xscale=1.3,yscale=1.3] [align=left] {{\footnotesize $\displaystyle O(\log t)$}\\{\footnotesize bands}};
\draw (384.51,395.24) node [anchor=north west][inner sep=0.75pt]  [font=\footnotesize,xscale=1.3,yscale=1.3] [align=left] {\begin{minipage}[lt]{77.67844000000001pt}\setlength\topsep{0pt}
	\begin{center}
	{\footnotesize $\displaystyle O(\log t)$ type-2 \ elements}
	\end{center}
	
	\end{minipage}};
\draw (165.56,374.16) node [anchor=north west][inner sep=0.75pt]  [font=\footnotesize,xscale=1.3,yscale=1.3] [align=left] {{\footnotesize type-1}};
\draw (247.94,374.16) node [anchor=north west][inner sep=0.75pt]  [font=\footnotesize,xscale=1.3,yscale=1.3] [align=left] {{\footnotesize type-2}};
\draw (40,189.7) node [anchor=north west][inner sep=0.75pt]  [font=\footnotesize,xscale=1.3,yscale=1.3] [align=left] {{\footnotesize Band-values}};

\end{tikzpicture}}
\caption{ Each block in the figure represents an element stored in \QS. The ranks of elements increase along the horizontal axis. The figure illustrates why \Cref{alg:simple} might end up storing $O(\ell \log^2 t)$ elements in \QS.  By \Cref{lem:type-2}, there could be as many as  $O(\ell \cdot \log t)$ \Tt elements in \QS. Each of these \Tt elements could be preceded by a sequence of $O(\log t)$ \To elements (since there are $O(\log t)$ bands). }
\label{fig:tightness}
\end{figure}
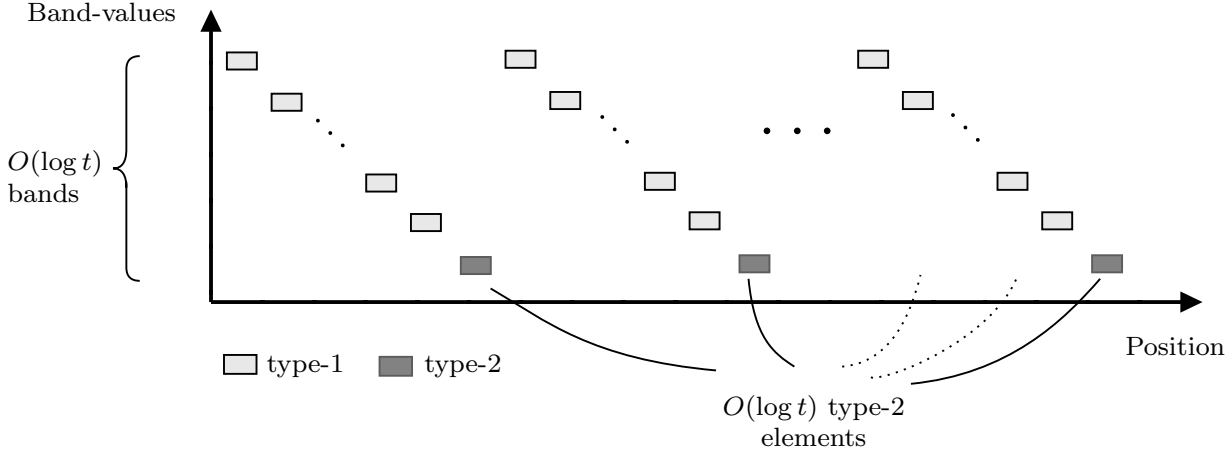

As we say in~\Cref{rem:greedy}, one source of sub-optimality of~\Cref{alg:simple} was the large number of \To elements stored in the summary compared to the \Tt ones. A way to improve this is to \emph{actively} try to decrease  the number of stored \To elements. Roughly speaking, this is done by deleting \Tt elements from the summary only if it does 
not contribute to creating a long sequence of \To elements (e.g., as in~\Cref{fig:tightness}). Note that our \Cref{alg:weighted-GK} was precisely doing this by only deleting an element only if its entire segment can be deleted along with it.

We conclude this section with the following remark; this will be helpful for us in implementing the algorithm efficiently. 

\begin{remark} [Delaying Deletions]
\label{rem: Delaying-Deletions-simple}
Suppose in \Cref{alg:simple}, instead of running the \textbf{deletion steps} in Line (ii) after each chunk, we  run them after inserting $k$ consecutive chunks for some integer $k > 1$; then, the space complexity of the algorithm only increases by an $O(k \cdot \ell)$ additive term.
\end{remark}
Note that this is simply because after running the deletion step at any time $t$, the number of elements reduces to $O(\ell \cdot \log^2 t)$ as proved earlier. Thus, the extra space is only due to storing the extra $O(k \cdot \ell)$ elements in \QS that are inserted but not deleted yet. 

\subsection{An Efficient Implementation of \texorpdfstring{\Cref{alg:simple}}{the greedy algorithm} }\label{subsec: fast-imp-simple}

In this section, we present an efficient implementation of \Cref{alg:simple}. This is very similar to the implementation of \Cref{alg:weighted-GK} described in \Cref{subsec: fast-imp-weighted-GK}. Formally, we show the following:
\begin{lemma}\label{lem:greedy-time}
For any $\eps > 0$ and a stream of length $n$, there is an implementation of \Cref{alg:simple} that takes $O\big(\log(1/\eps)+ \log \log (\eps n)\big)$ worst-case processing time per element.
\end{lemma}

\paragraph*{\textbf{Part \RNum{1}:  Storing \QS:}}
We store our summary $\QS$ as a balanced binary search tree (BST), where each node contains an element of \QS along with its metadata. For each element $e$ we store $g(e), \Delta(e)$ and $t_0(e)$. The sorting key of the BST is the value of the elements. \textbf{Insert} and \textbf{Delete} respectively, insert and delete elements from the BST.

\paragraph*{\textbf{Part \RNum{2}: Performing a Deletion Step:}} We perform a deletion step efficiently (in time proportional to the summary size) using the following algorithm.
\begin{tbox}
\textbf{Algorithm. Performing a deletion step efficiently:}
\label{tbox:deletion-step-greedy}
\begin{enumerate}
    \item Perform an inorder traversal of $\QS$ (which is a BST) to obtain a temporary (doubly-linked) list of elements sorted by value. 
    \item Compute $\bv$ of all elements of \QS using \Cref{obs: bvalue-compute}. 
    \item Traverse the list from larger elements to smaller ones. For each element $\ei{i}$, delete it from BST (as well as the list), if it satisfies both the deletion conditions mentioned in \Cref{alg:simple}. 
\end{enumerate}
\end{tbox}
  First, we give an algorithm with a fast amortized update time and then using standard techniques, show how it can also be implemented to have the same asymptotic worst-case update time.

\begin{Implementation}\label{implementation: greedy-simple}
\textbf{Efficient Implementation of}~\Cref{alg:simple}

	\begin{itemize}
	\item Initialize $\QS$ to be an empty balanced binary search tree.
	
	\item $\DeleteTime \leftarrow 2$.

	\item For each time step $t$ with arriving items $(x^{(t)}_1,\ldots,x^{(t)}_\ell)$: 
	\begin{enumerate}[label=$(\roman*)$]
		\item 
		Run $\textbf{Insert}(x^{(t)}_j)$ for each element of the chunk. 
		\item If $(t = \DeleteTime)$:
		\begin{itemize}
		    \item 	Execute the deletion step and update $\DeleteTime \leftarrow \DeleteTime + \ceil{\log^2t}$.
		\end{itemize}
	
	\end{enumerate}
	\end{itemize}
\end{Implementation}

\paragraph*{Space Analysis.} 
The space complexity of \Cref{implementation: greedy-simple} is $O(\frac{1}{\eps} \log^2(\eps n))$. This is simply because, after the deletion step at any time $t$, we delay the deletion by $\ceil{\log^2 t}$ time steps. This, by~\Cref{rem: Delaying-Deletions-simple}, implies that the space used only increases by an additive term of $O(\ell \log ^2 t)=O(\frac{1}{\eps} \log^2(\eps n))$, as $t = O(\eps n)$ and $\ell = O(1/\eps)$ (\Cref{def:time}).
 
 \paragraph*{Time Analysis.}
Since \QS is stored as a BST, performing an $\textbf{Insert}$ or $\textbf{Delete}$ operation on \QS takes $O(\log s)$, where $s$ is the size of \QS. By the space analysis presented above, $s = O((1/\eps) \cdot \log^2{\!(\eps n)})$. This leads to the following observation:

\begin{observation} \label{obs: ins del time greedy}
Over a stream of length $n$, the total time taken by \Cref{implementation: greedy-simple} to perform all \textbf{Insert} and \textbf{Delete} operations is $O( n \cdot (\log (1/\eps) + \log \log (\eps n)))$. 
\end{observation}

The observation follows from the fact that each element is inserted and deleted at most once from \QS and each insertion or deletion takes $O(\log s) = O\big( \log (1/\eps)+ \log \log (\eps n) \big)$ time. The only time taken by \Cref{implementation: greedy-simple}, \textbf{not} taken into account in~\Cref{obs: ins del time greedy} is the part that determines \emph{which elements} to delete, which we will bound below.

\begin{claim}\label{lem: time to decide elements to delete-greedy}
Over a stream of length $n$, the total time taken by \Cref{implementation: greedy-simple} to decide which elements need to be deleted over all the executed deletion steps is $O(n)$. 
\end{claim}

\begin{proof}

The time taken to decide which elements need to be deleted inside a deletion step is $O(s)=O(\frac{1}{\eps} \log^2(\eps n))$. This is because: creating a linked list, followed by computation of $\bv$ of all elements can be computed in $O(s)$ time. Then, making a linear pass over the list from the largest to the smallest element (to check if the deletion conditions hold) requires $O(s)$ time.

Thus, we now focus on bounding the number of deletion steps that we may execute over the entire stream. Formally,
\begin{claim}\label{clm:number-of-merges }
Over a stream of length $n$, the number of deletion steps performed by \Cref{implementation: greedy-simple} is $O(\eps n/ \log^2 (\eps n))$.
\end{claim}
\begin{proof}
We partition all time steps into intervals of type $[2^i, 2^{i+1})$, for $1\leq i\leq \ceil{\log (\eps n)}$. Since we wait for $\ceil{\log^2 t }$ time steps after performing the deletion step at time step $t$, there are at least $i^2$ time steps between two consecutive deletion steps that happen in the time interval $[2^i, 2^{i+1})$. Given that the length of the interval is $2^i$, the number of times we perform a deletion step in the time interval $[2^i,2^{i+1})$, is at most $1+\frac{2^i}{i^2}$. Thus,
\[
\# \text{ deletion steps performed} \leq \sum_{i=1}^{\ceil{\log(\eps n)}}1+\frac{2^i}{i^2}=O\bigg(\,\frac{\eps n}{ \log^2 (\eps n)}\,\bigg). \Qed{\Cref{clm:number-of-merges }}
\] 
By \Cref{clm:number-of-merges }, we perform $O\big(\frac{\eps n}{\log^2(\eps n)}\big)$ deletion steps, and in each one we spend $O\big(\frac{1}{\eps} \log^2(\eps n)\big)$ time to decide which elements need to be deleted, and the lemma follows.
\end{proof} 

\end{proof}
\Cref{obs: ins del time greedy,lem: time to decide elements to delete-greedy} clearly imply that the total time taken by \Cref{implementation: greedy-simple} over a stream of length $n$ is $O\big(n \cdot(\log(1/\eps)+\log\log(\eps n)\big)$. Thus, the amortized update time per element is $O\big(\log(1/\eps)+\log\log(\eps n)\big)$.

\paragraph*{Worst-case update time.}We now look at the worst-case update time for any element. We might have to delete $O(s)$ elements from \QS in the worst case when a deletion step is performed which would take time $O(s \log s)$. To reduce the worst-case update time, we propose a minor modification to \Cref{implementation: greedy-simple}. We notice that a deletion step is next called after $\ceil{\log^2 t}$ time steps. What we do to reduce the worst-case time is spread the time it takes to perform a deletion step over all elements before the next deletion step. Formally, we have the following:
\begin{claim} \label{lem: worst case greedy}
There is an implementation of \Cref{alg:simple} that has worst-case update time $O( \log \frac{1}{\eps} + \log \log (\eps n))$ per element, without affecting its asymptotic space complexity.
\end{claim}
\begin{proof}
We know that inserting an element can be done in $O(\log s)$ time since \QS is a BST. Thus, the only step that takes more than $O(\log s)$ time is performing a deletion step. Let $t$ be a time step when a deletion step is performed. The next deletion step is at time step $t + \ceil{\log^2 t}$. So we perform a deletion step uniformly across $\ceil{\log^2 t}/2$ time steps and store the incoming elements of the stream in a buffer $B$. The update time for each element is $O(\log s)$ since $O(s)$ elements are inserted in $B$, and the deletion step takes $O(s \log s)$ time. Note that we equally space out the insertions, so the worst-case update time per element is $O(\log s)$. Now in the remaining $\ceil{\log^2 t}/2$ time steps, we insert two elements per time step which takes $O(\log s)$ worst-case update time per element. Also, at time $t + \ceil{\log^2 t}$, the buffer $B$ is empty since $2 \cdot \ceil{\log^2 t}/2 = \ceil{\log^2 t}$ elements are inserted in \QS. Thus, we get that the worst case update time for any element is $O(\log s)= O( \log \frac{1}{\eps} + \log \log \eps n)$.
Moreover, $B$ has size $O(s)$, thus the space only increases by a constant factor.
\end{proof}

This concludes the proof of \Cref{lem:greedy-time}. We now show that the time to answer quantile queries is $O(\log s)$.

\begin{claim}\label{clm:quantile_queries-greedy-unwt}
	The summary created by \Cref{implementation: greedy-simple}, after using $O(s)$ pre-processing time, can answer a sequence of quantile queries using $O(\log s)$ time per query where $s$ is the size of the summary.
\end{claim}
\begin{proof}
	Using \Cref{eq:g-delta}, the $\rmin$ and $\rmax$ values of all elements in $\QS$ can be computed by a single inorder traversal. This is the pre-processing step which requires $O(s)$ time. To answer a $\phi$-quantile query, it is sufficient to find an element whose $\rmin$ and $\rmax$ values both lie in the interval $[(\phi - \eps)n, (\phi + \eps)n]$. Such an element is guaranteed to exist by \Cref{lem:search} and can be found in $O(\log s)$ time as $\QS$ is stored as a BST sorted by value of elements and hence by $\rmin$ values.
\end{proof}
\Cref{clm:quantile_queries-greedy-unwt} along with \Cref{lem:greedy-space} and \Cref{lem:greedy-time} proves \Cref{thm:greedy} since $s=O(\frac{1}{\eps} \cdot \log^2{\!(\eps n)})$.

\section{The Final \texorpdfstring{$O(\frac{1}{\eps} \cdot \log{(\eps  n)})$ Size}{} Summary}\label{sec:final}


 In this section, we give our description of GK summaries. As we discussed in \Cref{rem:greedy}, in this algorithm, we actively try to avoid creation of long sequences of \To elements (see \Cref{fig:tightness}). Roughly speaking, while there is an element whose deletion \emph{together} with its entire segment doesn't violate \Cref{inv:g-delta} (and the same condition on $\bv$s), we delete the element and its entire segment. Let $\gs_i$ denote the sum of $\g$-values of elements in $\Seg(\ei{i})$. Below is a formal description of the algorithm.

\begin{Algorithm}\label{alg:GK}
	An improved algorithm for updating the quantile summary. 
	\medskip

	\medskip 
	For each time step $t$ with arriving items $(x^{(t)}_1,\ldots,x^{(t)}_\ell)$: 
	\begin{enumerate}[label=$(\roman*)$]
	\addtolength{\itemindent}{3mm}
		\item Run $\textbf{Insert}(x^{(t)}_j)$ for each element of the chunk.
		
		    \item While there exists an element $\ei{i}$ in \QS satisfying:  	
		\[
		(1)~\bv(\ei{i}) \leq \bv(\ei{i+1})\qquad \text{\underline{and}} \qquad (2)~\gs_i+\g_{i+1}+\Delta_{i+1} \leq t 
		\]
		run $\textbf{Delete}(\ei{k})$ for $\ei{k}$ in $\set{\ei{i}} \cup \Seg(\ei{i})$.
	\end{enumerate}
\end{Algorithm}
The main difference between \Cref{alg:simple} and \Cref{alg:GK} is that in the latter, we consider an element and its segment as ``one unit" when trying to delete them, i.e., we either delete an element together with its entire segment or we do not delete that element at all (note that it is possible for elements of this segment to be deleted on their own). We shall note that even though our description of~\Cref{alg:GK} varies from the presentation of GK summaries in~\cite{GreenwaldK01}, the two algorithms behave in an 
almost identical way. 

\begin{theorem}\label{thm:GK}
	For any $\eps > 0$ and a stream of length $n$,~\Cref{alg:GK} maintains an $\eps$-approximate quantile summary in $O(\frac{1}{\eps} \cdot \log{\!(\eps n)})$ space. Also, there is an implementation of~\Cref{alg:GK} that takes $O(\log(\frac{1}{\eps})+\log \log(\eps n))$ worst case update time per element.
\end{theorem}

\Cref{alg:GK} maintains~\Cref{inv:g-delta} since it may delete an element $\ei{i}$ (along with $\Seg(\ei{i}))$ only if $\gs_i + g_{i+1} + \Delta_{i+1} \leq t$, which by~\Cref{obs:g-del} implies that $g_{i+1} + \Delta_{i+1} \leq t$ \emph{after} the deletion (and all other $(g,\Delta)$-values remain
unchanged). As argued in~\Cref{sec:inv}, maintaining~\Cref{inv:g-delta} directly implies that $\QS$ is an $\eps$-approximate quantile summary throughout the stream. Thus, we focus on bounding the size of $\QS$ under this algorithm in~\Cref{subsec:GK-space}, and giving an efficient implementation of the same in \Cref{subsec: fast-imp-GK} to prove \Cref{thm:GK}.

\subsection{Space Analysis}
\label{subsec:GK-space}

In this subsection, we prove a bound on the space used by \Cref{alg:GK}. Formally, we have the following:
\begin{lemma}\label{lem:GK-space}
	For any $\eps > 0$ and a stream of length $n$,~\Cref{alg:GK} maintains an $\eps$-approximate quantile summary in $O(\frac{1}{\eps} \cdot \log{\!(\eps n)})$ space. 
\end{lemma}

Recall that a simple property of~\Cref{alg:simple} was that the elements in $\Band_{\leq \alpha}$ could only cover the elements from $\Band_{\leq \alpha}$ as well. We argue that this continues to be the case in~\Cref{alg:GK}.  

\begin{observation}\label{obs: deleted-band-inv}
Elements from $\Band_{\leq \alpha}$ in \QS only cover elements of $\Band_{\leq \alpha}$ at any time.
\end{observation}
This is because: when $\ei{i}$ and $\Seg(\ei{i})$ get covered by $\ei{i+1}$, \Cref{alg:GK} ensures that $\bv(\ei{i})\leq \bv(\ei{i+1})$. From \Cref{def:segment}, $\Seg(\ei{i})$ contains elements with $\bv$ less than $\bv(\ei{i})$. Thus, $C(\ei{i+1})$ contains elements with $\bv$ at most $\bv(\ei{i+1})$ and this continues to be the case at a later time by \Cref{obs:stable}. 

Having made the above observation, we now proceed to show bounds on the size of $\QS$. We borrow the definitions of \To and \Tt elements from \Cref{sec:greedy-space} with some modifications. After the deletion step at time $t$, an element $\ei{i}$ in $\QS$ is called a \emph{\To} element if it satisfies the condition $\bv(\ei{i}) > \bv(\ei{i+1})$ and is called a \emph{\Tt} element if it only satisfies the condition $\gs_i+\g_{i+1}+\Delta_{i+1} > t$.
After the deletion step, any element in \QS satisfies the first or the second condition (except the last element $\ei{s}=+\infty$); otherwise~\Cref{alg:GK} would have deleted this element. 

We first bound the number of \To elements. Unlike \Cref{alg:simple}, here we can show that the number of \To elements is $O(\ell \cdot \log t)$ at any time $t$, which is a factor of $O(\log t)$ smaller than before; this is exactly what improves the space of 
\Cref{alg:GK}. The following lemma is the heart of the proof. 

\begin{lemma}\label{lem:type-1-GK}
	After the deletion step at time $t$, the number of \To elements in \text{\QS} is $O(\ell \cdot \log t)$.
\end{lemma}
\begin{proof}
Let us partition the \To elements into $\Bt{t}$ sets $Y_0,Y_1,\ldots,Y_{{\Bt{t}}-1}$ where for any band-value $\alpha$:
	\[
		Y_{\alpha} := \set{\ei{i} \in \QS \mid \text{$\ei{i}$ is \To and $\bv(\ei{i+1}) = \alpha$}};
	\] 
	(notice that elements in $Y_{\alpha}$ are such that band-value of their next element is $\alpha$, not themselves; see~\Cref{foot1}). This partitioning is analogous to the one done in \Cref{lem:type-2} but now for \To elements. We will show that the size of any set $Y_\alpha$ is $O(\ell)$. To do this, we map each element $\ei{i}$ of $Y_{\alpha}$ to the first element $\ei{j}> \ei{i}$ with $\bv$ greater than $\alpha$; see \Cref{fig:type-1-notation} for an illustration. Let $T_{\alpha}$ be the set of all such elements $\ei{j}$. For any $\ei{j}$ in $T_{\alpha}$, we can say that
\begin{equation} \label{eq: lem To terminal}
\gs_{j-1} + \g_{j} + \Delta_j > t,
\end{equation}
since $e_j$ was the \emph{first} element with $\bv$ greater than $\alpha$ (and thus $e_{j-1}$ is a \Tt element).
It is easy to observe that no two elements in $Y_\alpha$ get mapped to the same element in $T_{\alpha}$. Thus, $\card{Y_{\alpha}}=\card{T_{\alpha}}$ and thus it is enough to upper bound $\card{T_{\alpha}}$.

Since $\bv(\ei{j})$ is greater than $\bv(\ei{i+1})=\alpha$, \Cref{obs:stable} ensures that $\ei{j}$ was present in \QS at time $t^\prime=t_0(\ei{i+1})$. Let ${\g}^{\prime}_j$ be the $\g$-value of $\ei{j}$ at $t^\prime$. Since $\Delta_j$ is not a function of time, \Cref{inv:g-delta} at time $t^\prime$ implies: 
\begin{equation} \label{eq: lem To intertion time}
\g^\prime_j+\Delta_j \leq t^\prime.
\end{equation}
Subtracting \Cref{eq: lem To intertion time} from \Cref{eq: lem To terminal} and using the bounds from~\Cref{obs:band-length} we conclude that,
\begin{equation}\label{eq: gk-To-main}
\gs_{j-1}+ (\g_j-\g^\prime_j)>t-t^\prime\geq2^{\alpha-1}-2.
\end{equation}

\noindent
 In the above equation: 
 \begin{enumerate}[label=$(\roman*)$]
 \item The term $(\g_j-\g^\prime_j)$ counts the elements covered by $\ei{j}$ after time $t^{\prime}$ by \Cref{obs: disjoint coverage}. We will show in \Cref{claim: coverage-time} that after time $t'$, $\ei{j}$ covers only elements from $\Band_{\leq \alpha}$. 
 
 \item The term $\gs_{j-1}$ counts the elements covered by $\ei{j-1}$ and $\Seg(\ei{j-1})$.  By \Cref{def:segment}, $\ei{j-1}$ and all elements in $\Seg(\ei{j-1})$ have $\bv\leq \alpha$. \Cref{obs: deleted-band-inv} allows us to conclude that the elements counted by $\gs_{j-1}$ are in $\Band_{\leq \alpha}$ as well. 
 
 \item Finally, by \Cref{obs: disjoint coverage},  for distinct $\ei{j_1}$ and $\ei{j_2}$ in $T_{\alpha}$, the set of elements covered by $\ei{j_1-1}$ and its segment is disjoint from the set of elements covered by $\ei{j_2-1}$ and its segment (as can be observed in \Cref{fig:type-1-segment}).
 \end{enumerate}
This implies that the LHS of~\Cref{eq: gk-To-main}, summed over all $T_{\alpha}$, is proportional to the total number of elements in $\Band_{\leq \alpha}$ which is $O(\ell \cdot 2^{\alpha})$ by~\Cref{eq:band-alpha}. Formally, 
   \begin{align}
		\card{T_{\alpha}} \cdot (2^{\alpha-1}-2) &\leq \sum_{\ei{j} \in T_{\alpha}} \gs_{j-1} + \sum_{\ei{j} \in T_{\alpha}} (\g_j-\g^\prime_j) \leq \card{\Band_{\leq \alpha}}+\card{\Band_{\leq \alpha}} \leq \ell \cdot 2^{\alpha+2}.
	\end{align}
This implies that $\card{T_{\alpha}}=\card{Y_{\alpha}}=O(\ell)$ as desired. 
\begin{claim}\label{claim: coverage-time}
 All the elements covered by $\ei{j}$ after $t^\prime$ have $\bv$ at most $\alpha$ at time $t$.\footnote{Note that \Cref{claim: coverage-time} does not hold for \Cref{alg:simple} and thus we can not carry out a similar analysis for it. The reason for the better space bound here is deleting the segment along with the element.}
\end{claim}
\begin{proof}
Let us assume that there exists an element with $\bv>\alpha$ which gets covered by $\ei{j}$ after time $t^\prime$. 
All such elements are less than $\ei{j}$ and greater than $\ei{i+1}$ since they get covered by $\ei{j}$.
Now consider the smallest such element $e$. Clearly, $\ei{i+1}$ belongs to the segment of $e$ at $t^\prime$. Since $e$ does not belong to the summary at time $t$, it must have been deleted. This implies that its segment, which contained $\ei{i+1}$,  got deleted. This means that $\ei{i+1}$ is also deleted, a contradiction. \Qed{\Cref{claim: coverage-time}}

\end{proof}

We showed that for any $\alpha$, $\card{Y_{\alpha}}=O(\ell)$. Since there are $O(\log t)$ possible values of $\alpha$ (\Cref{eq:band-alpha}), there are $O(\ell \cdot \log t)$ \To elements, concluding the proof of \cref{lem:type-1-GK}. \qedhere
\end{proof}

Now we will bound the number of \Tt elements in \QS. The analysis is almost identical to the proof of~\Cref{lem:type-2} and is provided for completeness.

\begin{lemma}\label{lem:type-2-GK}
	After the deletion step at time $t$, the number of \Tt elements in \text{\QS} is $O(\ell \cdot \log t)$.
\end{lemma}
\begin{proof}
	Any \Tt element $\ei{i}$ in \QS, has the property that $\gs_i+\g_{i+1}+\Delta_i>t$. We will use this fact to conclude that the elements $\ei{i}$, its segment and $\ei{i+1}$ cover a ``large" number of elements. The next claim is an analog of~\Cref{clm:lower-g} of \Cref{alg:simple}.

	\begin{claim}\label{clm:lower-g-GK}
		After the deletion step at time $t$, for any \Tt element $\ei{i}$, 
		$ \gs_i + \g_{i+1} \geq 2^{\bv(\ei{i+1})-1}~-~2$.
	\end{claim}
	\begin{proof}
		As $\ei{i}$ is a \Tt element, $\gs_i + \g_{i+1} + \Delta_{i+1} > t$. By~\Cref{eq:del-t0}, $\Delta_{i+1} \leq t_0(\ei{i+1})$  and therefore, 
	\begin{align*}
	\gs_i + \g_{i+1} > t - t_0(\ei{i+1}) \geq 2^{\bv(\ei{i+1})-1}-2,
	\end{align*}
	where the second inequality is by~\Cref{obs:band-length}. \Qed{\Cref{clm:lower-g-GK}}
	
\end{proof}
\Cref{clm:lower-g-GK} gives us a lower bound on the $\gs$-value of each \Tt element $\ei{i}$ as a function of the $g$-value and band-value of the \emph{next} element $\ei{i+1}$. Therefore, we partition the \Tt elements 
into $\Bt{t}$ (number of possible bands) sets $X_0,X_1,\ldots,X_{\Bt{t}}$ where for any band-value $\alpha$:
\[
X_{\alpha} := \set{\ei{i} \in \QS \mid \text{$\ei{i}$ is \Tt and $\bv(\ei{i+1}) = \alpha$}}.
\] 
Moreover, for any $\ei{i} \in X_{\alpha}$, since $\ei{i}$ is a \Tt element, 
	$\bv(\ei{i}) \leq \bv(\ei{i+1}) = \alpha$.  We sum up the inequality of~\Cref{clm:lower-g-GK} for each element in $X_\alpha$ to obtain:
	\begin{align} \label{eq:lem-type-2-GK}
		\card{X_{\alpha}} \cdot (2^{\alpha-1}-2) &\leq \sum_{\ei{i} \in X_{\alpha}} \gs_i + g_{i+1}.
	\end{align}
We now try to upper bound the summations in the above equation by some function of $\alpha$.

\begin{claim} \label{clm: gs-bound-GK}
After the deletion step at time $t$, for any $\alpha \geq 0$, 
	$\sum \limits_{\ei{i} \in X_{\alpha}} \gs_i  \hspace{3mm} \leq \hspace{3mm} 2 \hspace{-15pt}\sum \limits_{\ei{j} \in \QS \cap \Band_{\leq \alpha}} \hspace{-15pt} \g_j.$
\end{claim}
\begin{proof}


 There is no direct way to bound this summation as the segments of elements in $X_{\alpha}$ may overlap with each other. To overcome this issue, we partition $X_\alpha$ into two disjoint sets $X_\alpha \cap \Band_{\alpha}$ and $X_\alpha \cap \Band_{\leq \alpha-1}$ such that two elements from any  one of these sets have disjoint segments. 

Any two distinct elements from $X_{\alpha} \cap \Band_\alpha$ have disjoint segments. This is because their segments, which contain elements from bands less than $\alpha$, are separated by at least one element from $\Band_{\alpha}$. Therefore,
	\[ \sum \limits_{\ei{i} \in X_{\alpha} \cap \Band_{\alpha}} \hspace{-10pt} \gs_i  \hspace{15pt}\leq \sum \limits_{\ei{j} \in \text{\QS} \cap \Band_{\leq \alpha}} \hspace{-15pt}\g_j .\]
	
	By definition of $X_\alpha$, between any two distinct elements in $X_{\alpha} \cap \Band_{\leq \alpha-1}$, there is at least one element from $\Band_\alpha$. Thus, the elements of $X_{\alpha} \cap \Band_{ \leq \alpha-1}$ also  have disjoint segments containing elements from $\Band{\leq \alpha}$. Therefore,
	\[ 
	\sum \limits_{\ei{i} \in X_{\alpha} \cap \Band_{\leq \alpha -1}} \hspace{-15pt} \gs_i  \hspace{15pt}\leq \sum \limits_{\ei{j} \in \text{\QS} \cap \Band_{\leq \alpha}} \hspace{-15pt}\g_j, 
	\]
implying the claim. \Qed{\Cref{clm: gs-bound-GK}}

\end{proof}
The following claim bounds the sum of $g$-values of the elements in \QS from $\Band_{\leq \alpha}$ by $\ell \cdot 2^{\alpha+1}$. Its proof is 
identical to \Cref{clm:cover} and is thus omitted. 
\begin{claim}\label{clm:cover-gk}
	At any time $t$ and for any $\alpha$, 
	$
	\sum_{\ei{i} \in \QS \cap \Band_{\leq \alpha}} g_i \leq \ell \cdot 2^{\alpha+1}.
	$ 
\end{claim}

\noindent
By plugging the bounds of \Cref{clm: gs-bound-GK} and \Cref{clm:cover-gk} in \Cref{eq:lem-type-2-GK}, we have that,
\[
	\card{X_{\alpha}} \cdot (2^{\alpha-1}-2)\leq \sum_{\ei{i} \in X_{\alpha}} \gs_i \hspace{5pt}+ \hspace{-10pt}\sum_{\ei{j} \in \QS \cap \Band_{\leq \alpha}} \hspace{-15pt} g_j \hspace{3pt} \leq 3 \hspace{-15pt}\sum_{\ei{j} \in \QS \cap \Band_{\leq \alpha}} \hspace{-15pt} g_j \leq  3\cdot \ell \cdot 2^{\alpha+1},
\]
which implies $\card{X_{\alpha}} = O(\ell)$. There are $O(\log t)$ sets $X_{\alpha}$ since there are $O(\log t)$ bands from \Cref{eq:band-alpha}, thus the lemma follows.
 \end{proof}

We have now shown that, after performing the deletion step at time $t$, the number of \To elements in \QS is $O(\ell \cdot \log t)$ by \Cref{lem:type-1-GK} and the number of \Tt elements in \QS is $O(\ell \cdot \log t)$ by \Cref{lem:type-2-GK}. Since each element in \QS (other than $+\infty$) 
is either \To or \Tt, the total number of elements in \QS is $O(\ell \cdot \log t)$.


This finalizes the proof of \Cref{lem:GK-space} since $t=O(\eps n)$ and $\ell=O(\frac{1}{\eps})$ (\Cref{def:time}). We conclude the discussion of the space complexity with the following remark; this is analogous to \Cref{rem: Delaying-Deletions-simple} for~\Cref{alg:simple}.

\begin{remark} [Delaying Deletions]
\label{rem: Delaying-Deletions-GK}
Suppose in \Cref{alg:GK}, instead of running the \textbf{deletion steps} in Line (ii) after each chunk, we run them after inserting $k$ consecutive chunks for some integer $k > 1$; then, the space complexity of the algorithm only increases by an $O(k \cdot \ell)$ additive term.
\end{remark}
The argument is the same as the one in \Cref{rem: Delaying-Deletions-simple}; after running the deletion step at any time $t$, the number of elements reduces to $O(\ell \cdot \log t)$ as proved earlier. Thus, the extra space is only due to storing the additional $O(k \cdot \ell)$ elements that are inserted in \QS.

\subsection{An Efficient Implementation of \texorpdfstring{\Cref{alg:GK}}{the algorithm} }\label{subsec: fast-imp-GK}

We use a similar strategy as in \Cref{subsec: fast-imp-simple} to implement \Cref{alg:GK} efficiently.
Formally, we show the following:
\begin{lemma}\label{lem:GK-time}
	For any $\eps > 0$ and a stream of length $n$, there is an implementation of \Cref{alg:GK} that takes $O(\log(\frac{1}{\eps})+\log \log(\eps n))$ worst case update time per element.
\end{lemma}

\paragraph*{\textbf{Part \RNum{1}:  Storing \QS:}}
The summary \QS is stored as a balanced binary search tree, exactly as in  \Cref{implementation: greedy-simple}.

\paragraph*{\textbf{Part \RNum{2}: Performing a Deletion Step:}} The way in which the deletion step is performed will also be very similar to the fast implementation of \Cref{alg:simple}. The only change is that the deletion conditions of \Cref{alg:GK} will be checked while deciding which elements to delete. This, however, requires the computation of $\gs$ values of elements of the summary which can be done similar to the computation of $\Gs$ value shown in \Cref{subsec:comp_gs_val}.
We now present an implementation of \Cref{alg:GK} with fast amortized update time per element. Using techniques, described in \Cref{subsec: fast-imp-weighted-GK} the same bound can be shown on the worst case update time per element.

\begin{Implementation}\label{implementation: gk}
\textbf{Efficient Implementation of}~\Cref{alg:GK}

	\begin{itemize}
	\item Initialize $\QS$ to be an empty balanced binary search tree.
	
	\item $\DeleteTime \leftarrow 2$.

	\item For each time step $t$ with arriving items $(x^{(t)}_1,\ldots,x^{(t)}_\ell)$: 
	\begin{enumerate}[label=$(\roman*)$]
		\item 
		Run $\textbf{Insert}(x^{(t)}_j)$ for each element of the chunk. 
		\item If $(t = \DeleteTime)$:
		\begin{itemize}
		    \item 	Execute the deletion step and update $\DeleteTime \leftarrow \DeleteTime + \ceil{\log t}$.
		\end{itemize}
	
	\end{enumerate}
	\end{itemize}
\end{Implementation}
\paragraph*{Space Analysis} 
The space complexity of the fast implementation is still $O(\frac{1}{\eps} \log(\eps n))$. This is simply because after the deletion step at any time $t$, we delay the deletion by $\ceil{\log t}$ time steps. This, by~\Cref{rem: Delaying-Deletions-GK}, implies that the space used only increases by an additive term of $O(\ell \log  t)=O(\frac{1}{\eps} \log(\eps n))$, as $t = O(\eps n)$ and $\ell = O(1/\eps)$ (\Cref{def:time}).

\paragraph*{Time Analysis} The arguments used to bound the time complexity of the implementation described  will be very similar to the ones in \Cref{subsec: fast-imp-simple}. 

We use the fact that $\QS$ is a BST with $O(\frac{1}{\eps} \log (\eps n)$ elements to make the following observation.
\begin{observation} \label{obs: ins del time gk}
Over a stream of length $n$, the total time taken by \Cref{implementation: gk} to perform all \textbf{Insert} and \textbf{Delete} operations is $O( n \cdot (\log (1/\eps) + \log \log (\eps n)))$. Moreover, this is the total runtime of \Cref{implementation: gk}, excluding the time required to decide which elements to delete in the deletion steps.
\end{observation}

Next, we bound the time spent across all deletion steps in deciding which elements to delete. This will complete our analysis.
\begin{claim}\label{lem: time to decide elements to delete-gk}
Over a stream of length $n$, the total time taken by \Cref{implementation: gk} to decide which elements need to be deleted over all the executed deletion steps is $O(n)$. 
\end{claim}

\begin{proof}

It is easy to verify that the time taken to decide which elements need to be deleted inside a deletion step is $O(s)=O(\frac{1}{\eps} \log(\eps n))$. Therefore, it is enough to bound  the number of deletion steps that we may have to execute over the entire stream. Formally,
\begin{claim}
\label{clm: number-of-merges-gk}
Over a stream of length $n$, the number of deletion steps performed by \Cref{implementation: gk} is $O\Big(\eps n/ \log (\eps n)\Big)$.
\end{claim}
The proof of the claim is omitted as it is very similar to that of \Cref{clm:number-of-merges }.
\Qed{\Cref{lem: time to decide elements to delete-gk}}

\end{proof}
\Cref{obs: ins del time gk,lem: time to decide elements to delete-gk} imply that the total time taken by \Cref{implementation: gk} over a stream of length $n$ is $O\Big(n \cdot\big(\log(1/\eps)+\log\log(\eps n)\big)\Big)$. Thus, the amortized update time per element is $O\big(\log(1/\eps)+\log\log(\eps n)\big)$. This implementation can be modified to decrease the worst case per element processing time by spreading out the deletions over many steps. Formally, we have:
\begin{claim}\label{lem: worst case time gk}
There is an implementation of \Cref{alg:GK} that has worst case update time $O( \log \frac{1}{\eps} + \log \log (\eps n))$ per element, without affecting its asymptotic space complexity.
\end{claim}
 The strategy is to spread the deletions of a deletion step over the time steps that occur before the next deletion step.
 This concludes the proof of \Cref{lem:GK-time}.

\section{Basic Setup (Weighted Setting)}\label{sec: weighted-prelims}

We now present the basic setup of our quantile summary and preliminary definitions for the 
weighted setting.
We start with an alternate equivalent formulation of the problem defined in \Cref{def:weighted-summary} in terms of the unweighted quantiles problem for which we first define the notion of \emph{unfolding} streams.

\paragraph*{Unfolding Streams.}For the weighted stream $\swt$, we define its corresponding unfolded stream $\suf$ to be the stream which contains $\wt(x_i)$ copies of $x_i$ for  $1 \leq i \leq n$. More explicitly,
\begin{align*}
    \suf := \langle \cpy{x_1}{1}, \cpy{x_1}{2}, \dots, \cpy{x_1}{\wt(x_1)}, \dots \cpy{x_n}{1}, \cpy{x_n}{2}, \dots, \cpy{x_n}{\wt(x_n)} \rangle
\end{align*}
where $\cpy{x_i}{j}$ is the $j$-th copy of element $x_i$.

It is easy to verify that the goal of the problem, as stated in \Cref{def:weighted-summary}, is equivalent to creating an $\eps$-approximate quantile summary of $\suf$. Note that  to break ties while assigning ranks to equal elements, we will assume that elements that appeared earlier in $\suf$ have lower ranks. 
 As a side note we would like to point out here that although the algorithm we present does not ``unfold'' the stream, we will continue working with $\suf$ to present the analysis of the algorithm.

We use $\WQS$ to denote the summary of $\swt$ that our algorithm creates. $\WQS$ will consist of a subset of the elements of the stream along with some auxiliary metadata about the stored elements.   We use $e_i$  to denote the $i$-th largest element of the stream stored in $\WQS$.  We use $\cpy{e_i}{j}$ to refer to the $j$-th copy of $e_i$ in $\suf$, for $1\leq j\leq \wt(e_i)$. We also use $e$ to refer to an arbitrary element of the summary (when the rank is not relevant). The number of elements of the stream stored in $\WQS$ shall be denoted by $s$.
For each element $e$, $\WQS$ stores $\wt(e)$.  The other main information we store for each element $e$ are its $\rmin$ and $\rmax$ values, which we now define:
\begin{itemize}
\item $\rmin(e)$ and $\rmax(e)$: are lower and upper bounds maintained by \WQS on the rank of $\cpy{e}{1}$ (the first copy of $e$ to appear in $\suf$). Since we are not storing all elements, we cannot determine the exact rank of a stored element, and thus focus on maintaining proper lower and upper bounds.
\end{itemize}
 To handle corner cases that arise later, we assume that $\WQS$ contains a sentinel element $e_0$ and define $\rmin(e_0) =  \rmax(e_0) = 0$ and $\wt(e_0) = 1$. Also, we insert a $+\infty$ element at the start of the stream which is considered larger than any other element and store it in $\WQS$ as $\ei{s}$.
The $\rmin$ and $\rmax$ of this element is also always equal to the weight of inserted elements (including itself). Since $+\infty$ is the largest element, inserting it in $\swt$ does not affect the rank of any other element.

\begin{observation}
$(\rmin(e)+j-1)$ and $(\rmax(e)+j-1)$ are lower and upper bounds on the rank of $\cpy{e}{j}$. 
\end{observation}
During the stream, we insert and delete elements from the summary. This changes the rank of the elements so we have to update \WQS to reflect the changes. The procedure used to update the r-min and r-max values of elements is described below: 
\begin{tbox}
	\textbf{Insert($x, \wt(x)$).} Inserts a given element $x$ with weight $\wt(x)$ into $\WQS$. 
	
	\begin{enumerate}[label=$(\roman*)$]
	\addtolength{\itemindent}{3mm}
	    \item Store the element $x$ along with its weight $\wt(x)$ in $\WQS$. 
		\item Find the smallest element $\ei{i}$ in $\WQS$ such that $\ei{i} > x$;
		\item Set $\rmin(x) = \rmin(\ei{i-1})+ \wt(\ei{i-1})$ and $\rmax(x) = \rmax(\ei{i})$; moreover, increase $\rmin(\ei{j})$ and $\rmax(\ei{j})$ by $\wt(x)$ for all $j \geq i$. 
	\end{enumerate}
	
	\textbf{Delete($\ei{i}$).} Deletes the element $\ei{i}$ from $\WQS$. 
	
	\begin{enumerate}[label=$(\roman*)$]
	\addtolength{\itemindent}{3mm}
		\item Remove element $\ei{i}$ from $\QS$; keep all remaining $\rmin$, $\rmax$ values unchanged. 
		
	\end{enumerate}
\end{tbox} 

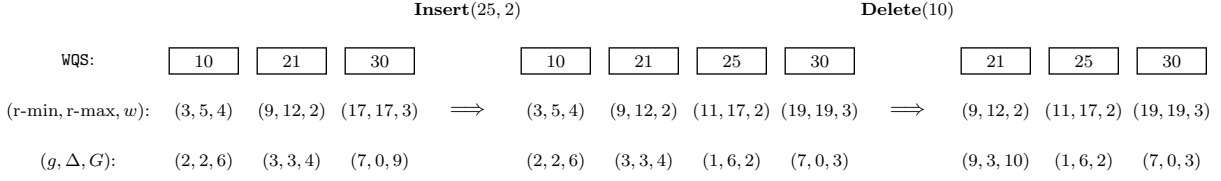
\begin{figure}[t]		
\centering	
 \hspace{1.5cm}\resizebox{0.9\textwidth}{!}{\hspace*{-5em}\begin {tikzpicture}[-latex, auto ,node distance =8mm and 14mm ,on grid ,
semithick ,
state/.style ={ rectangle ,top color =white , bottom color = white ,
draw,black , text=black , minimum width =11 mm},
labelstate/.style ={ rectangle ,top color =white , bottom color = white ,
draw,white , text=black , minimum width =11 mm}]

\node[state] (A){\footnotesize$10$}; 
\node[state] (B) [right =of A] {\footnotesize$21$}; 
\node[state] (C) [right =of B] {\footnotesize$30$}; 
\node[labelstate] (L) [left =20mm of A] {\footnotesize \WQS:};

\node[labelstate] (A1) [below =of A] {\footnotesize$(3,5,4)$}; 
\node[labelstate] (B1) [right =of A1] {\footnotesize$(9,12,2)$}; 
\node[labelstate] (C1) [right =of B1] {\footnotesize$(17,17,3)$}; 
\node[labelstate] (L1) [left =20mm of A1] {\footnotesize$( \rmin,\rmax,w)$:};

\node[labelstate] (A2) [below =of A1] {\footnotesize $(2,2,6)$}; 
\node[labelstate] (B2) [right =of A2] {\footnotesize $(3,3,4)$}; 
\node[labelstate] (C2) [right =of B2] {\footnotesize$(7,0,9)$}; 
\node[labelstate] (L2) [left =20mm of A2] {\footnotesize $( g,\Delta,G)$:};

\node[labelstate] (Ins) [above right =of C] {\footnotesize $\textbf{Insert}(25,2)$}; 

\node[labelstate] (D) [below right =of C] {$\Longrightarrow$}; 
\node[state] (E) [above right =of D] {\footnotesize\textcolor{black}{$10$}}; 
\node[state] (F) [right =of E] {\footnotesize$21$}; 
\node[state] (G) [right =of F] {\footnotesize\textcolor{black}{25}}; 
\node[state] (H) [right =of G] {\footnotesize$30$}; 

\node[labelstate] (E1) [below =of E] {\footnotesize\textcolor{black}{$(3,5,4)$}}; 
\node[labelstate] (F1) [right =of E1] {\footnotesize$(9,12,2)$}; 
\node[labelstate] (G1) [right =of F1] {\footnotesize\textcolor{black}{$(11,17,2)$}}; 
\node[labelstate] (H1) [right =of G1] {\footnotesize$(19,19,3)$}; 

\node[labelstate] (E2) [below =of E1] {\footnotesize\textcolor{black}{$(2,2,6)$}}; 
\node[labelstate] (F2) [right =of E2] {\footnotesize$(3,3,4)$}; 
\node[labelstate] (G2) [right =of F2] {\footnotesize\textcolor{black}{$(1,6,2)$}}; 
\node[labelstate] (H2) [right =of G2] {\footnotesize$(7,0,3)$}; 

\node[labelstate] (Del) [above right =of H] {\footnotesize\footnotesize $\textbf{Delete}(10)$}; 
\node[labelstate] (I) [below right =of H] {$\Longrightarrow$}; 
\node[state] (J) [above right =of I] {\footnotesize$21$}; 
\node[state] (K) [right =of J] {\footnotesize$25$}; 
\node[state] (L) [right =of K] {\footnotesize$30$}; 

\node[labelstate] (J1) [below =of J] {\footnotesize$(9,12,2)$}; 
\node[labelstate] (K1) [right =of J1] {\footnotesize$(11,17,2)$}; 
\node[labelstate] (L1) [right =of K1] {\footnotesize$(19,19,3)$}; 

\node[labelstate] (J2) [below =of J1] {\footnotesize$(9,3,10)$}; 
\node[labelstate] (K2) [right =of J2] {\footnotesize$(1,6,2)$}; 
\node[labelstate] (L2) [right =of K2] {\footnotesize$(7,0,3)$}; 

\end{tikzpicture}}	
\caption{An illustration of the update operations in the summary starting from some arbitrary state (the parameters $(g,\Delta, G)$ in this figure are defined in~\Cref{sec: weighted g delta}).}	
\label{fig:wt_insert-delete}	
\end{figure}
We now justify that after the above operations are performed, for each element $e$ in the summary, its $\rmin$ and $\rmax$ values are valid lower and upper bounds on the rank of $\cpy{e}{1}$. Suppose that a new element $x$ satisfying $e_{i-1} < x < e_{i}$ is inserted into $\WQS$. The rank of $x$ is at least one more than the rank of the last copy of $e_{i-1}$. Therefore, $\rmin(x)$, which is set to $(\rmin(e_{i-1}) + \wt(e_{i-1}) -1) + 1 = \rmin(e_{i-1}) + \wt(e_{i-1})$, is a valid lower bound on the rank of $\cpy{e}{1}$. The rank of $x$ is at most equal to the rank of the first copy of $e_i$. Therefore, setting $\rmax(x)$ equal to $\rmax(e_i)$ makes it a valid upper bound. After the insertion of $x$, the ranks of all elements in the summary larger than $x$ increase by $\wt(x)$ and hence their $\rmin$ and $\rmax$ values need to be updated. The ranks of elements smaller than $x$ do not change. Also, deleting an element from the summary does not change the bounds on the ranks of other elements in the summary.

 The following claim shows that if a certain condition on  $\rmin$ and $\rmax$ values of the elements in $\WQS$ is maintained, we can guarantee that $\WQS$ will be an $\eps$-approximate summary of $\suf$.  
\begin{claim}\label{lem:search-weighted}
 Suppose in $\WQS$ over a length $n$ stream, $\rmax(e_i) - (\rmin(e_{i-1}) + \wt(e_{i-1})-1) \leq \floor{\eps \wtsum{n}}$; then $\WQS$ is an $\eps$-approximate quantile summary of $\suf$.
\end{claim}
\begin{proof}
	First, note that for each $e$, $\WQS$ stores enough information to reconstruct $\rmin$ and 
	$\rmax$ 
	values of all copies of $e$ in $\suf$. Therefore, we assume in this proof that all copies of $e$ are 
	actually present in $\WQS$. 
	
	Due to \Cref{lem:search}, it is enough to show that for any two consecutive elements $e$ and 
	$e'$ of $\WQS$ satisfying $e \geq e'$, we have $\rmax(e) - \rmin(e') \leq \floor{\eps \wtsum{n}}$. 
	If $e$ and $e'$ are equal, i.e., $e = \cpy{e_i}{j}$ and $e' = \cpy{e_i}{j-1}$ for some $i \in [s]$ and 
	$2 \leq j \leq \wt(e_{i})$, then we have, 
	\begin{align*}
		\rmax(e) - \rmin(e') 
		& = (\rmax(e_i)+j-1)- (\rmin(e_i)+j-2)\\
		& = \rmax(e_i) - (\rmin(e_i) -1)\\
		& \leq \rmax(e_i) - (\rmin(e_{i-1} )+ \wt(e_{i-1}) - 1) \tag{From $\rmin(e_i) \geq \rmin(e_{i-1} )+ 
		\wt(e_{i-1})$}\\
		&\leq \floor{\eps \wtsum{n}}.\\
	\end{align*}
	The remaining case is when $e = \cpy{e_i}{1}$ and $e' = \cpy{e_{i-1}}{\wt(e_{i-1})}$ for some $i\in 
	[s]$.
	$$\rmax(e) - \rmin(e') = \rmax(e_i) -(\rmin(e_{i-1}) + \wt(e_{i-1}) -1) \leq \floor{\eps \wtsum{n}}. 
	\qedhere$$
\end{proof}

\subsection{Time Steps and Bands}

We define an equivalent notion of time steps for weighted streams. We say that $\tim{\emm} = 
\floor{\eps \wtsum{\emm}}$ time steps have elapsed after the arrival of $\emm$ elements in $\swt$. 
Intuitively, a chunk of total weight $\ell:= \frac{1}{\eps}$ arrives in the stream in a single time step. 
For each element $x_{\emm}$ in $\swt$, we define its insertion time step $t_0(x_{\emm}) = \floor{\eps (\wtsum{\emm-1}+1)}$.
 The band value of an element $x$ of $S_w$ is the band value assigned to the first copy of $x$ in $\suf$ by \Cref{def:band}.

We would also have a formal equivalent definition of bands for the weighted setting that will allow us to compute them in $O(1)$ time.  To see the equivalence, we refer the reader to \Cref{eq: closed-form} in \Cref{sec:basic}.
\begin{definition}[\textbf{Band-Values and Bands}]
	\label{eq:closed-form-weighted}
	When $\emm$ elements of $\swt$ have been inserted, for any element $x$ of the stream, $\bv(x)$ is $\alpha$ if and only if the following inequality is satisfied,
	\begin{align*}
	2^{\alpha-1}  + (\tim{\emm} \bmod 2^{\alpha-1}) \leq 
	\tim{\emm}-t_0(x)  < 2^{\alpha} + (\tim{\emm} \bmod 2^{\alpha}).
	\end{align*}

For any integer $\alpha \geq 0$, we refer to the set of all elements $x$ with $\bv(x) = \alpha$ as the \textbf{band} $\alpha$, denoted by $\Band_{\alpha}$; we also use $\Band_{\leq \alpha}$ to
denote the union of bands $0$ to $\alpha$. 		
\end{definition}

A corollary of \Cref{eq:closed-form-weighted} is that \emph{number of band-values} after seeing $\emm$ elements is $ \Bt{\emm} =O(\log \tim{\emm}) = O(\log (\eps \wtsum{\emm}))$. We also note that at any point, the sum of weights of all the elements belonging to bands $0$ to $\alpha$ is at most $O(\ell \cdot 2^{\alpha+1})$ because all the copies of all these elements belong to $\Band{\leq \alpha}$ for $\suf$. We note these facts below: 
\begin{align}
	\text{\# of b-values } \Bt{\emm}=O(\log \eps \wtsum{\emm} ) \quad \textnormal{and} \quad \text{$\sum_{x\in \Band{\leq \alpha}} w(x)\leq O(\ell \cdot 2^{\alpha+1})$ for all $\alpha \geq 0$}. \label{eq: weighted band-alpha}
\end{align}
We now make the following observation:
\begin{observation}\label{obs:weighted-stability}
    At any point in time, if $\bv(x) \leq \bv(y)$ for elements $x$ and $y$, then at any point after this, $\bv(x) \leq \bv(y)$.
\end{observation}	
This is simply because, in the unweighted setting, band-values of elements are updated \emph{simultaneously} based on the value of the current time step (\Cref{def:band}). Thus, the $\bv$ of the first copies of each element in a band is also updated simultaneously in $\suf$. Thus, \Cref{obs:weighted-stability} is true.

\subsection{Indirect handling of \texorpdfstring{$\rmin$}{rmin} and \texorpdfstring{$\rmax$}{rmax} values}\label{sec: weighted g delta}
To describe our algorithm, it is better to store the $\rmin$ and $\rmax$ values indirectly as $g$ and $\Delta$ values which we define for the weighted algorithm as follows. For any element $e_i$ in  $\WQS$,
\begin{align}\label{eq:g-delta-weighted}
g_i = \rmin(e_i) - (\rmin(e_{i-1}) + \wt(e_{i-1}) - 1), \qquad \Delta_i = \rmax(e_i) - \rmin(e_{i});
\end{align}

The $g$ value can be interpreted to be the difference between the minimum possible rank of $e_i$ 
and the minimum possible rank of the last copy of $e_{i-1}$. The $\Delta$ value is the difference 
between the $\rmax$ and $\rmin$ values of the first copy of $e_i$. The $\rmin$ and $\rmax$ values 
can be recovered given the $g$-values,  $\Delta$-values and the weights of all elements in $\WQS$ 
as follows:
\begin{align*}
    \rmin(e_i) = g_i + \sum\limits_{j = 1}^{i-1} (g_j+\wt(e_{j}) -1), \qquad \rmax(e_i) = \Delta_i + g_i +  \sum\limits_{j = 1}^{i-1} (g_j+\wt(e_{j}) -1).
\end{align*}

 This motivates the definition of the quantity $G_i$, for each element $\ei{i}$ in \WQS: 
\begin{align}\label{eq:G-weighted}
\G_i = \g_i+\wt(\ei{i})-1.
\end{align}
We will soon see that the $\G$-value has a nice property that will prove useful in the analysis of the algorithms that we propose. We now use the $g$ and $\Delta$ values defined to state the invariant that we maintain to ensure that $\WQS$ is an $\eps$-approximate quantile summary.
\begin{invariant}\label{inv: g-delta-weighted}
	After seeing $\emm$ elements of $\swt$, each element $e_i \in \WQS$ satisfies ${g_i + \Delta_i \leq \tim{\emm} }$.
\end{invariant}
From \Cref{eq:g-delta-weighted} we note that $g_i + \Delta_i = \rmax(e_i) - (\rmin(e_{i-1}) + \wt(e_{i-1})-1)$. Also, $\tim{\emm} = \floor{\eps \wtsum{\emm}}$ by definition. Therefore, if $\WQS$ maintains \Cref{inv: g-delta-weighted}, ~\Cref{lem:search-weighted} implies that it is an $\eps$-approximate quantile summary of $\swt$. 

The following observation now describes how $g$ and $\Delta$ values of elements are updated during $\textbf{Insert}$ and $\textbf{Delete}$ operations (see \Cref{sec: weighted-prelims}). 
\begin{observation}\label{obs: g-delta-update-weighted}
    In the summary $\WQS$:
    \begin{itemize}
        \item $\textbf{Insert}(x,w(x)):$ Sets $g(x) = 1$ and $\Delta(x) = g_i + \Delta_i -1$ and keeps the remaining $(g, \Delta)$ values unchanged.
        \item $\textbf{Delete}(e_i):$ Sets $g_{i+1}$ to equal $ g_{i+1}+  \G_i=g_{i+1} + (g_i + \wt(e_i) - 1) $, and keeps the remaining $(g,\Delta)$ values unchanged.
    \end{itemize}
\end{observation}
The correctness of \Cref{obs: g-delta-update-weighted} follows from \Cref{eq:g-delta-weighted} and the way in which $\rmin$ and $\rmax$ values change when these operations are performed. As promised, we present useful properties of $\G$ and $\Delta$ values.

\noindent
\textbf{$\bm{\G}$-value.} To understand this, we define the notion of \emph{coverage} of any element in \WQS. We say that $\ei{i}$ \emph{covers} $\ei{i-1}$ whenever $\ei{i-1}$ is deleted from
	the summary, in which case $\ei{i}$ also covers all elements that $\ei{i-1}$ was covering so far (every element only covers itself upon insertion). We define 
	\begin{itemize}
		\item $C(\ei{i})$: the set of elements covered by $\ei{i}$. By definition, at any point of time 
		\begin{align} \label{obs:weighted-disjoint coverage}
		\G_i = \sum_{x\in C(\ei{i})} w(x)
		\text{ and } C(\ei{i}) \cap C(\ei{j}) = \emptyset
		\end{align}
		for any $\ei{i},\ei{j}$ currently stored in $\QS$. 
	\end{itemize}

 We claim that $\G_i$ equals the sum of weights of the elements in the coverage of $e_i$. This is easy to verify by induction. When we insert an element, we set its $g$ value to be 1 and the element only covers itself, thus its $G$ value is equal to its weight by \Cref{eq:G-weighted}. In the way $G$-value is updated upon deletions, according to \Cref{obs: g-delta-update-weighted}, this continues to be the case throughout the algorithm.

\noindent
\textbf{$\Delta$}\textbf{-value.} The $\Delta$-value of an element is a measure of the error with which we know its rank. 
We can use \Cref{inv: g-delta-weighted} to deduce the following upper bound on the $\Delta$ value of an element $x$ in terms of its insertion time $t_0(x)$.  
\begin{align}\label{eq:del-t0-weighted}
 \Delta(x) \leq t_0(x). 
\end{align}

The intuition here is that after seeing $k$ elements of the stream, the maximum possible difference in possible ranks is bounded by $t_k$ if \Cref{inv: g-delta-weighted} is maintained. Hence, the error in the rank of a newly inserted element is also upper bounded by $t_k$. Formally, suppose that $x$ is the $j$-th element of the stream and satisfies ${e_{i-1} < x < e_i}$ at the time of insertion into $\WQS$. When $x$ is inserted into $\WQS$, we set $\Delta(x) = g_i + \Delta_i -1$. \Cref{inv: g-delta-weighted} implies that $\Delta(x) \leq \floor{\eps \wtsum{j-1}} - 1 \leq \floor{\eps (\wtsum{j-1}+1)} = t_0(x) $.

\section{A non-trivial extension of GK algorithm for weighted streams}\label{sec: weighted-stream-gk}

In this section, we present our extension of the GK algorithm for weighted streams. 
Recall the discussion in \Cref{sec:greedy-space} where we shed light on some counter-intuitive 
choices in GK summaries which turn out to be a basis for their 
tighter $O(\frac{1}{\eps} \log (\eps n))$ space. 
In particular, we motivate the following definition:
\begin{definition}[\textbf{Segment}]\label{def:segment}
 The \textbf{segment} of an element $\ei{i}$ in \WQS, denoted by $\Seg(\ei{i})$, is defined as the maximal set of consecutive elements ${\ei{j},\ei{j+1} ,\cdots ,\ei{i-1}}$ in \WQS with $\bv$ strictly less than $\bv(\ei{i})$. We let $\Gs_i$ be the sum
  of the $\G$-values of $\ei{i}$ and its segment, i.e., $\Gs_i=\G_i \hspace{2pt}+ \hspace{-7pt}\sum \limits_{\ei{k}\in \Seg(\ei{i})} \hspace{-2pt}G_k$.
\end{definition}

See \Cref{figs:segment} below for an illustration. 
\begin{figure}[!ht]
\centering
\hspace{-1cm}\resizebox{0.6\textwidth}{!}{		\tikzset{every picture/.style={line width=0.75pt}} 
		\resizebox{12cm}{!}{
		\begin{tikzpicture}[x=0.75pt,y=0.75pt,yscale=-1,xscale=1]

		\draw    (132,219) -- (503.11,219) (174,215) -- (174,223)(216,215) -- (216,223)(258,215) -- (258,223)(300,215) -- (300,223)(342,215) -- (342,223)(384,215) -- (384,223)(426,215) -- (426,223)(468,215) -- (468,223) ;
		\draw [shift={(506.11,219)}, rotate = 180] [fill={rgb, 255:red, 0; green, 0; blue, 0 }  ][line width=0.08]  [draw opacity=0] (8.93,-4.29) -- (0,0) -- (8.93,4.29) -- cycle    ;
		\draw    (132,12.42) -- (132,219) (136,47.42) -- (128,47.42)(136,82.42) -- (128,82.42)(136,117.42) -- (128,117.42)(136,152.42) -- (128,152.42)(136,187.42) -- (128,187.42) ;
		\draw [shift={(132,9.42)}, rotate = 90] [fill={rgb, 255:red, 0; green, 0; blue, 0 }  ][line width=0.08]  [draw opacity=0] (8.93,-4.29) -- (0,0) -- (8.93,4.29) -- cycle    ;
		\draw   (156,35.57) -- (186.51,35.57) -- (186.51,53) -- (156,53) -- cycle ;
		\draw  [fill={rgb, 255:red, 225; green, 225; blue, 225 }  ,fill opacity=1 ] (198,108.57) -- (228.51,108.57) -- (228.51,126) -- (198,126) -- cycle ;
		\draw  [fill={rgb, 255:red, 160; green, 160; blue, 160 }  ,fill opacity=1 ] (243,140.57) -- (273.51,140.57) -- (273.51,158) -- (243,158) -- cycle ;
		\draw  [fill={rgb, 255:red, 160; green, 160; blue, 160 }  ,fill opacity=1 ] (285,174.57) -- (315.51,174.57) -- (315.51,192) -- (285,192) -- cycle ;
		\draw  [fill={rgb, 255:red, 225; green, 225; blue, 225 }  ,fill opacity=1 ] (326,107.57) -- (356.51,107.57) -- (356.51,125) -- (326,125) -- cycle ;
		\draw   (364,70.57) -- (394.51,70.57) -- (394.51,88) -- (364,88) -- cycle ;
		\draw   (411,138.57) -- (441.51,138.57) -- (441.51,156) -- (411,156) -- cycle ;
		\draw   (451,101.57) -- (481.51,101.57) -- (481.51,119) -- (451,119) -- cycle ;
		\draw   (342.51,87.72) .. controls (342.51,83.05) and (340.18,80.72) .. (335.51,80.72) -- (288.01,80.72) .. controls (281.34,80.72) and (278.01,78.39) .. (278.01,73.72) .. controls (278.01,78.39) and (274.68,80.72) .. (268.01,80.72)(271.01,80.72) -- (220.51,80.72) .. controls (215.84,80.72) and (213.51,83.05) .. (213.51,87.72) ;
		\draw   (302.51,136.36) .. controls (302.51,131.69) and (300.18,129.36) .. (295.51,129.36) -- (291.82,129.36) .. controls (285.15,129.36) and (281.82,127.03) .. (281.82,122.36) .. controls (281.82,127.03) and (278.49,129.36) .. (271.82,129.36)(274.82,129.36) -- (268.51,129.36) .. controls (263.84,129.36) and (261.51,131.69) .. (261.51,136.36) ;
		
		\draw (80,100) node [anchor=north west, rotate=90][inner sep=0.75pt]  [font=\footnotesize,xscale=1.3,yscale=1.3] [align=left] {Band-value};
		\draw (500,230) node [anchor=north west][inner sep=0.75pt]  [font=\footnotesize,xscale=1.3,yscale=1.3] [align=left] {Position };
		\draw (251,56.08) node [anchor=north west][inner sep=0.75pt]  [font=\scriptsize,xscale=1.3,yscale=1.3] [align=left] {seg($\displaystyle e_{6}$)};
		\draw (107,185) node [anchor=north west][inner sep=0.75pt]  [font=\footnotesize,xscale=1.3,yscale=1.3] [align=left] {$\displaystyle 1$};
		\draw (107,150) node [anchor=north west][inner sep=0.75pt]  [font=\footnotesize,xscale=1.3,yscale=1.3] [align=left] {$\displaystyle 2$};
		\draw (108,115) node [anchor=north west][inner sep=0.75pt]  [font=\footnotesize,xscale=1.3,yscale=1.3] [align=left] {$\displaystyle 3$};
		\draw (109,80) node [anchor=north west][inner sep=0.75pt]  [font=\footnotesize,xscale=1.3,yscale=1.3] [align=left] {$\displaystyle 4$};
		\draw (110,45) node [anchor=north west][inner sep=0.75pt]  [font=\footnotesize,xscale=1.3,yscale=1.3] [align=left] {$\displaystyle 5$};
		\draw (169,231.36) node [anchor=north west][inner sep=0.75pt]  [font=\footnotesize,xscale=1.3,yscale=1.3] [align=left] {1};
		\draw (212,232.36) node [anchor=north west][inner sep=0.75pt]  [font=\footnotesize,xscale=1.3,yscale=1.3] [align=left] {2};
		\draw (253,232.36) node [anchor=north west][inner sep=0.75pt]  [font=\footnotesize,xscale=1.3,yscale=1.3] [align=left] {3};
		\draw (295,232.36) node [anchor=north west][inner sep=0.75pt]  [font=\footnotesize,xscale=1.3,yscale=1.3] [align=left] {4};
		\draw (338,232.36) node [anchor=north west][inner sep=0.75pt]  [font=\footnotesize,xscale=1.3,yscale=1.3] [align=left] {5};
		\draw (380,232.36) node [anchor=north west][inner sep=0.75pt]  [font=\footnotesize,xscale=1.3,yscale=1.3] [align=left] {6};
		\draw (422,231.36) node [anchor=north west][inner sep=0.75pt]  [font=\footnotesize,xscale=1.3,yscale=1.3] [align=left] {7};
		\draw (463,231.36) node [anchor=north west][inner sep=0.75pt]  [font=\footnotesize,xscale=1.3,yscale=1.3] [align=left] {8};
		\draw (250,100) node [anchor=north west][inner sep=0.75pt]  [font=\scriptsize,xscale=1.3,yscale=1.3] [align=left] {seg($\displaystyle e_{5}$)};
		\draw (369,73) node [anchor=north west][inner sep=0.75pt]  [xscale=1.3,yscale=1.3] [align=left] {$\displaystyle e_{6}$};
		\draw (329.51,110) node [anchor=north west][inner sep=0.75pt]  [xscale=1.3,yscale=1.3] [align=left] {$\displaystyle e_{5}$};
		\draw (288.51,176) node [anchor=north west][inner sep=0.75pt]  [xscale=1.3,yscale=1.3] [align=left] {$\displaystyle e_{4}$};
		\draw (247.51,142) node [anchor=north west][inner sep=0.75pt]  [xscale=1.3,yscale=1.3] [align=left] {$\displaystyle e_{3}$};
		\draw (202.51,109) node [anchor=north west][inner sep=0.75pt]  [xscale=1.3,yscale=1.3] [align=left] {$\displaystyle e_{2}$};
		\draw (160.51,37) node [anchor=north west][inner sep=0.75pt]  [xscale=1.3,yscale=1.3] [align=left] {$\displaystyle e_{1}$};
		\draw (415.51,139) node [anchor=north west][inner sep=0.75pt]  [xscale=1.3,yscale=1.3] [align=left] {$\displaystyle e_{7}$};
		\draw (456.51,103) node [anchor=north west][inner sep=0.75pt]  [xscale=1.3,yscale=1.3] [align=left] {$\displaystyle e_{8}$};

		\end{tikzpicture}
		}}
\caption{An illustration of \Cref{def:segment}. The ranks of elements increase along the horizontal axis. The segment of the element $\ei{5}$ contains $\ei{3}$ and $\ei{4}$. The segment of $\ei{6}$ contains  $\ei{2},\ei{3},\ei{4}$ and $\ei{5}$. }
\label{figs:segment}
\end{figure}
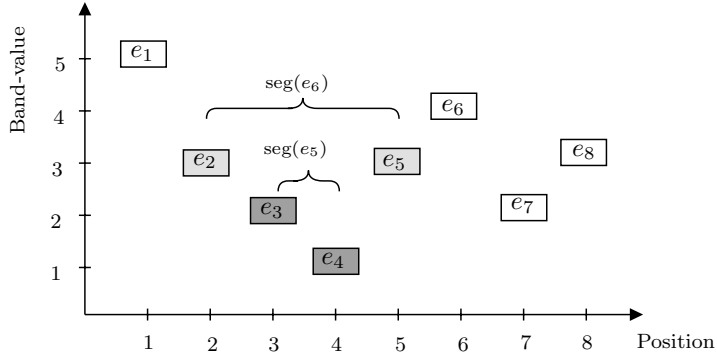

 At any step, the algorithm first inserts the arriving element into $\WQS$; we call this the \emph{insertion step}. It then only deletes an element from \WQS if it can be deleted \emph{together} with its entire segment without violating \Cref{inv: g-delta-weighted}.  While there is any such element whose deletion (along with its segment) does not violate~\Cref{inv: g-delta-weighted} (and another simple but important condition on \bv), the algorithm deletes it from \WQS; we call this the \emph{deletion step}. We now give a formal description of the algorithm.

\begin{Algorithm}\label{alg:weighted-GK}
	A generalization of the GK algorithm for weighted streams:
	\medskip
	
	For each arriving item $(x_j,\wt(x_j))$:
	\begin{enumerate}[label=$(\roman*)$]
	\addtolength{\itemindent}{3mm}
		\item Run $\textbf{Insert}(x_j, \wt(x_j))$:
		\item While there exists an element $\ei{i}$ in \QS satisfying:  	
		\[
		(1)~\bv(\ei{i}) \leq \bv(\ei{i+1})\qquad \text{\underline{and}} \qquad (2)~\Gs_i+\g_{i+1}+\Delta_{i+1} \leq \tim{j} 
		\]
		 run $\textbf{Delete}(\ei{k})$ for $\ei{k}$ in $\set{\ei{i}} \cup \Seg(\ei{i})$.
	\end{enumerate}
\end{Algorithm}
\begin{theorem}\label{thm:weighted-GK}
	For any $\eps > 0$ and a weighted stream of length $n$ with total weight $\wtsum{n}$,~\Cref{alg:weighted-GK} maintains an $\eps$-approximate quantile summary in $O(\frac{1}{\eps} \cdot \log{\!(\eps \wtsum{n})})$ space. Also, there is an implementation of ~\Cref{alg:weighted-GK} that takes $O\big(\log(1/\eps)+\log \log (\eps \wtsum{n}) + \frac{\log^2(\eps \wtsum{n})}{\eps n}\big)$ worst-case update time per element.
\end{theorem}

We remark here that as long as the weights are $\poly(n)$ bounded and $\eps \geq {1}/{n^{1-\delta}}$ for any fixed $\delta \in (0,1)$, the space used by the algorithm will be $O((1/\eps)\log(\eps n))$ and its update time will be $O(\log (1/\eps) + \log \log (\eps n))$. This matches the space and time complexities of the implementation of GK summary described in \cite{LuoWYC16}. Note that the interesting regime for $\eps$ is at least a small constant, because when $\eps < 1/n^{1-\delta}$, the information-theoretic lower bound of $(1/2\eps)$ on the summary size already implies that we need to store $\Omega(n ^{1-\delta})$ elements even for original GK summaries on unweighted inputs, which is prohibitive for most applications. 

 \Cref{alg:weighted-GK} maintains a valid $\eps$-approximate summary  since it may only delete an element $\ei{i}$ along with its segment if the condition $(ii)$: $\Gs_i+g_{i+1}+\Delta_{i+1} \leq \tim{k}$ is satisfied. Thus, \Cref{inv: g-delta-weighted} is satisfied for the element $\ei{i+1}$ after the deletion of $\ei{i}$ (other $g$ and $\Delta$ values are unaffected by this). We now focus on the bounding the space used by the algorithm in the following. Then, in \Cref{subsec: fast-imp-weighted-GK}, we give an efficient implementation to finalize the proof of \Cref{thm:weighted-GK}.
\subsection{Space Analysis}\label{subsec: space analysis wt gk}

In this subsection, we prove a bound on the space used by \Cref{alg:weighted-GK}. Formally, we have the following:
\begin{lemma}\label{lem:GK-space-weighted}
For any $\eps > 0$ and a stream of length $n$ with the total weight $\wtsum{n}$,~\Cref{alg:weighted-GK} maintains an $\eps$-approximate quantile summary in $O(\frac{1}{\eps} \cdot \log{\!(\eps \wtsum{n})})$ space. 
\end{lemma}

We first make a critical observation.
\begin{observation}\label{obs: deleted-band-inv-weighted}
Elements from $\Band_{\leq \alpha}$ in \WQS only cover elements of $\Band_{\leq \alpha}$ at any time.
\end{observation}
This is because when $\ei{i}$ and $\Seg(\ei{i})$ get covered by $\ei{i+1}$, \Cref{alg:weighted-GK} ensures that $\bv(\ei{i})\leq \bv(\ei{i+1})$. From \Cref{def:segment}, $\Seg(\ei{i})$ contains elements with $\bv$ less than $\bv(\ei{i})$. Thus, $C(\ei{i+1})$ contains elements with $\bv$ at most $\bv(\ei{i+1})$ and this continues to be the case at a later time by \Cref{obs:weighted-stability}. 

Another important observation is that, after executing a deletion step after $k$ insertions, an element $\ei{i}$ present in $\WQS$  either satisfies $\bv(\ei{i}) > \bv(\ei{i+1})$ or  $\Gs_i+\g_{i+1}+\Delta_{i+1} > \tim{k}$; otherwise~\Cref{alg:weighted-GK} would have deleted this element. We refer to the elements in $\WQS$ satisfying the former condition as \emph{\To} elements and the ones satisfying {only} the latter condition as \emph{\Tt} elements. Thus, each element is exactly one of the two types (except only $\ei{s}=+\infty$ which we can ignore). It will therefore suffice to obtain a bound on the number of \To and \Tt elements to bound the space complexity of $\WQS$. Let us first bound the number of \To elements in the following lemma.

\begin{lemma}\label{lem: type-1 weighted-GK}
After the deletion step when $\emm$ elements have been seen, the number of \To elements stored in \WQS is $O\big(\ell \cdot \log \tim{\emm} \big)$.
\end{lemma}
\begin{proof}
    We first partition the \To elements into $\Bt{\emm}$ sets $Y_0,\ldots,Y_{\Bt{\emm}}$ where for any band-value $\alpha$:
	\[
		Y_{\alpha} := \set{\ei{i} \in \WQS \mid \text{$\ei{i}$ is \To and $\bv(\ei{i+1}) = \alpha$}};
	\] 
	(notice that elements in $Y_{\alpha}$ are such that band-value of their next element is $\alpha$, not themselves\footnote{\label{foot1}While this may sound counter-intuitive at first glance, recall that the criteria for defining the type of an element is a function of both this element and the next one; this definition allows us to take this into account.})
	We will show that the size of any set $Y_\alpha$ is at most $O(\ell)$. We map each element of $\ei{i}$ to the smallest element $\ei{j}$ with $\bv$ greater than $\alpha$; see \Cref{fig:type-1-notation} for an illustration. Let $T_{\alpha}$ be the set of all such elements $\ei{j}$. Also, it is easy to see that the mapping from $Y_\alpha$ to $T_\alpha$ is one to one; giving us $\card{Y_{\alpha}}=\card{T_{\alpha}}$. Note that $\ei{j-1}$ must be a \Tt element. Hence,
\begin{equation} \label{eq: weighted-lem To terminal}
\Gs_{j-1} + \g_{j} + \Delta_j > \tim{k}.
\end{equation}

\begin{figure}[t]		
\centering	

\subcaptionbox{ The shaded blocks are elements of $Y_4$. The arrows indicate the mapping from elements in $Y_4$ to elements in $T_4$. Each element $\ei{i}$ in $Y_4$ is mapped to the first larger element $\ei{j}$ with a band-value higher than $4$. \label{fig:type-1-notation}}%
  [1\linewidth]{ \resizebox{0.6\textwidth}{!}{\tikzset{every picture/.style={line width=0.75pt}} 

\begin{tikzpicture}[x=0.75pt,y=0.75pt,yscale=-1,xscale=1]

\draw  [fill={rgb, 255:red, 222; green, 222; blue, 222 }  ,fill opacity=1 ] (164,60) -- (199.5,60) -- (199.5,80.29) -- (164,80.29) -- cycle ;
\draw  [fill={rgb, 255:red, 222; green, 222; blue, 222 }  ,fill opacity=1 ] (253,60) -- (288.5,60) -- (288.5,80.29) -- (253,80.29) -- cycle ;
\draw   (297,60) -- (332.5,60) -- (332.5,80.29) -- (297,80.29) -- cycle ;
\draw   (209,60) -- (244.5,60) -- (244.5,80.29) -- (209,80.29) -- cycle ;
\draw   (340,60) -- (375.5,60) -- (375.5,80.29) -- (340,80.29) -- cycle ;
\draw   (429,60) -- (464.5,60) -- (464.5,80.29) -- (429,80.29) -- cycle ;
\draw   (473,60) -- (508.5,60) -- (508.5,80.29) -- (473,80.29) -- cycle ;
\draw   (385,60) -- (420.5,60) -- (420.5,80.29) -- (385,80.29) -- cycle ;
\draw    (181.5,81.44) .. controls (211.88,96.14) and (243.22,97.4) .. (271.76,83.33) ;
\draw [shift={(273.5,82.44)}, rotate = 512.65] [color={rgb, 255:red, 0; green, 0; blue, 0 }  ][line width=0.75]    (10.93,-3.29) .. controls (6.95,-1.4) and (3.31,-0.3) .. (0,0) .. controls (3.31,0.3) and (6.95,1.4) .. (10.93,3.29)   ;
\draw    (269,59.44) .. controls (313.82,27.92) and (460.99,39.09) .. (492.16,57.59) ;
\draw [shift={(493.5,58.44)}, rotate = 214.16] [color={rgb, 255:red, 0; green, 0; blue, 0 }  ][line width=0.75]    (10.93,-3.29) .. controls (6.95,-1.4) and (3.31,-0.3) .. (0,0) .. controls (3.31,0.3) and (6.95,1.4) .. (10.93,3.29)   ;
\draw  [dash pattern={on 0.84pt off 2.51pt}]  (130,70.89) -- (153.5,70.89) ;
\draw  [dash pattern={on 0.84pt off 2.51pt}]  (519,68.89) -- (542.5,68.89) ;

\draw (174,62.44) node [anchor=north west][inner sep=0.75pt]  [font=\footnotesize,xscale=1.3,yscale=1.3] [align=left] {$\displaystyle 6$};
\draw (220,62.44) node [anchor=north west][inner sep=0.75pt]  [font=\footnotesize,xscale=1.3,yscale=1.3] [align=left] {$\displaystyle 4$};
\draw (264,63.44) node [anchor=north west][inner sep=0.75pt]  [font=\footnotesize,xscale=1.3,yscale=1.3] [align=left] {$\displaystyle 5$};
\draw (351,63.44) node [anchor=north west][inner sep=0.75pt]  [font=\footnotesize,xscale=1.3,yscale=1.3] [align=left] {$\displaystyle 2$};
\draw (309,62.44) node [anchor=north west][inner sep=0.75pt]  [font=\footnotesize,xscale=1.3,yscale=1.3] [align=left] {$\displaystyle 4$};
\draw (396,63.44) node [anchor=north west][inner sep=0.75pt]  [font=\footnotesize,xscale=1.3,yscale=1.3] [align=left] {$\displaystyle 2$};
\draw (483.5,63.29) node [anchor=north west][inner sep=0.75pt]  [font=\footnotesize,xscale=1.3,yscale=1.3] [align=left] {$\displaystyle 7$};
\draw (440.5,63.29) node [anchor=north west][inner sep=0.75pt]  [font=\footnotesize,xscale=1.3,yscale=1.3] [align=left] {$\displaystyle 3$};
\draw (258,89.98) node [anchor=north west][inner sep=0.75pt]  [font=\footnotesize,xscale=1.3,yscale=1.3] [align=left] {$ $$\displaystyle e_{i}$};
\draw (297,89.99) node [anchor=north west][inner sep=0.75pt]  [font=\footnotesize,xscale=1.3,yscale=1.3] [align=left] {$ $$\displaystyle e_{i+1}$};
\draw (481,85.99) node [anchor=north west][inner sep=0.75pt]  [xscale=1.3,yscale=1.3] [align=left] {$ $$\displaystyle e_{j}$};
\draw (430,89.99) node [anchor=north west][inner sep=0.75pt]  [font=\footnotesize,xscale=1.3,yscale=1.3] [align=left] {$ $$\displaystyle e_{j-1}$};

\end{tikzpicture}}}	
\bigskip
\vspace{0.25cm}	
\subcaptionbox{ Each dark gray block represents an element $\ei{j}$ in $T_4$. All elements which are either $\ei{j-1}$ or are in the $\Seg(\ei{j-1})$ are shaded light gray. \label{fig:all-bands} }%
 [1\linewidth]{ \resizebox{0.6\textwidth}{!}{\tikzset{every picture/.style={line width=0.75pt}} 

\begin{tikzpicture}[x=0.75pt,y=0.75pt,yscale=-1,xscale=1]

\draw   (161,31) -- (196.5,31) -- (196.5,51.29) -- (161,51.29) -- cycle ;
\draw  [fill={rgb, 255:red, 177; green, 177; blue, 177 }  ,fill opacity=1 ] (250,31) -- (285.5,31) -- (285.5,51.29) -- (250,51.29) -- cycle ;
\draw   (294,31) -- (329.5,31) -- (329.5,51.29) -- (294,51.29) -- cycle ;
\draw  [fill={rgb, 255:red, 224; green, 224; blue, 224 }  ,fill opacity=1 ] (206,31) -- (241.5,31) -- (241.5,51.29) -- (206,51.29) -- cycle ;
\draw  [fill={rgb, 255:red, 222; green, 222; blue, 222 }  ,fill opacity=1 ] (337,31) -- (372.5,31) -- (372.5,51.29) -- (337,51.29) -- cycle ;
\draw  [fill={rgb, 255:red, 224; green, 224; blue, 224 }  ,fill opacity=1 ] (426,31) -- (461.5,31) -- (461.5,51.29) -- (426,51.29) -- cycle ;
\draw  [fill={rgb, 255:red, 177; green, 177; blue, 177 }  ,fill opacity=1 ] (470,31) -- (505.5,31) -- (505.5,51.29) -- (470,51.29) -- cycle ;
\draw  [fill={rgb, 255:red, 222; green, 222; blue, 222 }  ,fill opacity=1 ] (382,31) -- (417.5,31) -- (417.5,51.29) -- (382,51.29) -- cycle ;
\draw   (352,59.89) .. controls (352,64.56) and (354.33,66.89) .. (359,66.89) -- (392,66.89) .. controls (398.67,66.89) and (402,69.22) .. (402,73.89) .. controls (402,69.22) and (405.33,66.89) .. (412,66.89)(409,66.89) -- (445,66.89) .. controls (449.67,66.89) and (452,64.56) .. (452,59.89) ;
\draw  [dash pattern={on 0.84pt off 2.51pt}]  (123,40.89) -- (146.5,40.89) ;
\draw  [dash pattern={on 0.84pt off 2.51pt}]  (516,40.89) -- (539.5,40.89) ;

\draw (171,33.44) node [anchor=north west][inner sep=0.75pt]  [font=\footnotesize,xscale=1.3,yscale=1.3] [align=left] {$\displaystyle 6$};
\draw (217,33.44) node [anchor=north west][inner sep=0.75pt]  [font=\footnotesize,xscale=1.3,yscale=1.3] [align=left] {$\displaystyle 4$};
\draw (261,34.44) node [anchor=north west][inner sep=0.75pt]  [font=\footnotesize,xscale=1.3,yscale=1.3] [align=left] {$\displaystyle 5$};
\draw (348,34.44) node [anchor=north west][inner sep=0.75pt]  [font=\footnotesize,xscale=1.3,yscale=1.3] [align=left] {$\displaystyle 2$};
\draw (306,33.44) node [anchor=north west][inner sep=0.75pt]  [font=\footnotesize,xscale=1.3,yscale=1.3] [align=left] {$\displaystyle 4$};
\draw (393,34.44) node [anchor=north west][inner sep=0.75pt]  [font=\footnotesize,xscale=1.3,yscale=1.3] [align=left] {$\displaystyle 2$};
\draw (480.5,34.29) node [anchor=north west][inner sep=0.75pt]  [font=\footnotesize,xscale=1.3,yscale=1.3] [align=left] {$\displaystyle 7$};
\draw (437.5,34.29) node [anchor=north west][inner sep=0.75pt]  [font=\footnotesize,xscale=1.3,yscale=1.3] [align=left] {$\displaystyle 3$};
\draw (363,77.97) node [anchor=north west][inner sep=0.75pt]  [font=\scriptsize,xscale=1.3,yscale=1.3] [align=left] {{\tiny $\displaystyle e_{j-1}$ and its segment}};

\end{tikzpicture}}}	
\caption{The two figures represent a section of the summary with each block representing an element. The number inside the block is the element's band-value. }	
\label{fig:type-1-segment}	

\end{figure}

Since $\bv(\ei{j})$ is greater than $\bv(\ei{i+1})=\alpha$, by~\Cref{obs:weighted-stability}, one can argue that $\ei{j}$ is inserted in \WQS before $\ei{i+1}$. Let ${\g}^{\prime}_j$ be the $\g$-value of $\ei{j}$ when $\ei{i+1}$ got inserted. By \Cref{inv: g-delta-weighted},
\begin{equation} \label{eq: weighted-lem To intertion time}
\g^\prime_j+\Delta_j \leq \tim{0}(\ei{i+1}) \quad (\textnormal{$\Delta$ value does not change over time)}. 
\end{equation}
Subtracting \Cref{eq: weighted-lem To intertion time} from \Cref{eq: weighted-lem To terminal} and using the bounds from~\Cref{eq:closed-form-weighted} we conclude that
\begin{equation}\label{eq: weighted-gk-To-main}
\Gs_{j-1}+ (\g_j-\g^\prime_j)>\tim{k}-\tim{0}(\ei{i+1}) \geq2^{\alpha-1}-2.
\end{equation}

 In the above equation: 
 \begin{enumerate}[label=$(\roman*)$]
 \item The term $(\g_j-\g^\prime_j)$ counts the sum of weights of the elements covered by $\ei{j}$ after $\ei{i+1}$ is inserted. \Cref{claim:weighted-coverage-time} will show that these elements are in $\Band_{\leq \alpha}$. 
 
 \item The term $\Gs_{j-1}$ counts the sum of the weights of the elements covered by $\ei{j-1}$ and $\Seg(\ei{j-1})$. By \Cref{def:segment}, $\ei{j-1}$ and all elements in $\Seg(\ei{j-1})$ have $\bv\leq \alpha$. \Cref{obs: deleted-band-inv-weighted} allows us to conclude that the sum of weights of elements counted by $\Gs_{j-1}$ are in $\Band_{\leq \alpha}$ as well. 
 
 \item Additionally, it is easy to see that for distinct $\ei{j_1}$ and $\ei{j_2}$ in $T_{\alpha}$, the segments of $\ei{j_1-1}$ and $\ei{j_2-1}$ do not overlap (as can be observed in \Cref{fig:type-1-segment}). Thus, by \Cref{obs:weighted-disjoint coverage}, the elements covered by $\ei{j_1}$ and its segment are distinct from the elements covered by and $\ei{j_2}$ and its segment.
 \end{enumerate}
Finally, from the above discussion, we conclude that the LHS of~\Cref{eq: weighted-gk-To-main}, summed over all $T_{\alpha}$, is proportional to the total weight of all the elements in $\Band_{\leq \alpha}$. Formally, 
   \begin{align*}
		\hspace{-5mm} \card{T_{\alpha}} \cdot (2^{\alpha-1}-2) &\leq \sum_{\ei{j} \in T_{\alpha}} \Gs_{j-1} + \sum_{\ei{j} \in T_{\alpha}} (\g_j-\g^\prime_j) \leq \hspace{-3mm} \sum_{x_k\in \Band_{\leq \alpha}} \hspace{-2mm} w(x_k)+ \hspace{-3mm}\sum_{x_k \in \Band_{\leq \alpha}}\hspace{-2mm} w(x_k) 
		\leq O(\ell \cdot 2^{\alpha+2}),
	\end{align*}
where the last inequality follows from \Cref{eq: weighted band-alpha}. Hence, $\card{T_{\alpha}}=\card{Y_{\alpha}}=O(\ell)$ for $\alpha\geq 3$. We have $O(\ell)$ elements for $\alpha=0,1,2$ anyway. Since there are $\Bt{k}=\log \tim{k}$ possible values of $\alpha$, the number of \To elements is $O(\ell \log \tim{k})$. 
\begin{claim}\label{claim:weighted-coverage-time}
 All the elements covered by $\ei{j}$ after $\ei{i+1}$ was inserted have $\bv$ at most $\alpha$ currently.
\end{claim}
\begin{proof}
Let us assume that there exists an element which currently has $\bv>\alpha$ but gets covered by $\ei{j}$ after $\ei{i+1}$ was inserted. All such elements are less than $\ei{j}$ and greater than $\ei{i+1}$ since they get covered by $\ei{j}$.
Now consider the smallest such element $e$. Clearly, $\ei{i+1}$ belongs to the segment of $e$ just after the insertion of $\ei{i+1}$. Since $e$ does not belong to the summary right now, it must have been deleted. This implies that its segment, which contained $\ei{i+1}$,  got deleted. This means that $\ei{i+1}$ is also deleted, which is a contradiction. \Qed{\Cref{claim:weighted-coverage-time}}

 \end{proof}

This finalizes the proof of \Cref{lem: type-1 weighted-GK}.
\end{proof}
It now remains to bound the number of \Tt elements in $\WQS$ which we do in the following lemma.
\begin{lemma}\label{lem:weighted-type-2-GK}
	After the deletion step when $k$ elements of the stream have been seen, the number of \Tt elements is $O(\ell \cdot \log \tim{k})$.
\end{lemma}
\begin{proof}
	Any \Tt element $\ei{i}$ in \WQS, has the property $\Gs_i+\g_{i+1}+\Delta_i>\tim{k}$. This will give a lower bound on $\Gs_i + g_{i+1}$ in terms of the $\bv(e_{i+1})$.

	\begin{claim}\label{clm:weighted-lower-g-GK}
		After seeing $k$ elements, for any \Tt element $\ei{i}$, 
		$ \Gs_i + \g_{i+1} \geq 2^{\bv(\ei{i+1})-1}-2$.
	\end{claim}
	\begin{proof}
		As $\ei{i}$ is a \Tt element, $\Gs_i + \g_{i+1} + \Delta_{i+1} > \tim{k}$. By~\Cref{eq:del-t0-weighted}, $\Delta_{i+1} \leq t_0(\ei{i+1})$ and therefore, 
	\begin{align*}
	 \Gs_i+\g_{i+1}> \tim{k} - t_0(\ei{i+1}) \geq 2^{\bv(\ei{i+1})-1}-2,
	\end{align*}
	where the second inequality is by~\Cref{eq:closed-form-weighted}. \Qed{\Cref{clm:weighted-lower-g-GK}}
	
\end{proof}

\Cref{clm:weighted-lower-g-GK} gives us a lower bound on the $\Gs$-value of each \Tt element $\ei{i}$ as a function of the $g$-value and band-value of the \emph{next} element $\ei{i+1}$. Therefore, we partition the \Tt elements into sets $X_0,\ldots,X_{\Bt{k}}$ such that, for any band-value $\alpha$,
\[
X_{\alpha} := \set{\ei{i} \in \WQS \mid \text{$\ei{i}$ is \Tt and $\bv(\ei{i+1}) = \alpha$}}.
\]

Moreover, for any $\ei{i} \in X_{\alpha}$, since $\ei{i}$ is a \Tt element, $\bv(\ei{i}) \leq \bv(\ei{i+1}) = \alpha$. Summing over the inequality of~\Cref{clm:weighted-lower-g-GK} for each element in $X_\alpha$, we obtain:
	\begin{align} \label{eq:weighted-lem-type-2-GK}
		\card{X_{\alpha}} \cdot (2^{\alpha-1}-2) &\leq \sum_{\ei{i} \in X_{\alpha}} \Gs_i + g_{i+1}.
	\end{align}
We next show an upper bound on the right hand side of \Cref{eq:weighted-lem-type-2-GK} which will imply the necessary bound on $|X_\alpha|$.

\begin{claim} \label{clm: weighted-gs-bound-GK}
After seeing $k$ elements, for any $\alpha \geq 0$, 
	$\sum \limits_{\ei{i} \in X_{\alpha}} \Gs_i  \hspace{3mm} \leq \hspace{3mm} 2 \hspace{-15pt}\sum \limits_{\ei{j} \in \WQS \cap \Band_{\leq \alpha}} \hspace{-15pt} \G_j.$
\end{claim}
\begin{proof}


 We partition $X_\alpha$ into two disjoint sets $X_\alpha \cap \Band_{\alpha}$ and $X_\alpha \cap \Band_{\leq \alpha-1}$ and observe that two elements from one of these two sets must have disjoint segments. Also, the elements in their segments must all be in $\Band_{\leq \alpha}$. Therefore,
\begin{align*}
\sum \limits_{\ei{i} \in X_{\alpha}} \Gs_i \;\;
& = \sum \limits_{\ei{i} \in X_{\alpha} \cap \Band_{\alpha}} \hspace{-5pt}\Gs_i  
+ \hspace{-3mm}\sum \limits_{\ei{i} \in X_{\alpha} \cap \Band_{\leq \alpha -1}} \hspace{-10pt} \Gs_i \hspace{5pt} \leq \hspace{-3mm}\sum \limits_{\ei{j} \in \text{\WQS} \cap \Band_{\leq \alpha}} \hspace{-15pt}\G_j 
	+  \sum \limits_{\ei{j} \in \text{\WQS} \cap \Band_{\leq \alpha}} \hspace{-15pt}\G_j.\\
	&= 2 \hspace{-15pt}\sum \limits_{\ei{j} \in \text{\WQS} \cap \Band_{\leq \alpha}} \hspace{-15pt}\G_j .\Qed{\Cref{clm: weighted-gs-bound-GK}}
	\end{align*}

\end{proof}
The next claim bounds the sum of $\G$-values of the elements in \WQS from $\Band_{\leq \alpha}$.
\begin{claim}\label{clm:weighted-cover-gk}
	After seeing $k$ elements, for any $\alpha\geq 0$, 
	$
	\sum_{\ei{i} \in \WQS \cap \Band_{\leq \alpha}} G_i \leq O(\ell \cdot 2^{\alpha+1}).
	$ 
\end{claim}
\begin{proof}
An element is only deleted by \Cref{alg:weighted-GK} if the condition (1) is satisfied. By \Cref{obs:weighted-stability}, this continues to be the case at any later point in the algorithm. Therefore, $C(\ei{i})$ only contains elements whose $\bv$ is at most $\bv(\ei{i})$. Therefore, 
	\begin{align*}
		\sum_{\ei{i} \in \WQS \cap \Band_{\leq \alpha}} \G_i &= \sum_{\ei{i} \in \WQS \cap \Band_{\leq \alpha}} \sum_{x_j\in C(\ei{i})} \wt(x_j)  \tag{as $\G_i = \sum_{x_j\in C(e_i)} w(x_j)$ by~\Cref{obs:weighted-disjoint coverage}} \\
		&\leq \sum_{x_j \in \Band{\leq \alpha}} w(x_j)  \tag{as $C(\ei{i})$'s are disjoint and their elements belong to $\Band{\leq\alpha}$}\\
		&=O (\ell \cdot 2^{\alpha+1}), \tag{by the bound in~\Cref{eq: weighted band-alpha}} 
	\end{align*}
	completing the argument. \Qed{\Cref{clm:weighted-cover-gk}} 
\end{proof}
\noindent
By plugging the bounds of \Cref{clm: weighted-gs-bound-GK} and \Cref{clm:weighted-cover-gk} in \Cref{eq:weighted-lem-type-2-GK} and using the fact that a $\G$ value of an element is at least its $\g$ value, we have that,
\[
	\card{X_{\alpha}} \cdot (2^{\alpha-1}-2)\leq \sum_{\ei{i} \in X_{\alpha}} \Gs_i \hspace{5pt}+ \hspace{-10pt}\sum_{\ei{j} \in \WQS \cap \Band_{\leq \alpha}} \hspace{-15pt} g_j \hspace{3pt} \leq\hspace{5pt} 2 \hspace{-15pt}\sum_{\ei{j} \in \WQS \cap \Band_{\leq \alpha}} \hspace{-15pt} \G_j \hspace{5pt}+ \hspace{-10pt}\sum_{\ei{j} \in \WQS \cap \Band_{\leq \alpha}} \hspace{-15pt} \G_j \hspace{3pt} \leq  3\cdot O(\ell \cdot 2^{\alpha+1}),
\]
which implies $\card{X_{\alpha}} = O(\ell)$ for $ 3 \leq \alpha \leq \Bt{k} $. There can  be $O(\ell)$ elements each in $X_0$, $X_1$ and $X_2$ since there are at most $O(\ell)$ elements in $\Band_{\leq 2}$. By \Cref{eq: weighted band-alpha}  we have $ \Bt{k} = O(\log \tim{k})$ and therefore that the number of \Tt elements is $O(\ell \cdot \log \tim{k})$.
 \end{proof}
 We have now shown that, after performing the deletion step after $k$ elements have been seen, the number of \To elements in \WQS is $O(\ell \cdot \log \tim{k})$ by \Cref{lem: type-1 weighted-GK} and the number of \Tt elements in \WQS is $O(\ell \cdot \log \tim{k})$ by \Cref{lem:weighted-type-2-GK}. Since each element in \WQS (other than $+\infty$) is either \To or \Tt, the total number of elements in \WQS is $O(\ell \cdot \log \tim{k})$.


This finalizes the proof of \Cref{lem:GK-space-weighted} since $\tim{n}=O(\eps \wtsum{n})$ and $\ell=O(\frac{1}{\eps})$. We conclude the discussion of the space complexity with the following remark; 
\begin{remark} [Delaying Deletions]
\label{rem: weighted-Delaying-Deletions-GK}
Suppose in \Cref{alg:weighted-GK}, instead of running the \textbf{deletion step} in Line (ii) after each element, we  run it only after inserting $c$ elements $c > 1$; then, the space complexity of the algorithm only increases by an additive term $O(c)$.
\end{remark}

Performing the deletion step after $k$ elements, the number of elements reduces to $O(\ell \cdot \log \tim{k})$ as proved earlier as long as we have been satisfying both the conditions of the deletions of \Cref{alg:weighted-GK} while performing every deletion. Thus, the extra space is only due to storing the additional $O(c)$ elements that are inserted in \WQS.

The above remark will be useful in proposing an implementation of \Cref{alg:weighted-GK} which 
has an asymptotically faster update time per element, which we show in the following.

 \subsection{An Efficient Implementation of \texorpdfstring{\Cref{alg:weighted-GK}}{Algorithm} } \label{subsec: fast-imp-weighted-GK}
In this section, we present an efficient implementation of \Cref{alg:weighted-GK}. This is similar to the implementation of the GK summary proposed in \cite{LuoWYC16}. The key idea is that the deletion
step is slow and therefore performing it after every time step is rather time inefficient.  However, not performing the deletion step for too long blows up the space. The fast implementation we present deals with this trade-off and chooses the delay between consecutive deletion steps so that both the time and space complexity of the algorithm are optimized. Formally, we show the following:
\begin{lemma}\label{lem:weighted-gk-time}
There is an implementation of \Cref{alg:weighted-GK} that takes $O\left(\log(1/\eps)+ \log \log (\eps \wtsum{n}) + \frac{\log^2(\eps \wtsum{n})}{\eps n} \right)$ worst case processing time per element.
\end{lemma}
\paragraph*{\textbf{Part \RNum{1}:  Storing \QS:}}
We store our summary $\WQS$ as a balanced binary search tree (BST), where each node contains an element of $\WQS$ along with its metadata. For each element $e$ we  store  $\wt(e), g(e), \Delta(e)$ and $t_0(e)$. The sorting key of the BST is the value of elements.  The \textbf{Insert} and \textbf{Delete} operations  insert elements into and delete elements from the BST respectively. 
\paragraph*{\textbf{Part \RNum{2}: Performing a Deletion Step:}} The deletion step involves the 
deletion of elements in the summary that satisfy the two conditions of \Cref{alg:weighted-GK}. 
Checking condition (ii) requires that we know the $\Gs$ values corresponding to each element of the 
summary. 
We first show how the $\Gs$ values of all elements can be computed. 

\paragraph*{Computing \texorpdfstring{$\Gs$}{G} values: } \label{subsec:comp_gs_val}
First we perform an inorder traversal of $\WQS$ and store the elements $\ei{i}$ in sorted order as a 
temporary linked list. The $\Gs$ value computation will use a stack and will make one pass over the 
list from the smallest to the largest element. We describe the computation when the traversal 
reaches the element $\ei{i}$ in the list. To obtain $\Gs_i$, we sum up the $\Gs$ values of all 
elements on the top of the stack with $\bv$ less than $\bv(\ei{i})$ and add the sum to $G_i$. All 
these elements are popped from the stack and then $\ei{i}$ along with its computed $\Gs$ value is 
pushed onto the stack. We claim that at this point $\Gs_i$ has been correctly computed. Since each 
element is pushed and popped from the stack at most once, the $\Gs$ values of all elements can be 
computed in time linear in the size of $\WQS$.

We now describe how each deletion step is performed.
\begin{tbox}
\label{tbox:deletion-step-weighted-gk}
\textbf{Algorithm. Performing a deletion step efficiently:}
\medskip
\begin{enumerate}
    \item Perform an inorder traversal of $\WQS$ (which is a BST) to obtain a temporary (doubly-linked) list of elements sorted by value. 
    \item Compute $\bv$s of all elements of \WQS using \Cref{eq:closed-form-weighted}. 
    \item Compute the $\Gs$ value of all elements using the algorithm described above.
    \item Traverse the list from larger elements to smaller ones. For each element $\ei{i}$, delete it from BST (as well as the list), if it satisfies both the deletion conditions mentioned in \Cref{alg:weighted-GK}. 
\end{enumerate}
\end{tbox}

Below is an implementation of \Cref{alg:weighted-GK} with fast amortized update time. We also describe how to modify this implementation to also get the same bound on the worst case update time.

\begin{Implementation}\label{implementation: GK-weighted}
\textbf{Efficient Implementation of}~\Cref{alg:weighted-GK}
    \medskip
    \begin{itemize}
	\item Initialize $\WQS$ to be an empty balanced binary search tree.
	
	\item $\DeleteTime \leftarrow 2$.

	\item For each arriving item $(x_\emm, w(x_\emm))$:  
	\begin{enumerate}[label=$(\roman*)$]
		\item 
		Run $\textbf{Insert}(x_\emm,w(x_\emm))$. 
	     \item If $(k = \DeleteTime)$:	\begin{itemize}
		    \item 	Execute the deletion step and update $\DeleteTime \leftarrow \DeleteTime +  \ceil{\ell \log t_k}$.
		\end{itemize}
	
	\end{enumerate}
	\end{itemize}
\end{Implementation}

\paragraph*{Space Analysis.} 
The space complexity of the above implementation is still $O(\frac{1}{\eps} \log(\eps \wtsum{n}))$. This follows from that fact that, after performing a deletion when $k$ elements have been seen, we wait for another $O(\ell \cdot \log \tim{k})$ elements only, which increases the space complexity by only a constant factor due to \Cref{rem: weighted-Delaying-Deletions-GK}. Thus, the space complexity , after $n$ insertions, remains $O(\ell \cdot \log \tim{n})=O(\frac{1}{\eps} \log(\eps \wtsum{n}))$, as $\ell=1/\eps$ and $\tim{n}=O(\eps \wtsum{n})$.

 \paragraph*{Time Analysis.}
The main purpose storing behind \WQS as a BST was to decrease the time require to perform an $\textbf{Insert}$ and $\textbf{Delete}$ operation on \WQS. This takes only $O(\log s)$, where $s$ is the summary size which is at most $O(\frac{1}{\eps} \log \eps \wtsum{n})$. Thus, we now have the following observations, which is directly implied by the fact that we perform $\textbf{Insert}$ and $\textbf{Delete}$ at most once per element. 
\begin{observation} \label{obs: ins del time GK weighted}
Over a stream of length $n$, the total time taken by the fast implementation of \Cref{implementation: GK-weighted} to perform all \textbf{Insert} and \textbf{Delete} operations is $O( n \cdot (\log (1/\eps) + \log \log (\eps \wtsum{n})))$. 
\end{observation}

Note that the only time taken by \Cref{implementation: GK-weighted} \textbf{not} taken into account in~\Cref{obs: ins del time GK weighted} is the part that determines \emph{which elements} to delete, which we bound in the following.

\begin{lemma}\label{lem: time to decide elements to delete-greedy-weighted}
Over a stream of length $n$, the total time taken by  \Cref{implementation: GK-weighted}  to decide which elements need to be deleted over all the executed deletion steps is $O(n + \frac{1}{\eps} \log^2(\eps \wtsum{n}) )$. 
\end{lemma}
\begin{proof}
The time taken to decide which elements need to be deleted inside one deletion step (when $k$ elements have been seen) step is $O(s)=O(\ell \cdot \log \tim{k})$. This is because creating a linked list, followed by computation of $\bv$ and $\Gs$-value of all elements can be performed in $O(s)$ time. Finally, making a linear pass over the list from the largest to the smallest element (to check if the deletion conditions hold) requires $O(s)$ time.

Next, we obtain a bound on the  number of deletion steps performed by the algorithm. Consider the deletion steps performed when $\tim{k}$ is the intervals $[2^i, 2^{i+1})$, for $1\leq i\leq \ceil{\log (\eps \wtsum{n})}$. Let $d(i)$ be the number of such deletion steps and $n(i)$ denote the number of elements $x_k$ of the stream for which $\tim{k}$ is in the range $[2^i,2^{i+1})$. After the deletion step when $k$ elements have been seen, we wait for $ \ceil{\ell \log \tim{k} }$ insertions. Therefore, there are at least $\ell \cdot i$ elements inserted between two consecutive deletion steps that happen in the considered interval. Therefore, we get the following bound on the number of deletion steps that are performed during the interval. 
 \begin{align}\label{eq:delstep-bound-weighted}
 d(i) \leq \frac{n(i)}{\ell \cdot i} + 1. 
 \end{align}
 The time spent deciding which element to delete in a deletion step (after seeing $k$ elements) is at most $O(\ell \log \tim{k}) = O(\ell \cdot i)$, when $t_k$ is in the interval $[2^i,2^{i+1})$.
 This and \Cref{eq:delstep-bound-weighted}, give the following bound on the total time spent to decide which elements to delete over all deletions steps.

 \begin{align*}
     O\left(\sum\limits_{i=1}^{\ceil{\log(\eps \wtsum{n})}} d(i) \cdot \ell i\right) &= O\left(\sum\limits_{i=1}^{\ceil{\log(\eps \wtsum{n})}} \left(n(i) + \ell i   \right) \right)\\
     &= O\left(n +  \frac{1}{\eps} \log^2(\eps \wtsum{n})  \right)
 \end{align*}
 
This finalizes the proof of the lemma.
\end{proof}

\Cref{obs: ins del time GK weighted} and \Cref{lem: time to decide elements to delete-greedy} together clearly imply that the total time taken by \Cref{implementation: GK-weighted} over a stream of length $n$ is $O\big(\; n \cdot(\log(1/\eps)+\log\log(\eps \wtsum{n}))+\frac{1}{\eps} \log^2(\eps \wtsum{n})\;\big)$. Thus, the amortized update time per element is $O\big(\log(1/\eps)+ \log \log (\eps \wtsum{n}) + \frac{\log^2(\eps \wtsum{n})}{\eps n} \big)$.

We can obtain the same bound on the worst-case update time per element using standard ideas of distributing time of inefficient operations over multiple time steps.
The idea is to process the deletion step over all the following time steps before executing the next 
deletion step. Formally, we have the following:

\begin{claim} \label{lem: worst case wt GK}
There is an implementation of \Cref{alg:weighted-GK} with a worst-case update time of 
$O\big(\log(1/\eps)+ \log \log (\eps \wtsum{n}) + \frac{\log^2(\eps \wtsum{n})}{\eps n} \big)$.
\end{claim}

\subsection*{Acknowledgements}
We would like to thank Rajiv Gandhi for making the collaboration between the authors possible and 
for his support throughout this project.

\bibliographystyle{alpha}
\bibliography{new}

\newcommand{\etalchar}[1]{$^{#1}$}
\begin{thebibliography}{ABED{\etalchar{+}}21}

\bibitem[ABED{\etalchar{+}}21]{alon2021adversarial}
Noga Alon, Omri Ben-Eliezer, Yuval Dagan, Shay Moran, Moni Naor, and Eylon
  Yogev.
\newblock Adversarial laws of large numbers and optimal regret in online
  classification.
\newblock In {\em Proceedings of the 53rd Annual ACM SIGACT Symposium on Theory
  of Computing}, pages 447--455, 2021.

\bibitem[ACH{\etalchar{+}}12]{AgarwalCHPWY12}
Pankaj~K. Agarwal, Graham Cormode, Zengfeng Huang, Jeff~M. Phillips, Zhewei
  Wei, and Ke~Yi.
\newblock Mergeable summaries.
\newblock In {\em Proceedings of the 31st {ACM} {SIGMOD-SIGACT-SIGART}
  Symposium on Principles of Database Systems, {PODS} 2012, Scottsdale, AZ,
  USA, May 20-24, 2012}, 2012.

\bibitem[AMS96]{AlonMS96}
Noga Alon, Yossi Matias, and Mario Szegedy.
\newblock The space complexity of approximating the frequency moments.
\newblock In {\em Proceedings of the Twenty-Eighth Annual {ACM} Symposium on
  the Theory of Computing, Philadelphia, Pennsylvania, USA, May 22-24, 1996},
  pages 20--29, 1996.

\bibitem[BLRV13]{2013munro}
Andrej Brodnik, Alejandro L{\'{o}}pez{-}Ortiz, Venkatesh Raman, and Alfredo
  Viola, editors.
\newblock {\em Space-Efficient Data Structures, Streams, and Algorithms -
  Papers in Honor of J. Ian Munro on the Occasion of His 66th Birthday}, volume
  8066 of {\em Lecture Notes in Computer Science}. Springer, 2013.

\bibitem[CJP08]{ChakrabartiJP08}
Amit Chakrabarti, T.~S. Jayram, and Mihai Patrascu.
\newblock Tight lower bounds for selection in randomly ordered streams.
\newblock In {\em Proceedings of the Nineteenth Annual {ACM-SIAM} Symposium on
  Discrete Algorithms, {SODA} 2008, San Francisco, California, USA, January
  20-22, 2008}, pages 720--729, 2008.

\bibitem[CKMS06]{CormodeKMS06}
Graham Cormode, Flip Korn, S.~Muthukrishnan, and Divesh Srivastava.
\newblock Space- and time-efficient deterministic algorithms for biased
  quantiles over data streams.
\newblock In {\em Proceedings of the Twenty-Fifth {ACM} {SIGACT-SIGMOD-SIGART}
  Symposium on Principles of Database Systems, June 26-28, 2006, Chicago,
  Illinois, {USA}}, pages 263--272, 2006.

\bibitem[CM04]{CormodeM04}
Graham Cormode and S.~Muthukrishnan.
\newblock An improved data stream summary: The count-min sketch and its
  applications.
\newblock In {\em {LATIN} 2004: Theoretical Informatics, 6th Latin American
  Symposium, Buenos Aires, Argentina, April 5-8, 2004, Proceedings}, pages
  29--38, 2004.

\bibitem[CV20]{CormodeV20}
Graham Cormode and Pavel Vesel{\'{y}}.
\newblock A tight lower bound for comparison-based quantile summaries.
\newblock In {\em Proceedings of the 39th {ACM} {SIGMOD-SIGACT-SIGAI} Symposium
  on Principles of Database Systems, {PODS} 2020, Portland, OR, USA, June
  14-19, 2020}, pages 81--93, 2020.

\bibitem[FO15]{FelberO15}
David Felber and Rafail Ostrovsky.
\newblock A randomized online quantile summary in {$O$}{$(1/\epsilon \cdot
  \log(1/\epsilon))$} words.
\newblock In {\em Approximation, Randomization, and Combinatorial Optimization.
  Algorithms and Techniques, {APPROX/RANDOM} 2015, August 24-26, 2015,
  Princeton, NJ, {USA}}, pages 775--785, 2015.

\bibitem[GHR{\etalchar{+}}12]{gilbert2012recovering}
Anna~C Gilbert, Brett Hemenway, Atri Rudra, Martin~J Strauss, and Mary
  Wootters.
\newblock Recovering simple signals.
\newblock In {\em 2012 Information Theory and Applications Workshop}, pages
  382--391. IEEE, 2012.

\bibitem[GHS{\etalchar{+}}12]{gilbert2012reusable}
Anna~C Gilbert, Brett Hemenway, Martin~J Strauss, David~P Woodruff, and Mary
  Wootters.
\newblock Reusable low-error compressive sampling schemes through privacy.
\newblock In {\em 2012 IEEE Statistical Signal Processing Workshop (SSP)},
  pages 536--539. IEEE, 2012.

\bibitem[GK01]{GreenwaldK01}
Michael Greenwald and Sanjeev Khanna.
\newblock Space-efficient online computation of quantile summaries.
\newblock In {\em Proceedings of the 2001 {ACM} {SIGMOD} international
  conference on Management of data, Santa Barbara, CA, USA, May 21-24, 2001},
  pages 58--66, 2001.

\bibitem[GM09]{GuhaM09}
Sudipto Guha and Andrew McGregor.
\newblock Stream order and order statistics: Quantile estimation in
  random-order streams.
\newblock {\em {SIAM} J. Comput.}, 38(5):2044--2059, 2009.

\bibitem[GZ03]{GuptaZ03}
Anupam Gupta and Francis Zane.
\newblock Counting inversions in lists.
\newblock In {\em Proceedings of the Fourteenth Annual {ACM-SIAM} Symposium on
  Discrete Algorithms, January 12-14, 2003, Baltimore, Maryland, {USA}}, pages
  253--254, 2003.

\bibitem[HT10]{HungT10}
Regant Y.~S. Hung and Hing{-}Fung Ting.
\newblock An {$\Omega$} {$(\frac{1}{\varepsilon} \log \frac{1}{\varepsilon})$}
  space lower bound for finding {$\varepsilon$}-approximate quantiles in a data
  stream.
\newblock In {\em Frontiers in Algorithmics, 4th International Workshop, {FAW}
  2010, Wuhan, China, August 11-13, 2010. Proceedings}, pages 89--100, 2010.

\bibitem[HW13]{hardt2013robust}
Moritz Hardt and David~P Woodruff.
\newblock How robust are linear sketches to adaptive inputs?
\newblock In {\em Proceedings of the forty-fifth annual ACM symposium on Theory
  of computing}, pages 121--130, 2013.

\bibitem[ILL{\etalchar{+}}19]{ivkin2019streaming}
Nikita Ivkin, Edo Liberty, Kevin Lang, Zohar Karnin, and Vladimir Braverman.
\newblock Streaming quantiles algorithms with small space and update time.
\newblock {\em arXiv preprint arXiv:1907.00236}, 2019.

\bibitem[KLL16]{KarninLL16}
Zohar~S. Karnin, Kevin~J. Lang, and Edo Liberty.
\newblock Optimal quantile approximation in streams.
\newblock In {\em {IEEE} 57th Annual Symposium on Foundations of Computer
  Science, {FOCS} 2016, 9-11 October 2016, Hyatt Regency, New Brunswick, New
  Jersey, {USA}}, pages 71--78, 2016.

\bibitem[LWYC16]{LuoWYC16}
Ge~Luo, Lu~Wang, Ke~Yi, and Graham Cormode.
\newblock Quantiles over data streams: experimental comparisons, new analyses,
  and further improvements.
\newblock {\em {VLDB} J.}, 25(4):449--472, 2016.

\bibitem[MNS11]{mironov2011sketching}
Ilya Mironov, Moni Naor, and Gil Segev.
\newblock Sketching in adversarial environments.
\newblock {\em SIAM Journal on Computing}, 40(6):1845--1870, 2011.

\bibitem[MP78]{MunroP78}
J.~Ian Munro and Mike Paterson.
\newblock Selection and sorting with limited storage.
\newblock In {\em 19th Annual Symposium on Foundations of Computer Science, Ann
  Arbor, Michigan, USA, 16-18 October 1978}, pages 253--258, 1978.

\bibitem[MRL98]{RajagopalanML98}
Gurmeet~Singh Manku, Sridhar Rajagopalan, and Bruce~G. Lindsay.
\newblock Approximate medians and other quantiles in one pass and with limited
  memory.
\newblock In {\em {SIGMOD} 1998, Proceedings {ACM} {SIGMOD} International
  Conference on Management of Data, June 2-4, 1998, Seattle, Washington,
  {USA}}, pages 426--435, 1998.

\bibitem[MRL99]{MankuRL99}
Gurmeet~Singh Manku, Sridhar Rajagopalan, and Bruce~G. Lindsay.
\newblock Random sampling techniques for space efficient online computation of
  order statistics of large datasets.
\newblock In {\em {SIGMOD} 1999, Proceedings {ACM} {SIGMOD} International
  Conference on Management of Data, June 1-3, 1999, Philadelphia, Pennsylvania,
  {USA}}, pages 251--262, 1999.

\bibitem[NY15]{naor2015bloom}
Moni Naor and Eylon Yogev.
\newblock Bloom filters in adversarial environments.
\newblock In {\em Annual Cryptology Conference}, pages 565--584. Springer,
  2015.

\bibitem[SBAS04]{ShrivastavaBAS04}
Nisheeth Shrivastava, Chiranjeeb Buragohain, Divyakant Agrawal, and Subhash
  Suri.
\newblock Medians and beyond: new aggregation techniques for sensor networks.
\newblock In {\em Proceedings of the 2nd International Conference on Embedded
  Networked Sensor Systems, SenSys 2004, Baltimore, MD, USA, November 3-5,
  2004}, pages 239--249, 2004.

\bibitem[wik]{sub2}
List of open problems in sublinear algorithms -- problem 2: Quantiles.
\newblock \url{https://sublinear.info/2}.

\end{thebibliography}

\bigskip

\appendix

\section{A Greedy algorithm for weighted streams}\label{sec:wt-greedy}
	
In this section, present a \emph{time-efficient} algorithm that greedily deletes elements from the summary \WQS while maintaining ~\Cref{inv: g-delta-weighted} and the same condition on $\bv$ of elements as ~\Cref{alg:weighted-GK}. The idea is very similar; we first insert each arriving element in \WQS, and then we execute the deletion step to reduce the size of the summary. The only difference is that we do not delete an element along with its segment; instead we delete an element greedily even if it is not possible to delete its entire segment. Formally,
\begin{Algorithm}\label{alg:weighted-simple}
A greedy algorithm for weighted streams: 
	\medskip
	
	\medskip
For each arriving item $(x_\emm, w(x_\emm))$: 
\begin{enumerate}[label=$(\roman*)$]
\addtolength{\itemindent}{3mm}
	\item Run $\textbf{Insert}(x_\emm,w(x_\emm))$. 
	\item Repeatedly run $\textbf{Delete}(\ei{i})$ for any (arbitrarily chosen) element $\ei{i}$ in \WQS satisfying:
	\[
	(1)~\bv(\ei{i}) \leq \bv(\ei{i+1})\qquad \text{\underline{and}} \qquad (2)~\G_i+\g_{i+1}+\Delta_{i+1} \leq \tim{\emm}\]

\end{enumerate}
\end{Algorithm}

	We state the space and time complexity of \Cref{alg:weighted-simple} in the following theorem.
	\begin{theorem}\label{thm: simple-weighted}
    For any $\eps > 0$ and a weighted stream of length $n$ with total weight $\wtsum{n}$,~\Cref{alg:weighted-simple} maintains an $\eps$-approximate quantile summary in $O(\frac{1}{\eps} \cdot \log^2{\!(\eps \wtsum{n})})$ space. Also, there is an implementation of ~\Cref{alg:weighted-simple} that takes $O\left(\log(1/\eps) + \log \log (\eps \wtsum{n}) + \frac{\log^3(\eps \wtsum{n})}{\eps n}\right)$ update time per element.
\end{theorem}

We remark here that under the assumptions that $\eps \geq \frac{1}{\sqrt{n}}$ and when the weights 
are $\poly(n)$ bounded, the space complexity and the update time of the algorithm will respectively 
be $O((1/\eps) \log^2(\eps n))$ and $O(\log(1/\eps) + \log \log (\eps n))$. 

\Cref{alg:weighted-simple} deletes an element $\ei{i}$ only if $\G_i+\g_{i+1}+\Delta_{i+1} \leq \tim{k}$. By~\Cref{obs: g-delta-update-weighted}, one can say that $\g_{i+1}+\Delta_{i+1} \leq \tim{k}$ after the deletion, establishing \Cref{inv: g-delta-weighted}. The other $g$ and $\Delta$ values remain unchanged). As shown in \Cref{sec: weighted g delta}, maintaining \Cref{inv: g-delta-weighted} is enough for \WQS to be a valid $\eps$- approximate quantile summary. In what follows, we prove the guarantees on the space complexity of \Cref{alg:weighted-simple} stated in \Cref{thm: simple-weighted}.

\subsection{The space analysis}\label{sec: space-analysis-weighted}
We borrow the notion of \To and \Tt elements from \Cref{subsec: space analysis wt gk} with slight modification. After performing the deletion step (when $\emm$ elements of the stream have been seen), each element in \WQS satisfies $\bv(\ei{i})>\bv(\ei{i+1})$ or $\G_i+\g_{i+1}+\Delta_{i+1}>\tim{\emm}$; otherwise \Cref{alg:weighted-simple} would have deleted it. The elements that satisfy the former condition are identified as \To elements, and the rest (that only satisfy the latter condition) are identified as \Tt elements. Again, as previously, we will bound the number of \To and \Tt elements separately. 
	\begin{lemma}\label{lem: basic type-1 weighted}
	 After the deletion step when $\emm$ elements have been seen, the number of \To elements stored in \WQS is at most $O\big(\log \tim{\emm} \big)$ times the number of \Tt elements.
	\end{lemma}
	\begin{proof} Consider a maximal sequence of consecutive \To elements $\ei{i},\ei{i+1},\ldots,\ei{j}$ in \WQS, and hence, $\bv(\ei{i}) > \bv(\ei{i+1}) > \ldots > \bv(\ei{j})$. The number of band values is $\Bt{\emm} = O(\log \tim{\emm})$ by \Cref{eq: weighted band-alpha}. Therefore, the length of this sequence can only be $O(\log{\tim{\emm}})$. Thus, we cannot have have only $O(\log{\tim{\emm}})$ \To elements for every \Tt element, concluding the proof.
	\end{proof}
	We now focus on bounding the number of \Tt elements in the following lemma.
	\begin{lemma}\label{lem: basic type-2 weighted}
	 After the deletion step when $\emm$ elements have been seen, the number of \Tt elements stored in \WQS is $O(\ell \log \tim{\emm})$.
	\end{lemma}
	\begin{proof}
	We first show a lower bound on $\G_{i} + g_{i+1}$ for each \Tt element $\ei{i}$ in \WQS in the following. 
	\begin{claim}\label{clm:weighted lower-g}
	For any \Tt element $\ei{i}$,
		$ \G_{i+1}+g_i \geq 2^{\bv(\ei{i+1})-1}-2$.
\end{claim}
\begin{proof}
	Since $e_i$ is a  \Tt element we have ~$\G_i + \g_{i+1} + \Delta_{i+1} > \tim{\emm}$. Furthermore, we know that $\Delta_{i+1} \leq t_0(\ei{i+1})$ by \Cref{eq:del-t0-weighted}. Therefore, 
	\begin{align*}
		\G_i+ \g_{i+1} > \tim{\emm} - t_0(\ei{i+1}) \geq 2^{\bv(\ei{i+1})-1}-2,
	\end{align*}
	where the second inequality is by \Cref{eq:closed-form-weighted}. \Qed{\Cref{clm:weighted lower-g}} 
	
	\end{proof} 
	
	We partition the \Tt elements into different sets. Formally, we consider sets $X_0,\ldots,X_{\Bt{k}}$ such that for any band-value $\alpha$:
\[
X_{\alpha} := \set{\ei{i} \in \WQS \mid \text{$\ei{i}$ is \Tt and $\bv(\ei{i+1}) = \alpha$}}.
\] 
Note that each element in $X_\alpha$ by virtue of being \Tt is in $\Band_{\leq \alpha}$. Summing over the inequality of~\Cref{clm:weighted lower-g} for each element in $X_\alpha$, we obtain:
	\begin{align} \label{eq:weighted-lem-type-2}
		\card{X_{\alpha}} \cdot (2^{\alpha-1}-2) &\leq \sum_{\ei{i} \in X_{\alpha}} \G_i + g_{i+1} \leq 2 \sum_{\ei{i} \in \WQS \cap \Band{\leq \alpha}} \G_i,
	\end{align}
where the second inequality follows form the fact that $\ei{i}$ and $\ei{i+1}$ belong to $\Band_{\leq \alpha}$ and that $G_{i+1}$ is an upper bound on $g_{i+1}$.  The final step of the proof is to show a bound on the right hand side of  \Cref{eq:weighted-lem-type-2}.
Formally, we show the following claim.

\begin{claim}\label{clm:weighted-cover}
	After seeing $k$ elements of the stream, for any $\alpha$, 
	$
	\sum_{\ei{i} \in \WQS \cap \Band_{\leq \alpha}} \G_i \leq O(\ell \cdot 2^{\alpha+1}).
	$ 
\end{claim}

This claim is identical to \Cref{clm:weighted-cover-gk}. The proof only requires that a deletion must satisfy the condition $(i)$ (the condition on $\bv$s). We do have the same condition on $\bv$s in the deletion step of \Cref{alg:weighted-simple}. Thus, the proof is identical to that of \Cref{clm:weighted-cover-gk}. By plugging the bounds in~\Cref{clm:weighted-cover} into~\Cref{eq:weighted-lem-type-2}, we obtain, 

\[
	\card{X_{\alpha}} \cdot (2^{\alpha-1}-2) \leq 2 \cdot  O(\ell \cdot 2^{\alpha+1}).
\]
Hence, $\card{X_{\alpha}} = O(\ell)$ when $\alpha>2$ (and for $\alpha \leq 2$ the total number of elements are at most $8 \cdot \ell = O(\ell)$ by \Cref{eq: weighted band-alpha} anyway). Since we have only $O(\log{\tim{\emm}})$ possible values of bands by~\Cref{eq: weighted band-alpha}, we have at most $\log{\tim{\emm}})$ sets $X_{\alpha}$. Thus, we conclude that \WQS has at most $O(\ell \cdot \log \tim{\emm})$ \Tt elements.
\end{proof}
\Cref{lem: basic type-1 weighted,lem: basic type-2 weighted} immediately imply that the total number of elements stored in \WQS overall, after $n$ insertions, is $O(\ell \log^2 \tim{n})=O(\frac{1}{\eps} \log^2(\eps \wtsum{n}))$, concluding the space complexity.

We start the discussion by the following remark on the space used by \Cref{alg:weighted-simple}.
\begin{remark}
\label{rem:weighted-greedy}
\emph{
In the space analysis, we bounded the number of \Tt elements in the summary after the deletion step by $O(\ell \cdot \log{\tim{n}}) = O((1/\eps)\cdot\log{\!(\eps \wtsum{n})})$, which is quite efficient on its own. However, in the worst case, there can be  $O(\log \tim{n})$ \To elements for every \Tt element as shown in \Cref{fig:tightness-weighted}. Thus, \Cref{alg:weighted-simple} may end up storing as many as $O(\ell \cdot \log^2 \tim{n}) = O((1/\eps)\cdot\log^2{\!(\eps \wtsum{n})})$ \To elements in the summary, leading to its sub-optimal space requirement.   
}
\end{remark}

\begin{figure}[ht!]
\centering
\hspace{-1cm}\resizebox{1\textwidth}{!}{\tikzset{every picture/.style={line width=0.75pt}} 

\begin{tikzpicture}[x=0.75pt,y=0.75pt,yscale=-1,xscale=1]

\draw [line width=1.5]    (138.21,348.43) -- (653.11,348.43) (165.21,347.43) -- (165.21,349.43)(192.21,347.43) -- (192.21,349.43)(219.21,347.43) -- (219.21,349.43)(246.21,347.43) -- (246.21,349.43)(273.21,347.43) -- (273.21,349.43)(300.21,347.43) -- (300.21,349.43)(327.21,347.43) -- (327.21,349.43)(354.21,347.43) -- (354.21,349.43)(381.21,347.43) -- (381.21,349.43)(408.21,347.43) -- (408.21,349.43)(435.21,347.43) -- (435.21,349.43)(462.21,347.43) -- (462.21,349.43)(489.21,347.43) -- (489.21,349.43)(516.21,347.43) -- (516.21,349.43)(543.21,347.43) -- (543.21,349.43)(570.21,347.43) -- (570.21,349.43)(597.21,347.43) -- (597.21,349.43)(624.21,347.43) -- (624.21,349.43)(651.21,347.43) -- (651.21,349.43) ;
\draw [shift={(657.11,348.43)}, rotate = 180] [fill={rgb, 255:red, 0; green, 0; blue, 0 }  ][line width=0.08]  [draw opacity=0] (11.61,-5.58) -- (0,0) -- (11.61,5.58) -- cycle    ;
\draw [line width=1.5]    (138.21,199.24) -- (138.21,348.43) (139.21,223.24) -- (137.21,223.24)(139.21,247.24) -- (137.21,247.24)(139.21,271.24) -- (137.21,271.24)(139.21,295.24) -- (137.21,295.24)(139.21,319.24) -- (137.21,319.24)(139.21,343.24) -- (137.21,343.24) ;
\draw [shift={(138.21,195.24)}, rotate = 90] [fill={rgb, 255:red, 0; green, 0; blue, 0 }  ][line width=0.08]  [draw opacity=0] (11.61,-5.58) -- (0,0) -- (11.61,5.58) -- cycle    ;
\draw  [fill={rgb, 255:red, 232; green, 232; blue, 232 }  ,fill opacity=1 ] (146.52,217.68) -- (162.2,217.68) -- (162.2,226.64) -- (146.52,226.64) -- cycle ;
\draw  [fill={rgb, 255:red, 232; green, 232; blue, 232 }  ,fill opacity=1 ] (169.92,239.29) -- (185.61,239.29) -- (185.61,248.25) -- (169.92,248.25) -- cycle ;
\draw  [fill={rgb, 255:red, 232; green, 232; blue, 232 }  ,fill opacity=1 ] (219.43,281.6) -- (235.12,281.6) -- (235.12,290.56) -- (219.43,290.56) -- cycle ;
\draw  [fill={rgb, 255:red, 232; green, 232; blue, 232 }  ,fill opacity=1 ] (242.84,302.3) -- (258.52,302.3) -- (258.52,311.26) -- (242.84,311.26) -- cycle ;
\draw  [fill={rgb, 255:red, 232; green, 232; blue, 232 }  ,fill opacity=1 ] (292.35,216.78) -- (308.03,216.78) -- (308.03,225.74) -- (292.35,225.74) -- cycle ;
\draw  [fill={rgb, 255:red, 232; green, 232; blue, 232 }  ,fill opacity=1 ] (315.76,238.39) -- (331.44,238.39) -- (331.44,247.35) -- (315.76,247.35) -- cycle ;
\draw  [fill={rgb, 255:red, 232; green, 232; blue, 232 }  ,fill opacity=1 ] (365.27,280.7) -- (380.95,280.7) -- (380.95,289.66) -- (365.27,289.66) -- cycle ;
\draw  [fill={rgb, 255:red, 232; green, 232; blue, 232 }  ,fill opacity=1 ] (388.67,301.4) -- (404.36,301.4) -- (404.36,310.36) -- (388.67,310.36) -- cycle ;
\draw  [fill={rgb, 255:red, 232; green, 232; blue, 232 }  ,fill opacity=1 ] (476.89,216.78) -- (492.58,216.78) -- (492.58,225.74) -- (476.89,225.74) -- cycle ;
\draw  [fill={rgb, 255:red, 232; green, 232; blue, 232 }  ,fill opacity=1 ] (500.3,238.39) -- (515.98,238.39) -- (515.98,247.35) -- (500.3,247.35) -- cycle ;
\draw  [fill={rgb, 255:red, 232; green, 232; blue, 232 }  ,fill opacity=1 ] (549.81,280.7) -- (565.49,280.7) -- (565.49,289.66) -- (549.81,289.66) -- cycle ;
\draw  [fill={rgb, 255:red, 232; green, 232; blue, 232 }  ,fill opacity=1 ] (573.21,301.4) -- (588.9,301.4) -- (588.9,310.36) -- (573.21,310.36) -- cycle ;
\draw  [fill={rgb, 255:red, 232; green, 232; blue, 232 }  ,fill opacity=1 ] (145,376.12) -- (162.02,376.12) -- (162.02,386.03) -- (145,386.03) -- cycle ;
\draw  [color={rgb, 255:red, 92; green, 92; blue, 92 }  ,draw opacity=1 ][fill={rgb, 255:red, 130; green, 130; blue, 130 }  ,fill opacity=1 ] (268.94,324.81) -- (284.63,324.81) -- (284.63,333.77) -- (268.94,333.77) -- cycle ;
\draw  [color={rgb, 255:red, 92; green, 92; blue, 92 }  ,draw opacity=1 ][fill={rgb, 255:red, 130; green, 130; blue, 130 }  ,fill opacity=1 ] (414.78,323.91) -- (430.46,323.91) -- (430.46,332.87) -- (414.78,332.87) -- cycle ;
\draw  [color={rgb, 255:red, 92; green, 92; blue, 92 }  ,draw opacity=1 ][fill={rgb, 255:red, 130; green, 130; blue, 130 }  ,fill opacity=1 ] (599.32,323.91) -- (615,323.91) -- (615,332.87) -- (599.32,332.87) -- cycle ;
\draw  [color={rgb, 255:red, 92; green, 92; blue, 92 }  ,draw opacity=1 ][fill={rgb, 255:red, 130; green, 130; blue, 130 }  ,fill opacity=1 ] (226.41,377.11) -- (243.42,377.11) -- (243.42,387.02) -- (226.41,387.02) -- cycle ;
\draw   (100.91,219.61) .. controls (96.24,219.61) and (93.91,221.94) .. (93.91,226.61) -- (93.91,268.45) .. controls (93.91,275.12) and (91.58,278.45) .. (86.91,278.45) .. controls (91.58,278.45) and (93.91,281.78) .. (93.91,288.45)(93.91,285.45) -- (93.91,330.3) .. controls (93.91,334.97) and (96.24,337.3) .. (100.91,337.3) ;
\draw    (284.53,341.3) .. controls (313.51,358.24) and (336.51,378.24) .. (402.51,384.24) ;
\draw    (504.51,391.24) .. controls (555.51,386.24) and (584.51,358.24) .. (603.51,336.24) ;
\draw    (419.51,336.24) .. controls (422.51,367.24) and (432.51,375.24) .. (443.51,382.24) ;
\draw  [dash pattern={on 0.84pt off 2.51pt}]  (483.53,387.9) .. controls (492.07,389.91) and (547.51,368.24) .. (560.51,334.24) ;
\draw  [dash pattern={on 0.84pt off 2.51pt}]  (468.51,382.24) .. controls (490.51,380.24) and (502.51,362.24) .. (509.51,334.24) ;

\draw (422.34,255.06) node [anchor=north west][inner sep=0.75pt]  [font=\huge,xscale=1.3,yscale=1.3] [align=left] {$\displaystyle \dotsc $};
\draw (192.87,247.93) node [anchor=north west][inner sep=0.75pt]  [font=\large,rotate=-45,xscale=1.3,yscale=1.3] [align=left] {$\displaystyle \dotsc $};
\draw (341.7,246.03) node [anchor=north west][inner sep=0.75pt]  [font=\large,rotate=-45,xscale=1.3,yscale=1.3] [align=left] {$\displaystyle \dotsc $};
\draw (525.24,245.03) node [anchor=north west][inner sep=0.75pt]  [font=\large,rotate=-45,xscale=1.3,yscale=1.3] [align=left] {$\displaystyle \dotsc $};
\draw (614.83,364.8) node [anchor=north west][inner sep=0.75pt]  [font=\footnotesize,xscale=1.3,yscale=1.3] [align=left] {{\footnotesize Position}};
\draw (29.8,267.07) node [anchor=north west][inner sep=0.75pt]  [font=\footnotesize,xscale=1.3,yscale=1.3] [align=left] {{\footnotesize $\displaystyle O(\log t)$}\\{\footnotesize bands}};
\draw (384.51,395.24) node [anchor=north west][inner sep=0.75pt]  [font=\footnotesize,xscale=1.3,yscale=1.3] [align=left] {\begin{minipage}[lt]{77.67844000000001pt}\setlength\topsep{0pt}
	\begin{center}
	{\footnotesize $\displaystyle O(\log t)$ type-2 \ elements}
	\end{center}
	
	\end{minipage}};
\draw (165.56,374.16) node [anchor=north west][inner sep=0.75pt]  [font=\footnotesize,xscale=1.3,yscale=1.3] [align=left] {{\footnotesize type-1}};
\draw (247.94,374.16) node [anchor=north west][inner sep=0.75pt]  [font=\footnotesize,xscale=1.3,yscale=1.3] [align=left] {{\footnotesize type-2}};
\draw (40,189.7) node [anchor=north west][inner sep=0.75pt]  [font=\footnotesize,xscale=1.3,yscale=1.3] [align=left] {{\footnotesize Band-values}};

\end{tikzpicture}}
\caption{ Each block in the figure represents an element stored in \WQS. The ranks of elements increase along the horizontal axis. The figure illustrates why \Cref{alg:weighted-simple} might end up storing $O(\ell \log^2 \tim{n})$ elements in \WQS.  By \Cref{lem: basic type-2 weighted}, there could be as many as  $O(\ell \cdot \log \tim{n})$ \Tt elements in \WQS. Each of these \Tt elements could be preceded by a sequence of $O(\log \tim{n})$ \To elements (since there are $O(\log \tim{n})$ bands). }
\label{fig:tightness-weighted}
\end{figure}

As we say in~\Cref{rem:weighted-greedy}, one source of sub-optimality of~\Cref{alg:weighted-simple} was the large number of \To elements stored in the summary compared to the \Tt ones. A way to improve this is to \emph{actively} try to decrease  the number of stored \To elements. Roughly speaking, this is done by deleting \Tt elements from the summary only if it does 
not contribute to creating a long sequence of \To elements (e.g., as in~\Cref{fig:tightness-weighted}). Note that our \Cref{alg:weighted-GK} was precisely doing this by only deleting an element only if its entire segment can be deleted along with it.

We also note that even if we do not execute the deletion step after every insertion, the size of the summary still reduces to $O((1/\eps) \log^2 (\eps \wtsum{n}))$ after performing a deletion step by the above analysis. Therefore, we have the following remark.

\begin{remark} [Delaying Deletions]
\label{rem: weighted-Delaying-Deletions-simple}
Suppose in \Cref{alg:weighted-simple}, instead of running the \textbf{deletion step} in Line (ii) after each element, we  run it only after inserting $c$ elements $c > 1$; then, the space complexity of the algorithm only increases by an additive term $O(c)$.
\end{remark}
We now give a \emph{fast} implementation of \Cref{alg:weighted-simple} in the following.

\subsection{An Efficient Implementation of \texorpdfstring{\Cref{alg:weighted-simple}}{the algorithm} }\label{subsec: fast-imp-weighted-simple}

In this section, we present a faster implementation of \Cref{alg:weighted-simple} which is based on the (previously seen) idea of delaying the deletion steps.
\paragraph*{\textbf{Part \RNum{1}:  Storing \QS:}}
We store our summary $\WQS$ as a balanced binary search tree (BST), where each node contains an element of $\WQS$ along with its metadata. For each element $e$ we store $\wt(e), g(e), \Delta(e)$ and $t_0(e)$. The sorting key of the BST is the value of elements.  The \textbf{Insert} and \textbf{Delete} operations insert elements into and delete elements from the BST respectively. 

\paragraph*{\textbf{Part \RNum{2}: Performing a Deletion Step:}} We perform this step as follows:
\begin{tbox}
\label{tbox:deletion-step-weighted-greedy}

\begin{enumerate}
    \item Perform an inorder traversal of $\WQS$ (which is a BST) to obtain a temporary (doubly-linked) list of elements sorted by value. 
    \item 
     Compute $\bv$ of all elements of \WQS using \Cref{eq:closed-form-weighted}. 
    \item Traverse the list from larger elements to smaller ones. For each element $\ei{i}$, delete it from BST (as well as the list), if it satisfies both the deletion conditions mentioned in \Cref{alg:weighted-simple}. 
\end{enumerate}
\end{tbox}
 Note that after one pass of the list, no more deletions arise since we traverse the list from larger elements to smaller ones, and a deletion can only increase the $g$ value of the \emph{next} element.

Now, we describe an implementation of \Cref{alg:weighted-simple} with fast amortized update time. We then show how it can be extended so that it also has a fast worst-case update time.

\begin{Implementation}\label{implementation: greedy-weighted}
\textbf{Efficient Implementation of}~\Cref{alg:weighted-simple}
    \medskip
    \begin{itemize}
	\item Initialize $\WQS$ to be an empty balanced binary search tree.
	
	\item $\DeleteTime \leftarrow 2$.

	\item For each arriving item $(x_\emm, w(x_\emm))$:  
	\begin{enumerate}[label=$(\roman*)$]
		\item 
		Run $\textbf{Insert}(x_\emm,w(x_\emm))$. 
	     \item If $(k = \DeleteTime)$:	\begin{itemize}
		    \item 	Execute the deletion step and update $\DeleteTime \leftarrow \DeleteTime +  \ell \ceil{\log^2t_k}$.
		\end{itemize}
	
	\end{enumerate}
	\end{itemize}
\end{Implementation}

\paragraph*{Space Analysis.} 
The space complexity of the implementation remains $O(\frac{1}{\eps} \log^2(\eps \wtsum{n}))$. This is simply because after the deletion step when $k$ elements have been inserted, the next deletion is only performed after another $\ceil {\ell \log^2 t_k}$ insertions. This, by \Cref{rem: weighted-Delaying-Deletions-simple}, implies that the space used only increases by an additive term of $O(\ell \log ^2 t_k)=O(\frac{1}{\eps} \log^2(\eps \wtsum{n}))$, as $t_k = O(\eps \wtsum{n})$ and $\ell = O(1/\eps)$.
 
 \paragraph*{Time Analysis.}
Since \WQS is stored as a BST, performing an $\textbf{Insert}$ or $\textbf{Delete}$ operation on \WQS takes $O(\log s)$, where $s$ is the number of elements stored in \WQS. By the above space analysis , $s = O(\frac{1}{\eps} \log^2(\eps \wtsum{n}))$. This leads to the following observation:

\begin{observation} \label{obs: ins del time greedy weighted}
Over a stream of length $n$, the total time taken by the fast \Cref{implementation: greedy-weighted} to perform all \textbf{Insert} and \textbf{Delete} operations is $O( n \cdot (\log (1/\eps) + \log \log (\eps \wtsum{n})))$. 
\end{observation}

The observation follows from the fact that each element is inserted and deleted at most once from \WQS and each insertion or deletion takes $O(\log s) = O\big( \log (1/\eps)+ \log \log (\eps \wtsum{n}) \big)$ time. The only time taken by \Cref{implementation: greedy-weighted} \textbf{not} taken into account in~\Cref{obs: ins del time greedy weighted} is the part that determines \emph{which elements} to delete, which we will bound below.

\begin{lemma}\label{lem: time to decide elements to delete- wt greedy}
Over a stream of length $n$, the total time taken by  \Cref{implementation: greedy-weighted}  to decide which elements need to be deleted over all the executed deletion steps is $O(n + \frac{1}{\eps} \log^3(\eps \wtsum{n}) )$. 
\end{lemma}

\begin{proof}
We first divide the elements $x_k$ of $\swt$ into $\ceil{\log(\eps \wtsum{n})}$ different groups depending on the value of $t_k$: all elements $e_k$ for which $t_k \in [2^i, 2^{i+1})$ belong to the same group. Let $n(i)$ denote the number of elements of the stream $e_k$ for which $\tim{k}\in [2^i,2^{i+1})$. Similarly, we let $d(i)$ be the number of deletion steps performed when $\tim{k} \in [2^i, 2^{i+1})$. Between any two consecutive deletions, there are at least $\ell i^2$ elements of the stream inserted since we wait for $ \ell \ceil{ \log^2 t_k }$ insertions after performing the deletion step at time step $t_k$. Therefore, we get the following bound on the number of deletion steps in the interval $[2^i,2^{i+1})$:
 \begin{align}\label{eq:delstep-bound}
 d(i) \leq \frac{n(i)}{\ell \cdot i^2} + 1. 
 \end{align}
 An important thing to notice is that the time taken to decide which elements need to be deleted during a deletion step when $t_k \in [2^i, 2^{i+1})$ is proportional to the size of the summary: $O(s)=O(\ell \log^2 \tim{k})=O(\ell \cdot i^2)$. This is because: creating a linked list, followed by computation of $\bv$ of all elements can be performed in $O(s)$ time. Finally, the pass made over the list from the largest to the smallest element (to check if the deletion conditions hold) also requires $O(s)$ time.

 This and \Cref{eq:delstep-bound}, give the following bound on the total time spent to decide which elements to delete over all deletions steps.

 \begin{align*}
     O\left(\sum\limits_{i=1}^{\ceil{\log(\eps \wtsum{n})}} d(i) \cdot \ell i^2\right) &= O\left(\sum\limits_{i=1}^{\ceil{\log(\eps \wtsum{n})}} \left(n(i) + \ell i^2 \right) \right)\\
     &=O\left( n + \frac{1}{\eps} \log^3(\eps \wtsum{n})  \right)
 \end{align*}
 
This concludes the proof of the lemma.
\end{proof}
\Cref{obs: ins del time greedy weighted} and \Cref{lem: time to decide elements to delete- wt greedy} clearly imply that the total time taken by \Cref{implementation: greedy-weighted} over a stream of length $n$ is $O\big(\; n \cdot(\log(1/\eps)+\log\log(\eps \wtsum{n}))+\frac{1}{\eps} \log^3(\eps \wtsum{n})\;\big)$. Thus, the amortized update time per element is $O\big(\log(1/\eps)+\log\log(\eps n) \big)$ when $\wtsum{n} = \poly (n)$ and $\eps \geq \log^3(n)/n$.

\paragraph*{Worst-case update time.}We now look at the worst-case update time for any element. We might have to delete $O(s)$ elements from \QS in the worst case when a deletion step is performed which would take time $O(s \log s)$. To reduce the worst-case update time, we propose a minor modification to \Cref{implementation: greedy-weighted}. We notice that a deletion step is next called after $\ceil{\log^2 t}$ time steps. What we do to reduce the worst-case time is spread the time it takes to perform a deletion step over all elements before the next deletion step. Formally, we have the following:
\begin{claim} \label{lem: worst case wt greedy}
There is an implementation of \Cref{alg:weighted-simple} with $O(\log(1/\eps) + \log\log(\eps n))$ worst case update time per element when $W_n$ is $\poly(n)$ and $\eps \geq \log^3 (n)/n.$
\end{claim}

We use the same standard technique used in the proof of~\Cref{lem: worst case greedy}. The strategy is to spread the deletion step uniformly over all the following time steps before executing the next deletion step. 
The same argument works because the parameters are the same in terms of $s$. 
This increases the space only by a constant factor thus we conclude the proof of \Cref{thm: simple-weighted}.

\end{document}